\documentclass{article}

\usepackage{amssymb,amsmath,amsthm}
\usepackage{makeidx}         
\usepackage{graphicx}        

\usepackage{psfrag}
\usepackage{a4wide}

\theoremstyle{plain}
\newtheorem{thm}{Theorem}
\newtheorem{prop}[thm]{Proposition}
\newtheorem{lemma}[thm]{Lemma}
\newtheorem{cor}[thm]{Corollary}
\theoremstyle{definition}
\newtheorem{definition}[thm]{Definition}
\newtheorem{conjecture}[thm]{Conjecture}

\newtheorem{remark}[thm]{Remark}

\newtheorem{remarks}[thm]{Remarks}

\newtheorem{ind}[thm]{Inductive assumption}

\newcommand{\tn}[1]{\ensuremath{\mathbb{T}^{#1}}}
\newcommand{\rn}[1]{\ensuremath{\mathbb{R}^{#1}}}

\newcommand{\so}{\ensuremath{\mathbb{S}^{1}}}
\newcommand{\ro}{\ensuremath{\mathbb{R}}}
\newcommand{\no}{\ensuremath{\mathbb{N}}}
\newcommand{\zo}{\mathbb{Z}}
\newcommand{\tr}{\mathrm{tr}}
\newcommand{\Ein}{\mathrm{Ein}}
\newcommand{\tw}{\mathrm{tw}}
\newcommand{\pre}{\mathrm{pre}}
\newcommand{\rfr}{\mathrm{ref}}
\newcommand{\vel}{\mathrm{vel}}
\newcommand{\bas}{\mathrm{bas}}
\newcommand{\dS}{\mathrm{dS}}
\newcommand{\sca}{\mathrm{sc}}

\renewcommand{\a}{\alpha}
\renewcommand{\b}{\beta}
\newcommand{\g}{\gamma}
\newcommand{\G}{\Gamma}
\newcommand{\bga}{\bar{\gamma}}

\newcommand{\bn}{\overline{\nabla}}
\newcommand{\de}{\delta}
\renewcommand{\d}{\partial}
\newcommand{\dvt}{\d_{\vartheta}}
\newcommand{\me}{\mathcal{E}}
\newcommand{\mH}{\mathcal{H}}
\newcommand{\mQ}{\mathcal{Q}}
\newcommand{\mG}{\mathcal{G}}

\newcommand{\sfR}{\mathsf{R}}
\newcommand{\ma}{\mathcal{A}}

\newcommand{\hma}{\hat{\ma}}

\newcommand{\mP}{\mathcal{P}}

\newcommand{\irn}[1]{\int_{\rn{#1}}}
\newcommand{\is}{\int_{\so}}

\newcommand{\bx}{\bar{x}}
\newcommand{\bg}{\bar{g}}
\newcommand{\bk}{\bar{k}}
\newcommand{\ldr}[1]{\langle #1\rangle}

\newcommand{\hZ}{\hat{Z}}
\newcommand{\hPs}{\hat{\Psi}}

\newcommand{\hE}{\hat{E}}
\newcommand{\hF}{\hat{F}}

\newcommand{\hQ}{\hat{\mQ}}

\newcommand{\hl}{\hat{\lambda}}
\renewcommand{\l}{\lambda}
\newcommand{\e}{\epsilon}
\newcommand{\bsfp}{\bar{\mathsf{p}}}
\newcommand{\bsfx}{\bar{\mathsf{x}}}
\newcommand{\bsfy}{\bar{\mathsf{y}}}
\newcommand{\bsfz}{\bar{\mathsf{z}}}
\newcommand{\sfg}{\mathsf{g}}
\newcommand{\bsfg}{\bar{\sfg}}
\newcommand{\bxi}{\bar{\xi}}
\newcommand{\sfk}{\mathsf{k}}
\newcommand{\bsfk}{\bar{\sfk}}
\newcommand{\mff}{\mathfrak{f}}
\newcommand{\bmff}{\bar{\mathfrak{f}}}
\newcommand{\pros}[1]{\mathrm{pr}_{#1}}
\newcommand{\bp}{\bar{p}}
\newcommand{\bz}{\bar{z}}
\newcommand{\sffb}{\bar{\mathsf{f}}}
\newcommand{\roV}{\mathrm{Vl}}
\newcommand{\rosf}{\mathrm{sf}}
\newcommand{\krov}{K_{\roV}}
\newcommand{\fb}{\bar{f}}
\newcommand{\mup}[1]{\mu_{\mP_{#1}}}
\newcommand{\ip}[1]{\int_{\mP_{#1}}}
\newcommand{\ml}{\mathcal{L}}
\newcommand{\maco}[1]{\ensuremath{\mathsf{\Xi}_{\mathsf{#1}}}}
\newcommand{\bmaco}[1]{\ensuremath{\bar{\mathsf{\Xi}}_{\bar{\mathsf{#1}}}}}
\newcommand{\sfp}{\mathsf{p}}
\newcommand{\sfx}{\mathsf{x}}
\newcommand{\sff}{\mathsf{f}}
\newcommand{\dm}{\mathfrak{D}}
\newcommand{\lbpr}{\langle \bp\rangle}
\newcommand{\bsff}{\bar{\sff}}
\newcommand{\dhy}{\bar{\mathfrak{D}}}
\newcommand{\bchi}{\bar{\chi}}
\newcommand{\bvr}{\bar{\varrho}}

\newcommand{\current}{\bar{J}}
\newcommand{\br}{\bar{r}}
\newcommand{\bX}{\bar{X}}
\newcommand{\backg}{\mathrm{bg}}
\newcommand{\einj}{\e_{\mathrm{inj}}}
\newcommand{\Ric}{\mathrm{Ric}}
\newcommand{\stp}{r}
\newcommand{\bphi}{\bar{\phi}}
\newcommand{\rom}{\mathrm{m}}
\newcommand{\assh}{\bar{\mathfrak{v}}}
\newcommand{\sfh}{\mathsf{h}}
\newcommand{\bq}{\bar{q}}
\newcommand{\bsfh}{\bar{\sfh}}
\newcommand{\bmfq}{\bar{\mathfrak{q}}}
\newcommand{\sfhb}{\bar{\mathsf{h}}}

\begin{document}

\title{Proof of the cosmic no-hair conjecture in the 
$\tn{3}$-Gowdy symmetric Einstein-Vlasov setting}
\author{H\aa kan Andr\'{e}asson and Hans Ringstr\"{o}m}
\maketitle
\begin{abstract}
The currently preferred models of the universe undergo accelerated expansion
induced by dark energy. One model for dark energy is a positive cosmological 
constant. It is consequently of interest to 
study Einstein's equations with a positive cosmological constant coupled to 
matter satisfying the ordinary energy conditions; the dominant 
energy condition etc. Due to the difficulty of analysing the behaviour of 
solutions to Einstein's equations in general, it is common to either study
situations with symmetry, or to prove stability results. In the present paper,
we do both. In fact, we analyse, in detail, the future asymptotic behaviour
of $\tn{3}$-Gowdy symmetric solutions to the Einstein-Vlasov equations with 
a positive cosmological constant. In particular, we prove the cosmic no-hair
conjecture in this setting. However, we also prove that the solutions are
future stable (in the class of all solutions). Some of the results hold in 
a more general setting. In fact, we obtain conclusions concerning the 
causal structure of $\tn{2}$-symmetric solutions, assuming only the presence
of a positive cosmological constant, matter satisfying various energy conditions
and future global existence. Adding the assumption of $\tn{3}$-Gowdy symmetry
to this list of requirements, we obtain $C^{0}$-estimates for all but one of 
the metric components. There is consequently reason to expect that many of the 
results presented in this paper can be generalised to other types of matter. 
\end{abstract}

\section{Introduction}

At the end of 1998, two research teams studying supernovae of type Ia 
announced the unexpected conclusion that the universe is expanding at an 
accelerating rate; cf. \cite{retal,petal99}. After the observations had 
been corroborated by other sources, there was a corresponding shift in the class
of solutions to Einstein's equations used to model the universe. In particular,
physicists attributed the acceleration to a form of matter they referred to 
as 'dark energy'. However, as the nature of the dark energy remains unclear, 
there are several models for it. The simplest one is that of a positive
cosmological constant (which is the one we use in the present paper), but
there are several other possibilities; cf., e.g., 
\cite{alan,aslowroll,akessence} and references cited therein for some 
examples. Combining the different observational data, the currently preferred
model of the universe is spatially homogeneous and isotropic (i.e., the 
cosmological principle is assumed to be valid), 
spatially flat, and has matter of the following forms: ordinary matter
(usually modelled by a radiation fluid and dust), dark matter (often modelled
by dust), and dark energy (often modelled by a positive cosmological constant). 
In the present paper, we are interested in the Einstein-Vlasov system. This 
corresponds to a different description of the matter than the one usually 
used. However, this system can also be used in order
to obtain models consistent with observations; cf., e.g., 
\cite[Chapter~28]{stab}. In fact, Vlasov matter has the property that it 
naturally behaves as radiation close to the singularity and as dust in the 
expanding direction, a desirable feature which is usually put in by hand when 
using perfect fluids to model the matter. 

\textbf{The cosmic no-hair conjecture.}
The standard starting point in cosmology is the 
assumption of spatial homogeneity and isotropy. However, it is 
preferable to prove that solutions generally isotropise and that the
spatial variation (as seen by observers) becomes negligible. This is 
expected to happen in the presence of a positive cosmological constant;
in fact, solutions are in that case expected to appear de Sitter like to 
observers at late times. The latter expectation goes under the name of 
the \textit{cosmic no-hair conjecture}; cf. Conjecture~\ref{conj:conoha} for
a precise formulation. The main objective when studying
the expanding direction of solutions to Einstein's equations with a 
positive cosmological constant is to verify this conjecture. 

\textbf{Spatial homogeneity.} Turning to the results that have been obtained
in the past, it is natural to begin
with the spatially homogeneous setting. In 1983, Robert Wald wrote a short,
but remarkable, paper \cite{wald2}, in which he proves results concerning 
the future asymptotic behaviour of spatially homogeneous solutions to 
Einstein's equations with a positive cosmological constant. In particular, 
he confirms that the cosmic no-hair conjecture holds. What is remarkable
about the paper is the fact that he is able to obtain conclusions assuming
only that certain energy conditions hold and that the solution does not break 
down in finite time. Concerning the symmetry type, the only issue that 
comes up in the argument is whether it is compatible with 
the spatial hypersurfaces of homogeneity having positive scalar curvature or
not; positive scalar curvature of these hypersurfaces sometimes leads to 
recollapse. The results should be contrasted with the case of Einstein's 
vacuum equation in the spatially homogeneous setting, where the behaviour is 
strongly dependent on the symmetry type. Since Wald does not prove future 
global existence, 
it is necessary to carry out a further analysis in order to confirm the 
picture obtained in \cite{wald2} in specific cases. In the case of the 
Einstein-Vlasov system, this was done in \cite{hayounglambda}. It is also
of interest to note that it is possible to prove results analogous to those
of Wald for more general models for dark energy; cf., e.g., 
\cite{alan,aslowroll,akessence,hayoung}.

\textbf{Surface symmetry.}
Turning to the spatially inhomogeneous setting, there are results in the 
surface symmetric case with a positive cosmological constant; cf. 
\cite{TR,tan, tch1,letc}, and see \cite{alancs} 
for a definition of surface symmetry. In this case, the isometry group 
(on a suitable covering space) is $3$-dimensional. Nevertheless, the system
of equations that result after symmetry reduction is $1+1$-dimensional.
However, the extra symmetries do eliminate some of the degrees of freedom. 
Again, the main results are future causal geodesic completeness and a 
verification of the cosmic no-hair conjecture. 

\textbf{$\tn{2}$-symmetry.}
A natural next step to take after surface symmetry is to consider Gowdy
or $\tn{2}$-symmetry. That is the purpose of the present paper. In particular, 
we prove future causal geodesic completeness of solutions to the 
$\tn{3}$-Gowdy symmetric Einstein-Vlasov equations with a positive cosmological
constant (note, however, the caveat concerning global existence stated in 
Subsection~\ref{subsection:ge}). Moreover, we verify that the cosmic no-hair 
conjecture holds. It is
of interest to note that most of the arguments go through under the assumption
of $\tn{2}$-symmetry. However, in order to obtain the full picture in this 
setting, it is necessary to prove one crucial inequality, cf. 
Definition~\ref{def:las}, which we have not yet been able to do in general.

\textbf{Stability.} A fundamental question in the study of cosmological 
solutions is that of future stability: given initial data corresponding
to an expanding solution, do small perturbations thereof yield maximal globally
hyperbolic developments which are future causally geodesically complete and 
globally similar to the future? In the case of a positive cosmological 
constant, the first result was obtained by Helmut Friedrich; he proved 
stability of de Sitter space in $3+1$ dimensions in \cite{f}. Later, he 
and Michael Anderson generalised the result to higher (even) dimensions and 
to include various matter fields; cf. \cite{f3,anderson05}. Moreover, results 
concerning radiation fluids were obtained in \cite{lav}. However, conformal
invariance plays an important role in the arguments presented in
these papers. As a consequence, 
there seems to be a limitation of the types of matter models that can be 
dealt with using the corresponding methods. The paper \cite{mig4} was 
written with the goal of developing methods that are more generally 
applicable. The papers \cite{mig5,s,ras,speck12,speck12a,ial}, in which 
the methods developed in \cite{mig4} play a central role, indicate that 
this goal was achieved. In fact, a general future global non-linear stability
result for spatially homogeneous solutions to the Einstein-Vlasov equations
with a positive cosmological constant was obtained in \cite{stab}, the 
ideas developed in \cite{mig4} being at the core of the argument. In the 
present paper, we not only derive detailed future asymptotics of 
$\tn{3}$-Gowdy symmetric solutions to the Einstein-Vlasov equations with a 
positive cosmological constant. We also prove that all the resulting solutions
are future stable in the class of all solutions (without symmetry assumptions). 

\textbf{Outlook.} As we describe in the next subsection, some of the results 
concerning $\tn{3}$-Gowdy symmetric solutions hold irrespective of the matter 
model (as long as it satisfies the dominant energy condition and the 
non-negative pressure condition). As a consequence, we expect that it might 
be possible to derive detailed asymptotics in the case of the Einstein-Maxwell 
equations (with a positive cosmological constant), and in the case of the
Einstein-Euler system (though the issue of shocks may be relevant in the latter
case). Due to the stability results demonstrated in 
\cite{s,ras,speck12, speck12a}, it might also be possible to prove stability
of the corresponding solutions. 

\subsection{General results under the assumption of 
$\tn{2}$-symmetry}\label{subsection:ge}

\textbf{$\tn{2}$-symmetry.}
In the present paper, we are interested in $\tn{2}$-symmetric solutions 
to Einstein's equations. There are various geometric ways of imposing 
this type of symmetry (cf., e.g., \cite{c90,smulevici}), but for the purposes 
of the present paper, we simply assume the topology to be of the form 
$I\times\tn{3}$, where $I$ is an open interval contained in $(0,\infty)$. If 
$\theta$, $x$ and $y$ are 'coordinates' on $\tn{3}$ and $t$ is the 
coordinate on $I$, we also assume the metric to be of the form 
\begin{equation}\label{eq:metric}
g=t^{-1/2}e^{\lambda/2}(-dt^{2}+\a^{-1}d\theta^{2})+
te^{P}[dx+Qdy+(G+QH)d\theta]^{2}+
te^{-P}(dy+Hd\theta)^{2},
\end{equation}
where the functions $\a>0$, $\lambda$, $P$, $Q$, $G$ and $H$ only depend on 
$t$ and $\theta$; cf., e.g., \cite{smulevici}. Note that translation in the 
$x$ and $y$ directions defines 
a smooth action of $\tn{2}$ on the spacetime (as well as on each constant 
$t$-hypersurface). Moreover, the metric is invariant under this action, and the 
corresponding orbits are referred to as the symmetry orbits, given by 
$\{t\}\times\{\theta\}\times\tn{2}$. Note that the area of the symmetry orbits 
is proportional to $t$. For this reason, the foliation of the spacetime 
corresponding to the metric form (\ref{eq:metric}) is referred to as the 
constant areal time foliation. The case of $\tn{3}$-Gowdy symmetry
corresponds to the functions $G$ and $H$ being independent of time; again, 
there is a more geometric way of formulating this condition: the spacetime
is said to be Gowdy symmetric if the so-called \textit{twist quantities},
given by 
\begin{equation}\label{eq:twq}
J=\e_{\a\b\g\de}X^{\a}Y^{\b}\nabla^{\g}X^{\de},\ \ \
K=\e_{\a\b\g\de}X^{\a}Y^{\b}\nabla^{\g}Y^{\de},
\end{equation}
vanish, where $X=\d_{x}$ and $Y=\d_{y}$ are Killing fields of the above
metric and $\e$ is the volume form. A basic question to ask concerning
$\tn{2}$-symmetric solutions to Einstein's equations is whether the maximal 
globally hyperbolic development of initial data admits a constant areal time
foliation which is future global. There is a long history of proving such 
results. The first one was obtained by Vincent Moncrief, cf. \cite{m81},
in the case of vacuum solutions with $\tn{3}$-Gowdy symmetry. The case of 
$\tn{2}$-symmetric vacuum solutions with and without a positive cosmological 
constant have also been considered in \cite{cai} and \cite{bciam} respectively.
Turning to 
Vlasov matter, \cite{andreasson} contains an analysis of the existence of 
foliations in the $\tn{3}$-Gowdy symmetric Einstein-Vlasov setting. The
corresponding results were later extended to the $\tn{2}$-symmetric case in 
\cite{arw}. However, from our point of view, the most relevant result is that
of \cite{smulevici}. Due to the results of this paper, there is, given 
$\tn{2}$-symmetric initial 
data to the Einstein-Vlasov equations with a positive cosmological constant,
a future global foliation of the spacetime of the form (\ref{eq:metric}).
In other words $I=(t_{0},\infty)$.
Moreover, if the distribution function is not identically zero, then $t_{0}=0$.
Finally, if the initial data have Gowdy symmetry, then the same is true of the 
development. Strictly speaking, the future global existence result in 
\cite{smulevici} is based on the observation that the argument should not 
be significantly different from the proofs in \cite{bciam,cai,arw}. 
It would be preferable to have a complete proof of future global existence in 
the case of interest here, but we shall not provide it in this paper. 

\textbf{Results.} Turning to the results, it is of interest to note that some 
of the conclusions
can be obtained without making detailed assumptions concerning the matter 
content. For that reason, let us, for the remainder of this subsection, 
assume that we have a solution to Einstein's equations with
a positive cosmological constant, where the metric is of the form 
(\ref{eq:metric}), the existence interval $I$ is of the form $(t_{0},\infty)$
and the matter satisfies the dominant energy condition and the non-negative
pressure condition; recall that the matter is said to satisfy the 
\textit{dominant energy condition} if $T(u,v)\geq 0$ for all pairs $u,v$ of 
future directed timelike vectors (where $T$ is the stress energy tensor
associated with the matter); and that it is said to satisfy
the \textit{non-negative pressure condition} if $T(w,w)\geq 0$ for every 
spacelike vector $w$. To begin with, 
there is a constant $C>0$ such that $\a(t,\theta)\leq Ct^{-3}$ for all 
$(t,\theta)\in [t_{0}+2,\infty)\times\so$; cf. Proposition~\ref{prop:abd}.
In fact, this conclusion also holds if we replace the cosmological constant 
with a non-linear scalar field with a positive lower bound; cf. 
Remark~\ref{remark:nlsf}. One particular consequence of this estimate for
$\a$ is that the $\theta$-coordinate
of a causal curve converges. Moreover, observers whose $\theta$-coordinates
converge to different $\theta$-values are asymptotically unable to communicate.
In this sense, there is asymptotic silence. In the case of Gowdy symmetry, 
more can be deduced. In fact, for every $\e>0$, there is a $T>t_{0}$ such 
that 
\[
\l(t,\theta)\geq-3\ln t+2\ln\left(\frac{3}{4\Lambda}\right)-\e
\]
for all $(t,\theta)\in [T,\infty)\times\so$; cf. Proposition~\ref{prop:llb}.
This estimate turns out to be of crucial importance also in the general 
$\tn{2}$-symmetric case. For this reason, we introduce the following 
terminology. 

\begin{definition}\label{def:las}
A metric of the form (\ref{eq:metric}) which is defined for $t> t_{0}$
for some $t_{0}\geq 0$ is said to have $\lambda$-\textit{asymptotics} if 
there, for every $\e>0$, is a $T>t_{0}$ such that
\[
\l(t,\theta)\geq-3\ln t+2\ln\left(\frac{3}{4\Lambda}\right)-\e
\]
for all $(t,\theta)\in [T,\infty)\times\so$. 
\end{definition}
\begin{remark}
All Gowdy solutions have $\lambda$-asymptotics under the above assumptions;
cf. Proposition~\ref{prop:llb}.
\end{remark}

\begin{prop}
Consider a $\tn{2}$-symmetric solution to Einstein's equations with a 
positive cosmological constant. Assume that the matter satisfies the dominant 
energy condition and the non-negative pressure condition. Assume, moreover, 
that the corresponding metric admits a foliation of the form (\ref{eq:metric}),
on $I\times\tn{3}$, where $I=(t_{0},\infty)$ and $t_{0}\geq 0$. Finally, 
assume that the solution has $\lambda$-asymptotics and let $t_{1}=t_{0}+2$. 
Then there is a constant $C>0$ such that 
\begin{eqnarray*}
\left\|\lambda(t,\cdot)+3\ln t-2\ln\frac{3}{4\Lambda}\right\|_{C^{0}}
 & \leq & Ct^{-1/2},\\
t^{-3/2}\ldr{\a^{-1/2}(t,\cdot)}+
\|Q(t,\cdot)\|_{C^{0}}+\|P(t,\cdot)\|_{C^{0}} & \leq & C,\\
\|H_{t}(t,\cdot)\|_{L^{1}}+\|G_{t}(t,\cdot)\|_{L^{1}} & \leq & Ct^{-3/2}
\end{eqnarray*}
for all $(t,\theta)\in [t_{1},\infty)\times\so$. 
\end{prop}
\begin{remark}
The choice $t_{1}=t_{0}+2$ may seem unnatural. However, we need to stay
away from $t_{0}$ (since we do not control the solution close to $t_{0}$).
Moreover, in some situations we need to know that $\ln t$ is positive and
bounded away from zero. Since $t_{0}=0$ for most solutions, it is therefore
natural to only consider the interval $t\geq t_{0}+2$ in the study of the 
future asymptotics. 
\end{remark}
\begin{remark}\label{remark:average}
If $h$ is a scalar function on $\so$, we use the notation 
\begin{equation}\label{eq:ldrhdef}
\ldr{h}=\frac{1}{2\pi}\is hd\theta.
\end{equation}
Sometimes, we shall use the same notation for a scalar function
$h$ on $I\times\so$. In that case, $\ldr{h}$ is the function of $t$
defined by $\ldr{h(t,\cdot)}$. Finally, if $\bp\in\rn{3}$, we shall 
also use the notation $\ldr{\bp}$. However, in that case, 
$\ldr{\bp}=(1+|\bp|^{2})^{1/2}$; cf. Remark~\ref{remark:ws}.
\end{remark}
\begin{proof}
The statement is a consequence of Lemmas~\ref{lemma:lasstone},
\ref{lemma:pqfb} and \ref{lemma:twistest}.
\end{proof}
In particular, in the case of a $\tn{3}$-Gowdy symmetric solution,
there is asymptotic silence in the sense that the $\theta x y$-coordinates
of a causal curve converge, and causal curves whose asymptotic 
$\theta x y$-coordinates differ are asymptotically unable to communicate; cf. 
Proposition~\ref{prop:causal}.

\subsection{Results in the Einstein-Vlasov setting}\label{ssection:revl}

In order to be able to draw detailed conclusions, we need to restrict our
attention to a specific type of matter. In the present paper, we study the 
Einstein-Vlasov system. 

\textbf{A general description of Vlasov matter.}
Intuitively, Vlasov matter gives a statistical 
description of an ensemble of collections of particles. In practice, the 
matter is described by a distribution function defined on the space of 
states of particles. The possible states are given by the future directed
causal vectors (here and below, we assume the Lorentz manifolds under 
consideration to be time oriented). Usually, one distinguishes between massive 
and massless particles. In the latter case, the distribution function is 
defined on the future light cone, and in the former case, it is defined on 
the interior. In the present paper, we are interested in the massive case, 
and we assume all the particles to have unit mass (for a description of 
how to reduce the case of varying masses to the 
case of all particles having unit mass, see \cite{cah}). As a consequence, 
the distribution
function is a non-negative function on the mass shell $\mP$, defined to be
the set of future directed unit timelike vectors. In order to connect the 
matter to Einstein's equations, we need to associate a stress energy tensor
with the distribution function. It is given by 
\begin{equation}\label{eq:vset}
T^{\roV}_{\a\b}(\stp)=\ip{\stp} f p_{\a}p_{\b}\mup{\stp}.
\end{equation}
In this expression, $\mP_{\stp}$ denotes the set of future directed unit timelike
vectors based at the spacetime point $\stp$. In other words, if 
$T_{\rfr}\in T_{\stp}M$ is a future directed timelike vector, then 
\[
\mP_{\stp}=\{v\in T_{\stp}M\ |\ g(v,v)=-1,\ g(T_{\rfr},v)<0\}.
\]
Moreover, the Lorentz metric $g$ induces a Riemannian metric on $\mP_{\stp}$, 
and $\mup{\stp}$ denotes the corresponding volume form; cf. 
(\ref{eq:muprcor}) below for a coordinate representation of $\mup{\stp}$. 
Finally, $p_{\a}$ denotes the components of the 
one form obtained by lowering the index of $p\in \mP_{\stp}$ using the 
Lorentz metric $g$. Clearly, it is necessary to demand some degree of 
fall off of the distribution function $f$ in order for the integral 
(\ref{eq:vset}) to be well defined. In the present paper, we shall mainly
be interested in the case that the distribution function has compact 
support in the momentum directions (for a fixed spacetime point). However, 
in Subsections~\ref{ssection:stabno}--\ref{ssection:stable} we shall consider 
a somewhat more general
situation. Turning to the equation the distribution function has to satisfy, 
it is given by 
\begin{equation}\label{eq:gve}
\ml f=0.
\end{equation}
Here $\ml$ denotes the vector field induced on the mass shell by the 
geodesic flow; cf. (\ref{eq:mlcor}) below for a coordinate representation of
it. An alternate way to formulate this equation is to demand that
$f$ be constant along $\dot{\g}$ for every future directed unit timelike
geodesic $\g$. The intuitive interpretation of the Vlasov equation
(\ref{eq:gve}) is that collisions between particles are neglected. 
It is of interest to note that if $f$ satisfies the Vlasov equation, then 
the stress energy tensor is divergence free. To conclude, the 
\textit{Einstein-Vlasov equations with a positive cosmological constant} 
consist of (\ref{eq:gve}) and 
\begin{equation}\label{eq:efe}
\Ein+\Lambda g=T^{\roV},
\end{equation}
where $T^{\roV}$ is given by the right hand side of (\ref{eq:vset}) and $\Lambda$
is a positive constant. Moreover, 
\[
\Ein=\mathrm{Ric}-\frac{1}{2}Sg
\]
is the Einstein tensor, where 
$\mathrm{Ric}$ is the Ricci tensor and
$S$ is the scalar curvature of a Lorentz 
manifold $(M,g)$. The above description is somewhat brief, and the reader 
interested in more details is referred to, e.g.,
\cite{ehlers,rendall-97,andreassonLR,stab}. 

\textbf{Vlasov matter under the assumption of $\tn{2}$-symmetry.}
In the case of $\tn{2}$-symmetry, it is convenient to use a
symmetry reduced version of the distribution function. Introduce, to this
end, the orthonormal frame
\begin{equation}\label{eq:ONframe}
\begin{split}
e_{0} = & t^{1/4}e^{-\lambda/4}\d_{t},\ \ \
e_{1}=t^{1/4}e^{-\lambda/4}\a^{1/2}(\d_{\theta}-G\d_{x}-H\d_{y}),\\
e_{2} = & t^{-1/2}e^{-P/2}\d_{x},\ \ \
e_{3}=t^{-1/2}e^{P/2}(\d_{y}-Q\d_{x}).
\end{split}
\end{equation}
Since the distribution function $f$ is defined on the mass shell, it is 
convenient to parametrise this set; note that the manifolds we are
interested in here are parallelisable (i.e., they have a global frame). An 
element in $\mP$ can be written $v^{\a}e_{\a}$, where 
\[
v^{0}=[1+(v^{1})^{2}+(v^{2})^{2}+(v^{3})^{2}]^{1/2}.
\]
As a consequence, we can think of $f$ as depending on $v^{i}$, $i=1,2,3$, and 
the base point. However, due to the symmetry requirements, the distribution 
function only depends on the $t\theta$-coordinates of the base point. As a 
consequence, the distribution function can be considered to be a function of 
$(t,\theta,v)$, where $v=(v^{1}, v^{2}, v^{2})$. In what follows, we shall abuse
notation and denote the symmetry reduced function, defined on 
$I\times\so\times\rn{3}$, by $f$.  A symmetry reduced version of
the equations is to be found in Section~\ref{section:syeq}. 

\begin{remark}\label{remark:cptsupp}
In the $\tn{2}$-symmetric setting, we always assume the distribution function 
$f$ to have compact support when restricted to constant $t$-hypersurfaces. Under
the assumptions made in the present paper, $f$ has this 
property, assuming the initial datum for $f$ to have compact support. 
\end{remark}

The first question to ask concerning $\tn{2}$-symmetric solutions is that of 
existence of constant areal time foliations for an interval of the form 
$(t_{0},\infty)$. However, due to previous results, cf. \cite{smulevici}, we
know that $\tn{2}$-symmetric solutions to the Einstein-Vlasov equations with 
a positive cosmological constant are future global in this setting (keeping the
caveat stated in Subsection~\ref{subsection:ge} in mind). In other words, there
is a $t_{0}\geq 0$
such that the solution admits a foliation of the form (\ref{eq:metric})
on $I\times\tn{3}$, where $I=(t_{0},\infty)$. Consequently, the issue of 
interest here is that of the asymptotics. Unfortunately, we are unable to 
derive detailed asymptotics for all $\tn{2}$-symmetric solutions. However, 
we do obtain results for solutions with $\lambda$-asymptotics; recall that
all $\tn{3}$-Gowdy symmetric solutions fall into this class. 

\begin{thm}\label{thm:as}
Consider a $\tn{2}$-symmetric solution to the Einstein-Vlasov system with
a positive cosmological constant. Choose coordinates so that the corresponding
metric takes the form (\ref{eq:metric}) on $I\times\tn{3}$, where 
$I=(t_{0},\infty)$. Assume that the solution has $\lambda$-asymptotics
and let $t_{1}=t_{0}+2$. Then 
there are smooth functions $\a_{\infty}>0$, $P_{\infty}$, $Q_{\infty}$, $G_{\infty}$ 
and $H_{\infty}$ on $\so$, and, for every $0\leq N\in\zo$, a constant 
$C_{N}>0$ such that
\begin{eqnarray}
t\|H_{t}(t,\cdot)\|_{C^{N}}+t\|G_{t}(t,\cdot)\|_{C^{N}}+
\|H(t,\cdot)-H_{\infty}\|_{C^{N}}+\|G(t,\cdot)-G_{\infty}\|_{C^{N}}
 & \leq & C_{N}t^{-3/2},\label{eq:HGest}\\
t\|P_{t}(t,\cdot)\|_{C^{N}}+t\|Q_{t}(t,\cdot)\|_{C^{N}}+
\|P(t,\cdot)-P_{\infty}\|_{C^{N}}+\|Q(t,\cdot)-Q_{\infty}\|_{C^{N}}
 & \leq & C_{N}t^{-1},\label{eq:PQest}\\
\left\|\frac{\a_{t}}{\a}+\frac{3}{t}\right\|_{C^{N}}
+\left\|\lambda_{t}+\frac{3}{t}\right\|_{C^{N}} & \leq & C_{N}t^{-2},
\label{eq:altest}\\
\left\|t^{3}\a(t,\cdot)-\a_{\infty}\right\|_{C^{N}}+
\left\|\lambda(t,\cdot)+3\ln t-2\ln\frac{3}{4\Lambda}
\right\|_{C^{N}} & \leq & C_{N}t^{-1},\label{eq:alest}
\end{eqnarray}
for all $t\geq t_{1}$. Define $f_{\sca}$ via
$f_{\sca}(t,\theta,v)=f(t,\theta,t^{-1/2}v)$. Then there is an $R>0$ such that 
\[
\mathrm{supp} f_{\sca}(t,\cdot)\subseteq\so\times B_{R}(0)
\]
for all $t\geq t_{1}$, where $B_{R}(0)$ is the ball of radius $R$ in 
$\rn{3}$. Moreover, there is a smooth, non-negative function with compact 
support, say $f_{\sca,\infty}$, on $\so\times\rn{3}$, such that 
\[
t\|\d_{t}f_{\sca}(t,\cdot)\|_{C^{N}(\so\times\rn{3})}+
\|f_{\sca}(t,\cdot)-f_{\sca,\infty}\|_{C^{N}(\so\times\rn{3})}\leq C_{N}t^{-1}
\]
for all $t\geq t_{1}$. Turning to the geometry, let $\bg(t,\cdot)$ and
$\bk(t,\cdot)$ denote the metric and second fundamental form induced by 
$g$ on the hypersurface $\{t\}\times\tn{3}$, and let $\bg_{ij}(t,\cdot)$ 
denote the components of $\bg(t,\cdot)$ with respect to the vectorfields
$\d_{1}=\d_{\theta}$, $\d_{2}=\d_{x}$ and $\d_{3}=\d_{y}$ etc. Then 
\begin{equation}\label{eq:bgbkest}
\|t^{-1}\bar{g}_{ij}(t,\cdot)-\bar{g}_{\infty,ij}\|_{C^{N}}
+\|t^{-1}\bk_{ij}-\mH\bg_{\infty,ij}\|_{C^{N}}\leq C_{N}t^{-1}
\end{equation}
for all $t\geq t_{1}$, where $\mH=(\Lambda/3)^{1/2}$ and
\begin{equation}\label{eq:bginfdef}
\bar{g}_{\infty}=\frac{3}{4\Lambda\a_{\infty}}d\theta^{2}
+e^{P_{\infty}}[dx+Q_{\infty}dy+(G_{\infty}+Q_{\infty}H_{\infty})d\theta]^{2}+
e^{-P_{\infty}}(dy+H_{\infty}d\theta)^{2}.
\end{equation}
Moreover, the solution is future causally geodesically complete. 
\end{thm}
The proof of the above theorem is to be found in Section~\ref{section:stab}.

It is of interest to record how the spacetime appears to an observer. In 
particular, we wish to prove the cosmic no-hair conjecture in the present 
setting. The rough statement of this conjecture is that the spacetime
appears de Sitter like to late time observers. However, in order to be able
to state a theorem, we need a formal definition. Before proceeding
to the details, let us provide some intuition. Let 
\begin{equation}\label{eq:fldS}
g_{\dS}=-dt^{2}+e^{2\mH t}\bg_{E},
\end{equation}
where $\mH=(\Lambda/3)^{1/2}$ and $\bg_{E}$ denotes the standard flat Euclidean
metric. Then $(\rn{4},g_{\dS})$ corresponds to a part of de Sitter space. It may
seem more reasonable to consider de Sitter space itself. However, as far as
the asymptotic behaviour of de Sitter space is concerned, (\ref{eq:fldS}) is
as good a model as de Sitter space itself. Consider a future directed 
and inextendible causal curve in $(\rn{4},g_{\dS})$, say $\g=(\g^{0},\bga)$,
defined on $(s_{-},s_{+})$. Then $\bga(s)$ converges to some $\bx_{0}\in\rn{3}$
as $s\rightarrow s_{+}-$. Moreover, $\g(s)\in C_{\bx_{0},\Lambda}$ for all $s$, 
where
\[
C_{\bx_{0},\Lambda}=\{ (t,\bx):|\bx-\bx_{0}|\leq \mH^{-1}e^{-\mH t}\}.
\]
In practice, it is convenient to introduce a lower bound on the time 
coordinate and to introduce a margin in the spatial direction. Moreover,
it is convenient to work with open sets. We shall therefore be interested in 
sets of the form 
\begin{equation}\label{eq:clkt}
C_{\Lambda,K,T}=\{ (t,\bx):t> T,\ \ |\bx|< K\mH^{-1}e^{-\mH t}\};
\end{equation}
note that $\bx_{0}$ can be translated to zero by an isometry.
Since we are interested in the late time behaviour of solutions, it is natural 
to restrict attention to sets of the form $C_{\Lambda,K,T}$ for some 
$K\geq 1$ and $T>0$. 

\begin{definition}\label{def:fasdes}
Let $(M,g)$ be a time oriented, globally hyperbolic Lorentz manifold which is 
future causally geodesically complete. Assume, moreover, that $(M,g)$ is a 
solution to Einstein's equations with a positive cosmological constant 
$\Lambda$. Then $(M,g)$ is said to be 
\textit{future asymptotically de Sitter like} if there
is a Cauchy hypersurface $\Sigma$ in $(M,g)$ such that for every future 
oriented and inextendible causal curve $\g$ in $(M,g)$, the following holds:
\begin{itemize}
\item there is an open set $D$ in $(M,g)$, such that 
$J^{-}(\g)\cap J^{+}(\Sigma)\subset D$, and $D$ is diffeomorphic to 
$C_{\Lambda,K,T}$ for a suitable choice of $K\geq 1$ and $T>0$, 
\item using $\psi:C_{\Lambda,K,T}\rightarrow D$ to denote the diffeomorphism; 
letting $R(t)=K\mH^{-1}e^{-\mH t}$;
 using $\bg_{\dS}(t,\cdot)$ and $\bk_{\dS}(t,\cdot)$ to denote the metric 
and second fundamental form induced on $S_{t}=\{t\}\times B_{R(t)}(0)$ by 
$g_{\dS}$; using $\bg(t,\cdot)$ and $\bk(t,\cdot)$ to denote the metric and 
second fundamental form induced on $S_{t}$ by $\psi^{*}g$ (where $\psi^{*}g$ 
denotes the pullback of $g$ by $\psi$); and letting $N\in\no$, the following 
holds:
\begin{equation}\label{eq:asdS}
\lim_{t\rightarrow\infty}\left(
\|\bg_{\dS}(t,\cdot)-\bg(t,\cdot)\|_{C^{N}_{\dS}(S_{t})}
+\|\bk_{\dS}(t,\cdot)-\bk(t,\cdot)\|_{C^{N}_{\dS}(S_{t})}\right)=0.
\end{equation}
\end{itemize}
\end{definition}
\begin{remark}
In the definition, we use the notation 
\[
\|h\|_{C^{N}_{\dS}(S_{t})}=\left(\sup_{S_{t}}\sum_{l=0}^{N}
\bg_{\dS,i_{1}j_{1}}\cdots \bg_{\dS,i_{l}j_{l}}
\bg_{\dS}^{im}\bg_{\dS}^{jn}
\bn^{i_{1}}_{\dS}\cdots \bn^{i_{l}}_{\dS}h_{ij}
\bn^{j_{1}}_{\dS}\cdots \bn^{j_{l}}_{\dS}h_{mn}
\right)^{1/2}
\]
for a covariant $2$-tensor field $h$ on $S_{t}$, where $\bn_{\dS}$ denotes
the Levi-Civita connection associated with $\bg_{\dS}(t,\cdot)$. Note also
that, even though $R(t)$ shrinks to zero exponentially, the diameter of 
$S_{t}$, as measured with respect to $\bg_{\dS}(t,\cdot)$, is constant. 
\end{remark}
\begin{remark}
In some situations it might be more appropriate to adapt the Cauchy hypersurface
$\Sigma$ to the causal curve $\g$; i.e., to first fix $\g$ and then $\Sigma$. 
\end{remark}

The above definition leads to a formal statement of the cosmic no-hair
conjecture. 

\begin{conjecture}[Cosmic no-hair]\label{conj:conoha}
Let $\ma$ denote the class of initial data such that the corresponding
maximal globally hyperbolic developments (MGHD's) are
future causally geodesically complete solutions to 
Einstein's equations with a positive cosmological constant $\Lambda$ (for some
fixed matter model). Then generic elements of $\ma$ yield MGHD's that are 
\textit{future asymptotically de Sitter like}.
\end{conjecture}
\begin{remark}
It is probably necessary to exclude certain matter models in order for the 
statement to be correct. Moreover, the statement, as it stands, is quite 
vague; there is no precise definition of the notion generic. However, what 
notion of genericity is most natural might depend on the situation.
\end{remark}
\begin{remark}
The Nariai spacetimes, discussed, e.g., in \cite[pp. 126-127]{mig4}, are
time oriented, globally hyperbolic, causally geodesically complete solutions
to Einstein's vacuum equations with a positive cosmological constant 
that do not exhibit future asymptotically de Sitter
like behaviour. They are thus potential counterexamples to the cosmic
no-hair conjecture.  There is a similar example in the 
Einstein-Maxwell setting (with a positive cosmological constant) in 
\cite[p. 127]{mig4}. However, both of these examples are rather special, and 
it is natural to conjecture them to be unstable. Nevertheless, they
constitute the motivation for demanding genericity. 
\end{remark}

Finally, we are in a position to phrase a result concerning the cosmic no 
hair conjecture in the $\tn{3}$-Gowdy symmetric setting. The
proof of the theorem below is to be found in Section~\ref{section:stab}.

\begin{thm}\label{thm:conoha}
Consider a $\tn{2}$-symmetric solution to the Einstein-Vlasov system with
a positive cosmological constant. Choose coordinates so that the corresponding
metric takes the form (\ref{eq:metric}) on $I\times\tn{3}$, where 
$I=(t_{0},\infty)$. Assume that the solution has $\lambda$-asymptotics. Then 
the solution is future asymptotically de Sitter like; i.e., the cosmic
no-hair conjecture holds. 
\end{thm}
\begin{remark}
Recall that all $\tn{3}$-Gowdy symmetric solutions have $\lambda$-asymptotics.
\end{remark}
\begin{remark}
In the particular case of interest here, the equality (\ref{eq:asdS}) can
actually be improved to the estimate
\[
\|\bg_{\dS}(\tau,\cdot)-\bg(\tau,\cdot)\|_{C^{N}_{\dS}(S_{\tau})}
+\|\bk_{\dS}(\tau,\cdot)-\bk(\tau,\cdot)\|_{C^{N}_{\dS}(S_{\tau})}\leq C_{N}e^{-2\mH\tau}
\]
for all $\tau> T$ and a suitable constant $C_{N}$.
\end{remark}
\begin{remark}
The main estimate needed to prove the theorem is (\ref{eq:bgbkest}).
In situations where such an estimate holds, it is thus to be expected
that the solution is future asymptotically de Sitter like.
\end{remark}

\subsection{Stability, notation and function spaces}
\label{ssection:stabno}

Let us now turn to the subject of stability. Combining Theorem~\ref{thm:as}
with the results of \cite{stab}, it turns out to be possible to prove that 
the solutions to which Theorem~\ref{thm:as} applies are also future stable. 
In the present subsection, we begin by introducing the terminology necessary 
in order to make a formal statement of this result. 

Let $(M,g)$ be a time oriented $n+1$-dimensional Lorentz 
manifold. We say that $(\sfx,U)$ are \textit{canonical local coordinates}
if $\d_{\sfx^{0}}$ is future oriented timelike on $U$ and 
$g(\d_{\sfx^{i}}|_{\stp},\d_{\sfx^{j}}|_{\stp})$, $i,j=1,\dots,n$, are the components 
of a positive 
definite metric for every $\stp\in U$; cf. \cite[p. 87]{stab}. If 
$p\in \mP_{\stp}$ for some $\stp\in U$, we then define
\begin{equation}\label{eq:macodef}
\maco{x}(p)=\maco{x}(p^{\a}\d_{\sfx^{\a}}|_{\stp})=
[\sfx(\stp),\bp],
\end{equation}
where $\bp=(p^{1},\dots,p^{n})$. Note that $\maco{x}$ are local coordinates
on the mass shell. If $f$ is defined on the mass shell, we shall use the
notation $\sff_{\sfx}=f\circ\maco{x}^{-1}$. It is also convenient to introduce
the notation $\bsfp_{\sfx}$ according to 
$\maco{x}(p)=[\sfx(\stp),\bsfp_{\sfx}(p)]$, assuming $p\in \mP_{\stp}$. With this 
notation, the measure $\mu_{\mP_{r}}$ can be written
\begin{equation}\label{eq:muprcor}
\mu_{\mP_{r}}=-\frac{|g_{\sfx}(r)|^{1/2}}{\sfp_{\sfx,0}\circ\iota_{r}}
\iota_{r}^{*}d\bsfp_{\sfx},
\end{equation}
where $|g_{\sfx}|$ is the determinant of the metric $g$, when expressed with
respect to the $\sfx$-coordinates; $\iota_{r}:\mP_{r}\rightarrow\mP$ is the
inclusion; $\sfp_{\sfx}^{\a}(p)$ are the components of $p$ with respect to the
coordinates $\sfx$; and $\sfp_{\sfx,\a}(p)=g_{\sfx,\a\b}\sfp_{\sfx}^{\b}(p)$.
The reader interested in a derivation of (\ref{eq:muprcor}) is referred to 
\cite[Section~13.3]{stab}. Let us also note that the operator $\ml$ is given by 
\begin{equation}\label{eq:mlcor}
\ml=\sfp_{\sfx}^{\a}\frac{\d }{\d \sfx^{\a}}
-\G^{i}_{\a\b}\sfp_{\sfx}^{\a}\sfp_{\sfx}^{\b}\frac{\d }{\d \bsfp_{\sfx}^{i}}
\end{equation}
with respect to the above coordinates. 

In order to proceed, we need to introduce function spaces for the distribution
functions. Recall, to that end, \cite[Definition~7.1, p.~87]{stab}:

\begin{definition}\label{def:dinfintro}
Let $1\leq n\in\zo$, $\mu\in\ro$, $(M,g)$ be a time oriented $n+1$-dimensional 
Lorentz manifold and $\mP$ be the set of future directed unit timelike vectors. 
The space $\dm^{\infty}_{\mu}(\mP)$ is defined to consist of the smooth functions 
$f:\mP\rightarrow\ro$ such that, for every choice of canonical local coordinates 
$(\sfx,U)$, $n+1$-multiindex $\a$ and $n$-multiindex $\b$, the derivative
$\d^{\a}_{x}\d^{\b}_{\bp}\sff_{\sfx}$ (where $x$ symbolises the first $n+1$ and
$\bp$ the last $n$ variables), considered as a function from $\sfx(U)$ to the 
set of functions from $\rn{n}$ to $\ro$, belongs to 
\begin{equation}\label{eq:contpredeffs}
C[\sfx(U),L^{2}_{\mu+|\b|}(\rn{n})].
\end{equation}
\end{definition}
\begin{remark}\label{remark:ws}
The space $L^{2}_{\mu}(\rn{n})$ is the weighted $L^{2}$-space corresponding
to the norm
\begin{equation}\label{eq:wltsppre}
\|h\|_{L^{2}_{\mu}}=\left(\irn{n} \lbpr^{2\mu}|h(\bp)|^{2}d\bp\right)^{1/2},
\end{equation}
where $\lbpr=(1+|\bp|^{2})^{1/2}$; recall the comments made in
Remark~\ref{remark:average}. 
\end{remark}
\begin{remarks}
If $f\in\dm^{\infty}_{\mu}(\mP)$ for some $\mu>n/2+1$, then the stress energy
tensor is a well defined smooth function; cf. 
\cite[Proposition~15.37, p.~246]{stab}. Moreover, the stress energy tensor
is divergence free if $f$ satisfies the Vlasov equation. 
\end{remarks}

It is worth pointing out that it is possible to introduce more general 
function spaces, corresponding to a finite degree of differentiability;
cf. \cite[Definition~15.1, p.~234]{stab}. However, the above definition
is sufficient for our purposes. The above function spaces are suitable
when discussing functions on the mass shell. However, we also need to 
introduce function spaces for the initial datum for the distribution 
function. If $(\bsfx,U)$ are local coordinates on a manifold $\Sigma$, we
introduce local coordinates on $T\Sigma$ by 
$\bmaco{x}(\bp^{i}\d_{\bsfx^{i}}|_{\bxi})=(\bsfx(\bxi),\bp)$ in analogy with
(\ref{eq:macodef}). Moreover, if $\fb$ is defined on $T\Sigma$, we shall use
the notation $\bsff_{\bsfx}=\fb\circ\bmaco{x}^{-1}$. Let us recall
\cite[Definition~7.5, p.~89]{stab}:

\begin{definition}\label{def:indavlintro}
Let $1\leq n\in\zo$, $\mu\in\ro$ and $\Sigma$ be an $n$-dimensional
manifold. The 
space $\dhy^{\infty}_{\mu}(T\Sigma)$ is defined to consist of the smooth 
functions $\fb:T\Sigma\rightarrow\ro$ such that, for every choice of local 
coordinates $(\bsfx,U)$, $n$-multiindex $\a$ and $n$-multiindex $\b$, the 
derivative $\d^{\a}_{\bx}\d^{\b}_{\bp}\bsff_{\bsfx}$ (where $\bx$ symbolises the 
first $n$ and $\bp$ the last $n$ variables),
considered as a function from $\bsfx(U)$ to the set 
of functions from $\rn{n}$ to $\ro$, belongs to 
\[
C[\bsfx(U),L^{2}_{\mu+|\b|}(\rn{n})].
\]
\end{definition}
\begin{remark}
According to the criteria appearing in Definitions~\ref{def:dinfintro} and 
\ref{def:indavlintro}, we need to verify continuity conditions for every choice 
of local coordinates. However, it turns out to be sufficient to consider a 
collection of local coordinates covering the manifold of interest; cf. 
\cite[Lemma~15.9, p.~235]{stab} and \cite[Lemma~15.19, p.~237]{stab}.
\end{remark}

Finally, in order to be able to state a stability result, we need a
norm. Recall, to this end, \cite[Definition~7.7, pp.~89--90]{stab}:

\begin{definition}\label{definition:fshypintro}
Let $1\leq n\in\zo$, $0\leq l\in\zo$, $\mu\in\ro$ and $\Sigma$ be a compact 
$n$-dimensional manifold.  Let, moreover, $\bchi_{i}$, $i=1,...,j$, be a 
finite partition of unity subordinate to a cover consisting of coordinate 
neighbourhoods, say $(\bsfx_{i},U_{i})$. Then $\|\cdot\|_{H^{l}_{\roV,\mu}}$
is defined by
\begin{equation}\label{eq:fidnorm}
\|\fb\|_{H^{l}_{\roV,\mu}}=
\left(\sum_{i=1}^{j}
\sum_{|\a|+|\b|\leq l}\int_{\bsfx_{i}(U_{i})\times\rn{n}}
\langle \bvr\rangle^{2\mu+2|\b|}\bchi_{i}(\bxi)(\d^{\a}_{\bxi}
\d^{\b}_{\bvr}\bsff_{\bsfx_{i}})^{2}(\bxi,\bvr)d\bxi d\bvr\right)^{1/2}
\end{equation}
for each $\fb\in \dhy^{\infty}_{\mu}(T\Sigma)$.
\end{definition}
\begin{remark}
Clearly, the norm depends on the choice of partition of unity and on the 
choice of coordinates. However, different choices lead to equivalent norms. 
Here, we are mainly interested in the case $\Sigma=\tn{3}$, in which case
it is neither necessary to introduce local coordinates nor a partition of 
unity. 
\end{remark}

\subsection{The Einstein-Vlasov-non-linear scalar field system}

In the present paper, we are mainly interested in the Einstein-Vlasov
system with a positive cosmological constant. However, in the proof of future
stability of $\tn{3}$-Gowdy symmetric solutions,
we use two results. First, we use the fact that solutions that start out
close to de Sitter space are future stable. Second, we use Cauchy stability.
There are results of this type in the literature. However, they are formulated
in the Einstein-Vlasov-non-linear scalar field setting. In order to make it 
clear that the statements appearing in the literature can be applied in 
our setting, it is therefore necessary to briefly describe the 
Einstein-Vlasov-non-linear scalar field system. This is the purpose of the 
present subsection.

In $3+1$-dimensions, the Einstein-Vlasov-non-linear scalar field system can
be written
\begin{eqnarray}
R_{\a\b}- T_{\a\b}+\frac{1}{2}(\tr T) g_{\a\b} & = & 0,\label{eq:efeB}\\
\nabla^{\a}\nabla_{\a}\phi-V'\circ\phi & = & 0,\label{eq:sfeB}\\
\ml f & = & 0;
\label{eq:vezB}
\end{eqnarray}
cf. \cite[(7.13)--(7.15), p.~91]{stab}. In these equations,
$\phi\in C^{\infty}(M)$ is referred to as the \textit{scalar field};  
$V:\ro\rightarrow\ro$ is a smooth function referred to as the 
\textit{potential}; $\nabla$ is the Levi-Civita connection associated with 
the metric $g$; and 
\[
T_{\a\b}=T_{\a\b}^{\rosf}+T^{\roV}_{\a\b},
\]
where $T^{\roV}$ is defined in (\ref{eq:vset}) and
\[
T^{\rosf}_{\a\b}=\nabla_{\a}\phi\nabla_{\b}\phi
-\left[\frac{1}{2}\nabla^{\g}\phi\nabla_{\g}\phi+V(\phi)\right]g_{\a\b}.
\]
Assuming $V$ to be such that $V'(0)=0$, it is consistent to demand that $\phi$ 
be zero in (\ref{eq:sfeB}). Moreover, if $\phi=0$, then $T^{\rosf}=-V(0)g$. 
Letting $\Lambda=V(0)$, the equations (\ref{eq:efeB})--(\ref{eq:vezB}) then 
reduce to the Einstein-Vlasov system with a positive cosmological constant 
$\Lambda$, assuming $V(0)>0$. In order to prove future
stability in the Einstein-Vlasov-non-linear scalar field setting, it is not 
sufficient to demand that $V'(0)=0$ and $V(0)>0$. It is also of interest to 
know that $V''(0)>0$. We shall therefore make this assumption from now on. 
Given $V$ such that $V'(0)=0$, $V(0)>0$ and $V''(0)>0$, it is convenient to
introduce
\begin{equation}\label{eq:hdef}
\mH=(V(0)/3)^{1/2}
\end{equation}
and
\begin{equation}\label{eq:chidef}
\chi=V''(0)/\mH^{2};
\end{equation}
cf. \cite[(7.9) and (7.10), p.~90]{stab}. Note that in the non-linear scalar 
field setting, we always assume $V(0)$ to be positive and we equate it with 
$\Lambda$. In particular, (\ref{eq:hdef}) is thus consistent with previous
definitions of $\mH$; cf., e.g, the statement of Theorem~\ref{thm:as}. 
In case we are interested in the Einstein-Vlasov system with a positive 
cosmological constant $\Lambda$, it is sufficient to choose $V$ to be
\begin{equation}\label{eq:Vcc}
V(\phi)=\Lambda+\Lambda\phi^{2}.
\end{equation}
Then $V(0)=\Lambda>0$, $V'(0)=0$ and $V''(0)=2\Lambda>0$. Moreover, 
$\mH=(\Lambda/3)^{1/2}$ and $\chi=6$. Clearly, (\ref{eq:Vcc}) is an arbitrary
choice; there are many other possibilities. 

Let us now recall the definition of initial data given in 
\cite[Definition~7.11, pp.~93--94]{stab} (note that the dimension $n$ is 
here assumed to equal $3$):

\begin{definition}\label{def:idintro}
Let $5/2<\mu\in\ro$.
\textit{Initial data} for (\ref{eq:efeB})--(\ref{eq:vezB}) consist of 
an oriented $3$-dimensional manifold $\Sigma$, a non-negative function 
$\fb\in\dhy^{\infty}_{\mu}(T\Sigma)$, a Riemannian metric $\bg$, a symmetric
covariant $2$-tensor field $\bk$ and two functions $\bphi_{0}$ and 
$\bphi_{1}$ on $\Sigma$, all assumed to be smooth and to satisfy
\begin{eqnarray}
\br-\bk_{ij}\bk^{ij}+(\mathrm{tr}\bk)^{2} & = &
\bphi_{1}^{2}+\bn^{i}\bphi_{0} \bn_{i}\bphi_{0}+2V(\bphi_{0})+2\rho^{\roV},
\label{eq:hc}\\
\bn^{j}\bk_{ji}-\bn_{i}(\mathrm{tr}\bk) & = & \bphi_{1}\bn_{i}\bphi_{0}-\current_{i}^{\roV},\label{eq:mc}
\end{eqnarray}
where $\bn$ is the Levi-Civita connection of $\bg$, $\br$ is the associated 
scalar curvature, indices are raised and lowered by $\bg$ and $\rho^{\roV}$ and
$\current_{i}^{\roV}$ are given by (\ref{eq:rhoini}) and (\ref{eq:jiini}) below
respectively. Given initial data, the \textit{initial value problem} is that 
of finding a solution $(M,g,f,\phi)$ to (\ref{eq:efeB})--(\ref{eq:vezB})
(in other words, an $4$-dimensional manifold $M$, a smooth time oriented 
Lorentz metric $g$ on $M$, a non-negative function 
$f\in\dm^{\infty}_{\mu}(\mP)$ and a $\phi\in C^{\infty}(M)$ such 
that (\ref{eq:efeB})--(\ref{eq:vezB}) are satisfied), and an embedding
$i:\Sigma\rightarrow M$ such that  
\[
i^{*}g=\bg,\ \ \ \phi\circ i=\bphi_{0},\ \ \ 
\fb=i^{*}(f\circ\pros{i(\Sigma)}^{-1})
\]
and if $N$ is the future directed unit normal and $\kappa$ is the second
fundamental form of $i(\Sigma)$, then $i^{*}\kappa=\bk$ and 
$(N\phi)\circ i=\bphi_{1}$. Such a quadruple $(M,g,f,\phi)$ is referred to as 
a \textit{development}
of the initial data, the existence of an embedding $i$ being tacit. If, 
in addition to the above conditions, $i(\Sigma)$ is a Cauchy hypersurface
in $(M,g)$, the quadruple is said to be a \textit{globally hyperbolic
development}.
\end{definition}
\begin{remark}\label{remark:prosdef}
The map $\pros{i(\Sigma)}$ is the diffeomorphism from the mass shell 
above $i(\Sigma)$ to the tangent space of $i(\Sigma)$ defined by mapping
a vector $v$ to its component perpendicular to the normal of $i(\Sigma)$.
\end{remark}
\begin{remark}
If $\bphi_{0}=\bphi_{1}=0$, the equations (\ref{eq:hc}) and (\ref{eq:mc})
become
\begin{eqnarray}
\br-\bk_{ij}\bk^{ij}+(\mathrm{tr}\bk)^{2} & = &
2\Lambda+2\rho^{\roV},\label{eq:hcL}\\
\bn^{j}\bk_{ji}-\bn_{i}(\mathrm{tr}\bk) & = & -\current_{i}^{\roV}.\label{eq:mcL}
\end{eqnarray}
These are the constraint equations for the Einstein-Vlasov system with a 
positive cosmological constant $\Lambda$.
\end{remark}
The \textit{energy density} and \textit{current} induced by the initial data 
are given by
\begin{eqnarray}
\rho^{\roV}(\bxi) & = & 
\int_{T_{\bxi}\Sigma}\fb(\bp)[1+\bg(\bp,\bp)]^{1/2}
\bar{\mu}_{\bg,\bxi},\label{eq:rhoini}\\
\current^{\roV}(\bX) & = & \int_{T_{\bxi}\Sigma}\fb(\bp)\bg(\bX,\bp)
\bar{\mu}_{\bg,\bxi}.\label{eq:jiini}
\end{eqnarray}
In these expressions, $\bxi\in \Sigma$, $\bX\in T_{\bxi}\Sigma$, 
$\bar{\mu}_{\bg,\bxi}$ is the volume form on $T_{\bxi}\Sigma$ induced by 
$\bg$ and $\bp\in T_{\bxi}\Sigma$. 
It is important to note that under the assumptions of the above definition, 
the energy density is a smooth function and the current is a smooth 
one-form field on $\Sigma$; cf. \cite[Lemma~15.40, p.~246]{stab}. 

Given initial data, there is a unique maximal globally hyperbolic 
development thereof; cf. \cite[Corollary~23.44, p.~418]{stab}
and \cite[Lemma~23.2, p. 398]{stab}. The definition of a maximal 
globally hyperbolic development is given by 
\cite[Definition~7.14, p. 94]{stab}:

\begin{definition}\label{def:mghd}
Given initial data for (\ref{eq:efeB})--(\ref{eq:vezB}), a \textit{maximal
globally hyperbolic development} of the data is a globally hyperbolic 
development $(M,g,f,\phi)$, with embedding $i:\Sigma\rightarrow M$, such
that if $(M',g',f',\phi')$ is any other globally hyperbolic development of 
the same data, with embedding $i':\Sigma\rightarrow M'$, then 
there is a map $\psi:M'\rightarrow M$ which is a diffeomorphism onto
its image such that $\psi^{*}g=g'$, $\psi^{*}f=f'$, $\psi^{*}\phi=\phi'$ and 
$\psi\circ i'=i$.
\end{definition}
It is worth noting that the maximal globally hyperbolic development is 
independent of the parameter $\mu$. The above discussion of the 
initial value problem for the Einstein-Vlasov-non-linear scalar field
system is somewhat brief, and the 
reader interested in a more detailed discussion is referred to 
\cite[Chapter~7]{stab}.

\subsection{Future stability in the spatially homogeneous and 
isotropic setting}

In the proof of stability of the $\tn{3}$-Gowdy symmetric solutions, we need 
to refer to \cite[Theorem~7.16, pp.~104--106]{stab}.
However, the statement of this theorem is based on terminology introduced
in \cite{stab}. Moreover, in the statement of Theorem~\ref{thm:stab}, we refer 
to the conclusions of \cite[Theorem~7.16]{stab}. For this 
reason, we here provide not only the notational background, but also the 
statement of 
\cite[Theorem~7.16]{stab}. However, the reader interested in a 
discussion giving a justification for why the particular formulation of the 
theorem is natural is referred to \cite[Sections~7.6--7.7]{stab}. 

The rough idea of the statement is to only make local assumptions concerning 
the initial data and to derive future global conclusions concerning the 
solution. Given a $3$-manifold $\Sigma$, we 
therefore focus on a local coordinate patch $(\bsfx,U)$. Here $U$ is the 
neighbourhood in which we make assumptions in the statement of the theorem. 
The conditions on the initial data are phrased in terms of Sobolev norms 
on $U$. Given a tensor field $\mathfrak{T}$ on $\Sigma$, we therefore define
\begin{equation}\label{eq:tensobnormintro}
\begin{split}
\|\mathfrak{T}\|_{H^{l}(U)} =
\left(\sum_{i_{1},...,i_{s}=1}^{3}\sum_{j_{1},...,j_{r}=1}^{3}
\sum_{|\a|\leq l}\int_{\bsfx(U)}|\d^{\a}
\mathfrak{T}^{i_{1}\cdots i_{s}}_{j_{1}\cdots j_{r}}\circ \bsfx^{-1}|^{2}d\bx
\right)^{1/2}.
\end{split}
\end{equation}
In this expression, the components of $\mathfrak{T}$ are computed with respect 
to the coordinates $\bsfx$ and the derivatives are taken with respect to 
$\bsfx$. In what follows, norms of the type $\|\mathfrak{T}\|_{H^{l}(U)}$ are 
always computed using a particular choice of local coordinates. The choice we 
have in mind should be clear from the context. 
In the formulation of Theorem~\ref{thm:main}, we also use the notation
\begin{equation}\label{eq:dermetrisobintro}
\|\d_{m}\bg\|_{H^{l}(U)}=
\left(\sum_{i,j=1}^{3}
\sum_{|\a|\leq l}\int_{\bsfx(U)}|\d^{\a}\d_{m}\bg_{ij}\circ \bsfx^{-1}|^{2}d\bx
\right)^{1/2}.
\end{equation}
To measure the local size of the distribution function, we need a weighted 
Sobolev norm. However, it is also necessary to allow the freedom to rescale 
the momentum variable in the definition of the norm. Since we have already 
motivated the need for this rescaling freedom in 
\cite[Subsection~7.6.1, pp.~100--102]{stab}, we shall not do so here. Given
a constant $w$, we simply define the local norm for the distribution function 
by 
\begin{equation}\label{eq:flocnormmth}
\begin{split}
{}^{w}\|\fb\|_{H^{l}_{\roV,\mu}(U)} =\left(
\sum_{|\a|+|\b|\leq l}\irn{3}\int_{\bsfx(U)}(e^{-w})^{2|\b|}\langle
e^{w}\bp \rangle^{2\mu+2|\b|}|\d^{\a}_{\bxi}
\d^{\b}_{\bp}\sffb_{\bsfx}|^{2}(\bxi,\bp)d\bxi d\bp\right)^{1/2}.
\end{split}
\end{equation}
Here $\bmaco{x}$ are the coordinates on $TU$ associated with $\bsfx$
(cf. Subsection~\ref{ssection:stabno}), and 
$\sffb_{\bsfx}=\fb\circ\bmaco{x}^{-1}$. 

Given the above notation, \cite[Theorem~7.16, pp.~104--106]{stab} takes
the following form for $n=3$. 

\begin{thm}\label{thm:main}
Let $5/2<\mu\in\ro$ and $7/2<k_{0}\in\zo$. Let $V$ be a smooth 
function on $\ro$ such that $V(0)=V_{0}>0$, $V'(0)=0$ and $V''(0)>0$. Let 
$\mH,\chi>0$ be defined by (\ref{eq:hdef}) and (\ref{eq:chidef}) respectively and
let $\krov\geq 0$. There is an $\varepsilon>0$, depending only on $\mu$ 
and $V$, such that if 
\begin{itemize}
\item $(\Sigma,\bg,\bk, \fb,\bphi_{0},\bphi_{1})$ are initial data 
for (\ref{eq:efeB})--(\ref{eq:vezB}) with $\mathrm{dim}\Sigma=3$,
\item $\bsfx:U\rightarrow B_{1}(0)$ are local coordinates with 
$\bsfx(U)=B_{1}(0)$,
\item the inequality
\begin{equation}\label{eq:gijdeijsubest}
|e^{-2K}\bg_{ij}-\de_{ij}|\leq \varepsilon
\end{equation}
holds on $U$ for all $i,j=1,...,n$, where $K$ is defined by $e^{K}=4/\mH$, 
\item using the notation introduced in (\ref{eq:tensobnormintro}) and 
(\ref{eq:dermetrisobintro}), the inequality
\begin{equation}\label{eq:maincon}
\begin{split}
\sum_{j=1}^{3}\mH^{2}\|\d_{j}\bg\|_{H^{k_{0}}(U)}+
\mH\|\bk-\mH\bg\|_{H^{k_{0}}(U)}
+ \|\bphi_{0}\|_{H^{k_{0}+1}(U)}+
\mH^{-1}\|\bphi_{1}\|_{H^{k_{0}}(U)}\leq \varepsilon e^{-\krov}
\end{split}
\end{equation}
holds,
\item using the notation introduced in (\ref{eq:flocnormmth}), the inequality
\begin{equation}\label{eq:fbcond}
{}^{w}\|\fb\|_{H^{k_{0}}_{\roV,\mu}(U)}
\leq \mH^{2}\varepsilon^{5/2}e^{-3K/2-\krov}
\end{equation}
holds with $w=K+\krov$,
\end{itemize}
then the maximal globally hyperbolic development $(M,g,f,\phi)$ of the 
initial data has the property that if $i:\Sigma\rightarrow M$
is the associated embedding, then all causal 
geodesics that start in $i\circ\bsfx^{-1}[B_{1/4}(0)]$ are future
complete. Furthermore, there is a $t_{-}<0$ and a smooth map
\begin{equation}\label{eq:emb}
\psi: (t_{-},\infty)\times B_{5/8}(0)\rightarrow M,
\end{equation}
which is a diffeomorphism onto its image, such that all causal 
curves that start in $i\circ\bsfx^{-1}[B_{1/4}(0)]$ remain in 
the image of $\psi$ to the future, and $g$, $f$ and $\phi$ have expansions 
of the form (\ref{eq:glimi})--(\ref{eq:vlasijimp})
in the solid cylinder $[0,\infty)\times B_{5/8}(0)$ when pulled back by $\psi$.
Finally, $\psi(0,\bxi)=i\circ \bsfx^{-1}(\bxi)$ for $\bxi\in B_{5/8}(0)$.
In the formulae below, Latin indices refer to the natural Euclidean
coordinates on $B_{5/8}(0)$ and $t$ is the natural time coordinate
on the solid cylinder. 
Let $\zeta=4\chi/9$,
\[
\lambda_{\pre}=\left\{\begin{array}{cl} \frac{3}{2}[1-(1-\zeta)^{1/2}] & 
\zeta\in (0,1)\\
\frac{3}{2} & \zeta\geq 1
\end{array}
\right.
\]
and $\lambda_{\rom}=\min \{ 1,\lambda_{\pre}\}$.
There is a smooth Riemannian metric $\bvr$ on $B_{5/8}(0)$ and, for every 
$l\geq 0$, a constant $K_{l}$ such that 
\begin{eqnarray}
\|e^{2\mH t+2K}g^{ij}(t,\cdot)-\bvr^{ij}\|_{C^{l}}
+\|e^{-2\mH t-2K}g_{ij}(t,\cdot)-\bvr_{ij}\|_{C^{l}} & \leq & 
K_{l}e^{-2\lambda_{\rom}\mH t},\phantom{tram}\label{eq:glimi}\\
\|e^{-2\mH t-2K}\d_{t}g_{ij}(t,\cdot)
-2\mH\bvr_{ij}\|_{C^{l}} &\leq & K_{l}e^{-2\lambda_{\rom}\mH t},\label{eq:gtlimi}
\end{eqnarray}
for every $l\geq 0$ and $t\geq 0$. Here $\bvr^{ij}$ denotes the components 
of the inverse of $\bvr$. Furthermore, $C^{l}$ denotes the $C^{l}$-norm on 
$B_{5/8}(0)$. Turning to $g_{0m}$, there is a $b>0$ and, for every $l\geq 0$, a
constant $K_{l}$ such that  
\begin{equation}\label{eq:slimi}
\left\| g_{0m}(t,\cdot)-\assh_{m}\right\|_{C^{l}}
+\|\d_{0}g_{0m}(t,\cdot)\|_{C^{l}}\leq K_{l}e^{-b\mH t},
\end{equation}
for all $l\geq 0$ and $t\geq 0$, where 
\begin{equation}\label{eq:videfi}
\assh_{m}=\frac{1}{\mH}\bvr^{ij}\g_{imj}
\end{equation}
and $\gamma_{imj}$ denote the Christoffel symbols of the metric 
$\bvr$, given by 
\[
\gamma_{imj}=\frac{1}{2}(\d_{i}\bvr_{jm}+\d_{j}\bvr_{im}-\d_{m}\bvr_{ij}). 
\]
Let $\bk_{ij}$ denote the components of the second fundamental form (induced on the constant-$t$ hypersurfaces)
with respect to the standard coordinates on $B_{5/8}(0)$.  If $\lambda_{\rom}<1$, there is, for every $l\geq 0$, a constant $K_{l}$
such that 
\begin{eqnarray*}
\|g_{00}(t,\cdot)+1\|_{C^{l}}+\|\d_{0}g_{00}(t,\cdot)\|_{C^{l}}
& \leq & K_{l}e^{-2\lambda_{\rom}\mH t},\\
\|e^{-2\mH t-2K}\bk_{ij}(t,\cdot)
-\mH\bvr_{ij}\|_{C^{l}} &\leq & K_{l}e^{-2\lambda_{\rom}\mH t}
\end{eqnarray*}
for every $l\geq 0$ and $t\geq 0$. If $\lambda_{\rom}=1$, there is, for every $l\geq 0$,  a constant $K_{l}$
such that 
\begin{eqnarray*}
\|[\d_{0}g_{00}+2\mH(g_{00}+1)](t,\cdot)\|_{C^{l}} & \leq & K_{l}
e^{-2\mH t},\\
\|g_{00}(t,\cdot)+1\|_{C^{l}} & \leq & K_{l}(1+t^{2})^{1/2} e^{-2\mH t},\\
\|e^{-2\mH t-2K}\bk_{ij}(t,\cdot)
-\mH\bvr_{ij}\|_{C^{l}} &\leq & K_{l}(1+t^{2})^{1/2}e^{-2\mH t}
\end{eqnarray*}
for every $l\geq 0$ and $t\geq 0$. In order to describe the asymptotics concerning $\phi$, let $\varphi=e^{\lambda_{\pre} \mH t}\phi$. If $\zeta<1$, there is a smooth function $\varphi_{0}$, a constant $b>0$ and, for every $l\geq 0$, a constant 
$K_{l}$ such that 
\begin{equation}\label{eq:pcoi}
\|\varphi(t,\cdot)-\varphi_{0}\|_{C^{l}} 
+\|\d_{0}\varphi\|_{C^{l}} \leq  K_{l}e^{-b \mH t}
\end{equation}
for all $l\geq 0$ and $t\geq 0$. If $\zeta=1$, there are smooth functions $\varphi_{0}$ and $\varphi_{1}$, a constant $b>0$
and, for every $l\geq 0$, a constant $K_{l}$ such that 
\begin{equation}\label{eq:pcti}
\|\d_{0}\varphi(t,\cdot)-\varphi_{1}\|_{C^{l}} 
+\|\varphi(t,\cdot)-\varphi_{1}t-\varphi_{0}\|_{C^{l}}
\leq 
K_{l}e^{-b\mH t}
\end{equation}
for all $l\geq 0$ and $t\geq 0$.
Finally, if $\zeta>1$, there is an anti symmetric matrix $A$, given by
\[
A=\left(\begin{array}{cc}
               0 & \de \mH\\
              -\de \mH & 0
        \end{array}
  \right),
\]
where $\delta=3(\zeta-1)^{1/2}/2$, smooth functions $\varphi_{0}$ and
$\varphi_{1}$, a constant $b>0$ and, for every $l\geq 0$, a constant $K_{l}$ 
such that 
\begin{equation}\label{eq:pcthi}
\left\| e^{-At}\left(\begin{array}{c}
                                         \de \mH \varphi\\
					 \d_{0}\varphi
                                  \end{array}
                            \right)(t,\cdot)
                -\left(\begin{array}{c}
                                         \varphi_{0}\\
					 \varphi_{1}
                                  \end{array}
                            \right)\right\|_{C^{l}}\leq K_{l}e^{-b\mH t}
\end{equation}
for all $l\geq 0$ and $t\geq 0$. In order to describe the asymptotics
for the distribution function, let $\sfx=\psi^{-1}$. Then $(\sfx,U)$
are canonical local coordinates, where 
\[
U=\psi[(t_{-},\infty)\times B_{5/8}(0)]. 
\]
Let $\sff_{\sfx}=f\circ\maco{x}^{-1}$ and
\begin{equation}\label{eq:hfdefi}
\sfh(t,\bx,\bq)=\sff_{\sfx}(t,\bx,e^{-2\mH t-K-\krov}\bq).
\end{equation}
Introduce, moreover, the notation
\[
\|\bsff\|_{H^{l}_{\roV,\mu}[B_{5/8}(0)\times\rn{3}]}=\left(
\sum_{|\a|+|\b|\leq l}\int_{B_{5/8}(0)}\irn{3}
\lbpr^{2\mu+2|\b|}|\d^{\a}_{\bx}\d^{\b}_{\bp}\bsff(\bx,\bp)|^{2}d\bp d\bx
\right)^{1/2}
\]
for $\bsff\in C^{\infty}[B_{5/8}(0)\times\rn{3}]$.
Then there is a constant $b>0$ and, for every $l$, a constant $K_{l}$ such that
\begin{equation}\label{eq:dhdtdeci}
\|\d_{t}\sfh(t,\cdot)\|_{H^{l}_{\roV,\mu}[B_{5/8}(0)\times\rn{3}]}
\leq K_{l}e^{-b\mH t}
\end{equation}
holds for all $l\geq 0$ and $t\geq 0$. There is also a function 
$\bsfh\in C^{\infty}[B_{5/8}(0)\times\rn{3}]$, a constant $b>0$ and, for every 
$l$, a constant $K_{l}$ such that 
\begin{eqnarray}
\|\bsfh\|_{H^{l}_{\roV,\mu}[B_{5/8}(0)\times\rn{3}]} & < & \infty,\nonumber\\
\|\sfh(t,\cdot)-\bsfh\|_{H^{l}_{\roV,\mu}[B_{5/8}(0)\times\rn{3}]} & \leq & 
K_{l}e^{-b\mH t}\label{eq:hmhbesti}
\end{eqnarray}
hold for all $l\geq 0$ and $t\geq 0$.
Furthermore, $\bsfh\geq 0$. Concerning the stress energy tensor associated with
the Vlasov matter, there is a $b>0$ and, for every $l\geq 0$, a constant 
$K_{l}$ such that the estimates
\begin{eqnarray}
\left\|e^{3(\mH t+\krov)}T^{\roV}_{00}-\irn{3} \bsfh|\bvr|^{1/2}d\bq
\right\|_{C^{l}} & \leq & K_{l}e^{-b\mH t},
\phantom{nonse}\label{eq:vlaszz}\\
\left\|e^{3(\mH t+\krov)}T^{\roV}_{0i}+\irn{3} \bmfq_{i}\bsfh|\bvr|^{1/2}
d\bq\right\|_{C^{l}} & \leq & K_{l}e^{-b\mH t},
\label{eq:vlaszi}\\
\left\|e^{2\mH t+3\krov}T^{\roV}_{ij}\right\|_{C^{l}} & \leq & 
K_{l}\label{eq:vlasij}
\end{eqnarray}
hold for all $l\geq 0$ and all $t\geq 0$, where $|\bvr|$ denotes
the absolute value of the determinant of $\bvr$,
\[
\bmfq_{i}=\assh_{i}+e^{K-\krov}\bvr_{ij}\bq^{j}
\] 
and $\assh_{i}$ is defined in (\ref{eq:videfi}). Finally, if $\mu>9/2$, 
there is a constant $b>0$ and, for every $l\geq 0$, a constant $K_{l}$ 
such that 
\begin{equation}\label{eq:vlasijimp}
\left\|
e^{3(\mH t+\krov)}T^{\roV}_{ij}-
\irn{3} \sfhb \bmfq_{i}\bmfq_{j}|\bvr|^{1/2}
d\bq\right\|_{C^{l}}
\leq K_{l}e^{-b \mH t}
\end{equation}
holds for all $l\geq 0$ and $t\geq 0$.
\end{thm}
\begin{remark}
In case one is only interested in the Einstein-Vlasov setting with a 
positive cosmological constant, more detailed information can be 
obtained; cf. \cite[Proposition~32.8, pp.~609--611]{stab}. 
\end{remark}

\subsection{Cauchy stability}

In what follows, we also need a Cauchy stability result in the 
Einstein-Vlasov-non-linear scalar field setting. There are such 
results in the literature; cf. \cite{stab}. However, for the convenience
of the reader, we introduce the 
necessary terminology and quote the relevant result in the present 
subsection. 

To begin with, we need to introduce the notion of a background solution;
cf. \cite[Definition~24.2, p.~421]{stab}. In the $3$-dimensional case, 
this definition takes the following form.
\begin{definition}\label{definition:bgsol}
Let $5/2<\mu\in\ro$, $\Sigma$ be a closed 
$3$-dimensional manifold, and let $g$ be a smooth time oriented Lorentz metric 
on $M=I\times \Sigma$, where $I$ is an open interval. Let $\partial_{t}$ denote
differentiation with respect to the first coordinate and assume that 
$g(\partial_{t},\partial_{t})=g_{00}<0$ and that the hypersurfaces 
$\Sigma_{t}=\{ t\}\times \Sigma$ are spacelike with respect to 
$g$ for $t\in I$. Finally, assume that $\phi\in C^{\infty}
(M)$ and $f\in\dm^{\infty}_{\mu}(\mP)$, together with $g$, 
satisfy (\ref{eq:efeB})--(\ref{eq:vezB}).
Then $(M,g,f,\phi)$ is called a \textit{background solution}. 
\end{definition}
\begin{remark}
In the case of $\tn{2}$-symmetric solutions, the metric is of the form 
(\ref{eq:metric}). Moreover, the distribution functions of interest have 
compact support on constant time hypersurfaces. As a consequence, it is clear 
that the $\tn{2}$-symmetric solutions we consider in the present paper are 
background solutions in the sense of the above definition. 
\end{remark}

Next, we introduce the notion of induced initial data on constant $t$
hypersurfaces; cf. \cite[Definition~24.3, p.~421]{stab}. In the 
$3$-dimensional case, this definition takes the following form.

\begin{definition}
Let $5/2<\mu\in\ro$, $\Sigma$ be a closed
$3$-dimensional manifold, and let $g$ be a smooth time oriented Lorentz metric 
on $M=I\times \Sigma$, where $I$ is an open interval. Let, furthermore, 
$\phi\in C^{\infty}(M)$, $f\in\dm^{\infty}_{\mu}(\mP)$ and assume that 
$(g,f,\phi)$ solve (\ref{eq:efeB})--(\ref{eq:vezB}). Let $t\in I$ and assume
$\Sigma_{t}=\{t\}\times \Sigma$ to be spacelike with respect to $g$. Let,
furthermore, $\kappa$ be the second fundamental form and $N$ be
the future directed unit normal of $\Sigma_{t}$. Finally, let 
$\iota_{t}:\Sigma\rightarrow M$ be defined by $\iota_{t}(\bx)=(t,\bx)$ and
\[
\bg=\iota_{t}^{*}g,\ \ \
\bk=\iota_{t}^{*}\kappa,\ \ \ 
\fb=\iota_{t}^{*}(f\circ\pros{\Sigma_{t}}^{-1}),\ \ \
\bphi_{0}=\iota_{t}^{*}\phi,\ \ \
\bphi_{1}=\iota_{t}^{*}(N\phi).
\]
Then $(\bg,\bk,\fb,\bphi_{0},\bphi_{1})$ are referred to
as the \textit{initial data induced on} 
$\Sigma_{t}$ by 
$(g,f,\phi)$, or simply the initial data induced on $\Sigma_{t}$ 
if the solution is understood from the context. 
\end{definition}

Finally, we are in a position to formulate the Cauchy stability result
we need here; cf. \cite[Corollary~24.10, p.~432]{stab}. In the 
$3$-dimensional case, this result takes the following form.

\begin{thm}\label{thm:caus}
Let $5/2<\mu\in\ro$ and $5/2<l\in\zo$.
Let $(M_{\backg},g_{\backg},f_{\backg},\phi_{\backg})$ be a background 
solution with $M_{\backg}=I_{\backg}\times\Sigma$ and 
recall the notation $\Sigma$, $\Sigma_{t}$ etc. from Definition 
\ref{definition:bgsol} (the interval which was denoted by $I$ in Definition 
\ref{definition:bgsol} will here be denoted by
$I_{\backg}$).  Assume that $0\in I_{\backg}$ and let 
$(\bg_{\backg},\bk_{\backg},\fb_{\backg},\bphi_{\backg,0},\bphi_{\backg,1})$ be the 
initial data induced on $\Sigma_{0}$ by 
$(g_{\backg},f_{\backg},\phi_{\backg})$. Make a choice of 
$H^{l}_{\roV,\mu}(T\Sigma)$-norms and a choice of Sobolev 
norms $\|\cdot\|_{H^{l}}$ on tensor fields on $\Sigma$. Let $J\subset 
I_{\backg}$ be a compact interval and let $\e>0$. Then there 
is a $\de>0$ such that if $(\Sigma,\bg,\bk,\fb,\bphi_{0},\bphi_{1})$ are initial 
data for the Einstein-Vlasov-non-linear scalar field system satisfying
\begin{equation*}
\begin{split}
\|\bg-\bg_{\backg}\|_{H^{l+1}}+\|\bk-\bk_{\backg}\|_{H^{l}}
+\|\bphi_{0}-\bphi_{\backg,0}\|_{H^{l+1}}
+\|\bphi_{1}-\bphi_{\backg,1}\|_{H^{l}}
+\|\fb-\fb_{\backg}\|_{H^{l}_{\roV,\mu}(T\Sigma)} \leq \de,
\end{split}
\end{equation*}
then there is an open interval $I$ containing $0$ and a solution 
$(g,f,\phi)$ to (\ref{eq:efeB})--(\ref{eq:vezB}) on $M=I\times \Sigma$
such that 
\begin{itemize}
\item the initial data induced on $\Sigma_{0}$ by $(g,f,\phi)$ are given by 
$(\bg,\bk,\fb,\bphi_{0},\bphi_{1})$,
\item $\partial_{t}$ is timelike with respect to $g$ and 
$\Sigma_{t}$ is a spacelike Cauchy hypersurface with respect
to $g$ for all $t\in I$,
\item $J\subset I$ and if the initial data induced on $\Sigma_{t}$ 
(for $t\in I_{\backg}\cap I$) by 
$(g,f,\phi)$ and $(g_{\backg},f_{\backg},\phi_{\backg})$ are denoted by 
$(\bg_{t},\bk_{t},\fb_{t},\bphi_{t,0},\bphi_{t,1})$ and 
$(\bg_{\backg,t},\bk_{\backg,t},\fb_{\backg,t},\bphi_{\backg,t,0}
,\bphi_{\backg,t,1})$ respectively, then 
\begin{equation}\label{eq:causeest}
\begin{split}
\|\bg_{t}-\bg_{\backg,t}\|_{H^{l+1}}+\|\bk_{t}-\bk_{\backg,t}\|_{H^{l}}
+\|\bphi_{t,0}-\bphi_{\backg,t,0}\|_{H^{l+1}}& \\
+\|\bphi_{t,1}-\bphi_{\backg,t,1}\|_{H^{l}}
+\|\fb_{t}-\fb_{\backg,t}\|_{H^{l}_{\roV,\mu}(T\Sigma)} & \leq \e
\end{split}
\end{equation}
for all $t\in J$. 
\end{itemize}
\end{thm}

\subsection{Stability of $\tn{3}$-Gowdy symmetric solutions}
\label{ssection:stable}

Combining Theorems~\ref{thm:as}, \ref{thm:main} and \ref{thm:caus}
yields a future stability result for the $\tn{2}$-symmetric solutions
considered in Theorem~\ref{thm:as}. Moreover, the solutions are stable
in the Einstein-Vlasov-non-linear scalar field setting. 

\begin{thm}\label{thm:stab}
Consider a $\tn{2}$-symmetric solution to the Einstein-Vlasov system with
a positive cosmological constant $\Lambda$. Choose coordinates so that the 
corresponding metric takes the form (\ref{eq:metric}) on $I\times\tn{3}$, 
where $I=(t_{0},\infty)$. Assume that the solution has $\lambda$-asymptotics.
Choose a $t\in I$ and let $i:\tn{3}\rightarrow I\times\tn{3}$ be given 
by $i(\bx)=(t,\bx)$. Let $\bg_{\backg}=i^{*}g$ and let $\bk_{\backg}$ denote the 
pullback (under $i$) of the second fundamental form induced on $i(\tn{3})$ by 
$g$. Let, moreover, 
\[
\fb_{\backg}=i^{*}(f\circ\pros{i(\tn{3})}^{-1}).
\]
Make a choice of $\mu>5/2$, a choice of norms as in 
Definition~\ref{definition:fshypintro} and a choice of Sobolev norms on 
tensorfields on $\tn{3}$. Let, in addition, $V:\ro\rightarrow\ro$ be a smooth 
function such that $V(0)=\Lambda$, $V'(0)=0$ and $V''(0)>0$. Then there is an 
$\e>0$ such that if $(\tn{3},\bg,\bk,\fb,\bphi_{0},\bphi_{1})$ are initial data 
for (\ref{eq:efeB})--(\ref{eq:vezB}), with
$\fb\in\dhy^{\infty}_{\mu}(T\tn{3})$, satisfying 
\[
\|\bg-\bg_{\backg}\|_{H^{5}}+\|\bk-\bk_{\backg}\|_{H^{4}}
+\|\fb-\fb_{\backg}\|_{H^{4}_{\roV,\mu}}+\|\bphi_{0}\|_{H^{5}}
+\|\bphi_{1}\|_{H^{4}}\leq \e,
\] 
then  the maximal globally hyperbolic development $(M,g,f,\phi)$ of the 
initial data is future causally geodesically complete. Moreover, there is a 
Cauchy hypersurface $\Sigma$ in $(M,g)$ such that for each point of $\Sigma$, there
is a neighbourhood $(\bsfx,U)$ such that Theorem~\ref{thm:main} applies. In 
particular, the asymptotics stated in Theorem~\ref{thm:main} thus hold. 
\end{thm}
\begin{remark}
Up to the point where we appeal to Theorem~\ref{thm:main}, Cauchy stability 
applies. It should thus be possible to obtain detailed control over the 
perturbed solutions for the entire future. The interested reader is encouraged 
to write down the details. 
\end{remark}
\begin{remark}
The function $\fb_{\backg}$ has compact support, but $\fb$ need not have 
compact support. 
\end{remark}
The proof is to be found in Section~\ref{section:stab}. 

\subsection{Outline}

Finally, let us give an outline of the paper. In Section~\ref{section:syeq},
we write down the equations in the case that the metric takes the form 
(\ref{eq:metric}) (though the reader interested in a derivation is 
referred to Appendix~\ref{section:derofequations}). In 
Section~\ref{section:prel}, we then collect the conclusions 
which are not dependent on the particular type of matter model (as long 
as it satisfies the dominant energy condition and the non-negative pressure
condition). The section ends with conclusions concerning the causal structure
of $\tn{3}$-Gowdy symmetric spacetimes. Turning to the more detailed 
conclusions, we specialise to the case of Vlasov matter. The natural
first step is to derive light cone estimates; i.e., to consider
the behaviour along characteristics. This is the subject of 
Section~\ref{section:lces}. As opposed to the vacuum case, we need to control
the characteristics associated with the Vlasov equation at the same time 
as the first derivatives of the metric components. Fortunately, the 
$e_{2}$- and $e_{3}$-components of the momentum are controlled automatically due 
to the symmetry. However, an argument is required in the case of the 
$e_{1}$-component. In order to obtain control of higher order derivatives, 
we need to take derivatives of the characteristic system (associated with 
the Vlasov equation; i.e. with the geodesic flow). Naively, this should 
require control of second order derivatives of the metric functions, something
we do not have. Nevertheless, by an appropriate choice of variables, 
controlling first order derivatives turns out to be sufficient.
It is of interest to note that a similar choice was already
suggested in \cite[Lemma~3, p.~363]{andreasson}; cf. also 
\cite[Lemma~3, p.~257]{arw}. However, in the present setting, it is not 
sufficient to derive a system involving only first order derivatives of the
metric functions. We also need to 
be able to use the system to derive the desired type of asymptotics for 
the derivatives of the characteristic system. It turns out to be possible
to do this, and we write down the required arguments in 
Section~\ref{section:charsys}. After we have obtained this conclusion, it 
turns out to be possible to proceed inductively in order to derive higher
order estimates for the characteristic system and the metric components.
The required arguments are written down in Sections~\ref{section:holce}
and \ref{section:hodecs}. In order to obtain the desired conclusions concerning
the distribution function, it turns out to be convenient to consider 
$L^{2}$-based energies. This subject is treated in Section~\ref{section:eedf}.
Finally, in Section~\ref{section:stab}, we prove the main theorems of the paper.
As an appendix to the paper, we include a derivation of Einstein's equations
as well as of the Vlasov equation; cf. Appendix~\ref{section:derofequations}.
We also provide a summary of the most important notation in 
Appendix~\ref{section:notation}.

\section{Symmetry assumptions and equations}\label{section:syeq}

In this paper, we study $\tn{2}$-symmetric solutions of Einstein's equations. 
Since it will turn out to be convenient to express the equations using the 
orthonormal frame (\ref{eq:ONframe}), let us introduce the notation 
\begin{equation}\label{eq:rhodefetc}
\rho=T(e_{0},e_{0}),\ \ \
J_{i}=-T(e_{0},e_{i}),\ \ \
P_{i}=T(e_{i},e_{i}),\ \ \
S_{ij}=T(e_{i},e_{j}),
\end{equation}
where we do not sum over any indices; here, and below, we tacitly assume 
Latin indices to range from $1$ to $3$ and Greek indices to range from $0$ to 
$3$. It is also convenient to introduce the notation 
\begin{equation}\label{eq:JKdef}
J=-t^{5/2}\a^{1/2}e^{P-\lambda/2}(G_{t}+QH_{t}),\ \ \
K=QJ-t^{5/2}\a^{1/2}e^{-P-\lambda/2}H_{t}.
\end{equation}
Note that these objects are the twist quantities introduced in (\ref{eq:twq});
cf. Appendix~\ref{ssection:twqua}. In order to derive Einstein's equations,
it is useful to calculate the Einstein tensor for a metric of the form 
(\ref{eq:metric}). The corresponding, somewhat lengthy, computations are 
to be found in Section~\ref{section:derofequations}. Using the above notation, 
the calculations yield the conclusion that the $00$ and $11$-components of 
Einstein's equations can be written 
\begin{eqnarray}
\lambda_{t}-2\frac{\a_{t}}{\a} & = &
t\left[P_{t}^{2}+\a P_{\theta}^{2}+e^{2P}(Q_{t}^{2}+\a Q_{\theta}^{2})\right]
+\frac{e^{\lambda/2-P}J^{2}}{t^{5/2}}
+\frac{e^{\lambda/2+P}(K-QJ)^{2}}{t^{5/2}}\label{eq:Ezz}\\
& & +4t^{1/2}e^{\lambda/2}(\rho+\Lambda),\nonumber\\
\lambda_{t} & = & t\left[P_{t}^{2}+\a P_{\theta}^{2}+e^{2P}(Q_{t}^{2}+\a Q_{\theta}^{2})
\right]-\frac{e^{\lambda/2-P}J^{2}}{t^{5/2}}
-\frac{e^{\lambda/2+P}(K-QJ)^{2}}{t^{5/2}}\label{eq:altlteq}\\
 & & +4t^{1/2}e^{\lambda/2}(P_{1}-\Lambda),\nonumber
\end{eqnarray}
respectively. The $22$-component minus the $33$-component can be written 
\begin{equation}\label{eq:ttmthth}
\begin{split}
\d_{t}(t\a^{-1/2}P_{t}) = & \d_{\theta}(t\a^{1/2}P_{\theta})
+t\a^{-1/2}e^{2P}(Q_{t}^{2}-\a Q_{\theta}^{2})+
\frac{\a^{-1/2}e^{\lambda/2-P}J^{2}}{2t^{5/2}}\\
 & -\frac{\a^{-1/2}e^{\lambda/2+P}(K-QJ)^{2}}{2t^{5/2}}
+t^{1/2}e^{\lambda/2}\a^{-1/2}(P_{2}-P_{3}).
\end{split}
\end{equation}
The $22$-component plus the $33$-component can be written 
\begin{equation}\label{eq:ttpthth}
\begin{split}
\d_{t}\left[t\a^{-1/2}\left(\lambda_{t}
-2\frac{\a_{t}}{\a}-\frac{3}{t}\right)\right] = &
\d_{\theta}\left(t\a^{1/2}\lambda_{\theta}\right)
-t\a^{-1/2}\left[P_{t}^{2}+e^{2P}Q_{t}^{2}-\a (P_{\theta}^{2}+e^{2P}Q_{\theta}^{2})
\right]\\
 & -2t\a^{-1/2}\left(\frac{e^{\lambda/2-P}J^{2}}{t^{7/2}}
+\frac{e^{\lambda/2+P}(K-QJ)^{2}}{t^{7/2}}\right)\\
 & +\a^{-1/2}\lambda_{t}+2t^{1/2}e^{\lambda/2}\a^{-1/2}(2\Lambda-P_{2}-P_{3}).
\end{split}
\end{equation}
The $01$, $02$, $03$, $12$ and $13$-components are equivalent to 
\begin{eqnarray}
\lambda_{\theta} & = & 2t(P_{t}P_{\theta}+e^{2P}Q_{t}Q_{\theta})
-4t^{1/2}e^{\lambda/2}\a^{-1/2}J_{1},\label{eq:altltheq}\\
J_{\theta} & = & 2t^{5/4}\a^{-1/2}e^{P/2+\lambda/4}J_{2},\label{eq:Jth}\\
K_{\theta} & = & 2t^{5/4}\a^{-1/2}e^{-P/2+\lambda/4}J_{3}
+2t^{5/4}\a^{-1/2}e^{P/2+\lambda/4}QJ_{2},\label{eq:Kth}\\
J_{t} & = & -2t^{5/4}e^{\lambda/4+P/2}S_{12},\label{eq:Jt}\\
K_{t} & = & -2t^{5/4}e^{\lambda/4+P/2}QS_{12}
-2t^{5/4}e^{-P/2+\lambda/4}S_{13},\label{eq:Kt}
\end{eqnarray}
respectively. Finally, the $23$-component reads
\begin{equation}\label{eq:Etth}
\d_{t}(t\a^{-1/2}e^{2P}Q_{t})-\d_{\theta}(t\a^{1/2}
e^{2P}Q_{\theta})=t^{-5/2}\a^{-1/2}e^{\lambda/2+P}J(K-QJ)
+2t^{1/2}\a^{-1/2}e^{\lambda/2+P}S_{23}.
\end{equation}
For future reference, it is also of interest to note that 
\begin{eqnarray}
\frac{\a_{t}}{\a} & = & -\frac{e^{-P+\lambda/2}J^{2}}{t^{5/2}}
-\frac{e^{P+\lambda/2}(K-QJ)^{2}}{t^{5/2}}-4t^{1/2}e^{\lambda/2}\Lambda
-2t^{1/2}e^{\lambda/2}(\rho-P_{1}),\label{eq:altateq}\\
\lambda_{t}-\frac{\a_{t}}{\a} & = & t\left[P_{t}^{2}+\a P_{\theta}^{2}+e^{2P}(Q_{t}^{2}
+\a Q_{\theta}^{2})\right]+2t^{1/2}e^{\lambda/2}(\rho+P_{1}).\label{eq:altltmat}
\end{eqnarray}

\subsection{Preliminary calculations}

Since the metric components only depend on two variables, it is natural
to derive estimates by integrating along characteristics. In the present
subsection, we record a general calculation which is of interest in that 
context. To begin with, let us define
\begin{equation}\label{eq:dpmapm}
\d_{\pm}=\d_{t}\pm\a^{1/2}\d_{\theta}, \ \ \
\ma_{\pm}=(\d_{\pm}P)^{2}+
e^{2P}(\d_{\pm}Q)^{2}.
\end{equation}
One reason for introducing $\ma_{\pm}$ is the equality (\ref{leg}) derived
below; since the right hand side only contains first derivatives of the 
metric components, it is possible to integrate along the characteristics to 
control $\ma_{\pm}$. 

\begin{lemma}
Consider a $\tn{2}$-symmetric solution to Einstein's equations with a 
cosmological constant $\Lambda$ such that the metric takes the form 
(\ref{eq:metric}). Then 
\begin{equation}\label{leg}
\begin{split}
\d_{\pm}{\cal A}_{\mp} = &
-\left(\frac{2}{t}-\frac{\alpha_{t}}{\alpha}\right)\ma_{\mp}\mp
\frac{2}{t}\a^{1/2}
(P_{\theta}\d_{\mp}P+e^{2P}Q_{\theta}\d_{\mp}Q)\\
&+\frac{e^{-P+\lambda/2}J^{2}}{t^{7/2}}\d_{\mp}P
-\frac{e^{P+\lambda/2}(K-QJ)^{2}}{t^{7/2}}\d_{\mp}P
+2\frac{e^{\lambda/2}J(K-QJ)}{t^{7/2}}e^{P}\d_{\mp}Q\\
&
+2t^{-1/2}e^{\lambda/2}(P_{2}-P_{3})\d_{\mp}P
+4t^{-1/2}e^{\lambda/2}S_{23}e^{P}\d_{\mp}Q.
\end{split}
\end{equation}
\end{lemma}
\begin{remark}
In this calculation, the cosmological constant need not be positive. 
\end{remark}
\begin{proof}
The statement follows from a lengthy computation. Let us, however, for the 
benefit of the reader, write down some of the intermediate steps. Using
(\ref{eq:ttmthth}), we obtain 
\begin{equation}\label{SD}
\begin{split}
\d_{\pm}\d_{\mp}P = & -\frac{1}{t}P_{t}+\frac{\alpha_{t}}{2\alpha}\d_{\mp}P
+e^{2P}(Q_{t}^{2}-\alpha Q_{\theta}^{2})\\
&+\frac{e^{-P+\lambda/2}J^{2}}{2t^{7/2}}
-\frac{e^{P+\lambda/2}(K-QJ)^{2}}{2t^{7/2}}
+t^{-1/2}e^{\lambda/2}(P_{2}-P_{3}).
\end{split}
\end{equation}
Similarly, due to (\ref{eq:Etth}), we obtain
\begin{equation}\label{eq:Qderint}
\d_{\pm}\d_{\mp}Q =  -\frac{1}{t}Q_{t}+\frac{\a_{t}}{2\a}\d_{\mp}Q
-2(Q_{t}P_{t}-\alpha Q_{\theta}P_{\theta})
 +\frac{e^{\lambda/2-P}J(K-QJ)}{t^{7/2}}+2t^{-1/2}e^{\lambda/2-P}S_{23}.
\end{equation}
Combining (\ref{SD}) and (\ref{eq:Qderint}) with the fact that 
\[
-4(Q_{t}P_{t}-\a P_{\theta}Q_{\theta})\d_{\mp}Q+2\d_{\pm}P(\d_{\mp}Q)^{2}=
-2\d_{\mp}P(Q_{t}^{2}-\a Q_{\theta}^{2}),
\]
a calculation yields the conclusion of the lemma. 
\end{proof}

\subsection{Vlasov matter}

The equations (\ref{eq:Ezz})--(\ref{eq:altltmat}) hold in general. However, we 
are here particularly interested in matter of Vlasov type. In order to 
derive the relevant form of the Vlasov equation, recall the 
conventions concerning $f$ introduced in Subsection~\ref{ssection:revl}. 
Recall, moreover, the fact that the Vlasov equation is equivalent to $f$ 
being constant along future directed unit timelike geodesics. As a 
consequence, it can be calculated (cf. Appendix~\ref{ssection:thevlasoveq}) 
that the Vlasov equation takes the form
\begin{equation}\label{vv}
\begin{split}
\frac{\partial f}{\partial
  t}+\frac{\a^{1/2}v^{1}}{v^{0}}\frac{\partial
  f}{\partial\theta}-&\left[
\frac{1}{4}\a^{1/2}\lambda_{\theta}\,v^{0}
+\frac{1}{4}\left(\lambda_{t}-\frac{2\alpha_t}{\alpha}-\frac{1}{t}\right)v^{1}
-\a^{1/2}e^{P}Q_{\theta}\frac{v^2v^3}{v^0}
\right.\\
&\left.\phantom{[}+\frac{1}{2}\a^{1/2}P_{\theta}\frac{(v^{3})^{2}
-(v^{2})^{2}}{v^{0}}
-t^{-7/4}e^{\lambda/4}\big(e^{-P/2}Jv^2+e^{P/2}(K-Q J)v^3\big)
\right]\frac{\partial
  f}{\partial v^{1}}\\
-&\left[\frac{1}{2}\left(P_{t}+\frac{1}{t}\right)v^{2}
+\frac{1}{2}\a^{1/2}P_{\theta}\frac{v^{1}v^{2}}{v^{0}}\right]
\frac{\partial f}{\partial v^{2}}\\
-&\left[\frac{1}{2}\left(\frac{1}{t}-P_{t}\right)v^{3}
-\frac{1}{2}\a^{1/2}P_{\theta}\frac{v^{1}v^{3}}{v^{0}}+e^{P}v^2
\left(Q_t+\a^{1/2}Q_{\theta}\frac{v^1}{v^0}\right)\right]\frac{\partial
  f}{\partial v^{3}}=0.
\end{split}
\end{equation}
Turning to the stress energy tensor, it satisfies
\begin{equation}\label{eq:setvl}
T(e_{\mu},e_{\nu})=\irn{3} v_{\mu}v_{\nu}f\frac{1}{-v_{0}}dv,
\end{equation}
where $v_{\a}=\eta_{\a\b}v^{\b}$ and $\eta=\mathrm{diag}\{-1,1,1,1\}$. In 
particular, in the case of Vlasov, we thus have 
\begin{equation}\label{eq:rhoPketcvl}
\rho=\int_{\mathbb{R}^{3}}v^{0}f\;dv,\ \ \
P_{k}=
\int_{\mathbb{R}^{3}}\frac{(v^{k})^{2}}{v^{0}}
f\;dv, \ \ \
 J_k=\int_{\mathbb{R}^{3}}v^{k}f\;dv, \ \ \
S_{jk}=\int_{\mathbb{R}^{3}}\frac{v^jv^k}{v^0}f\;dv,
\end{equation}
where $j,k=1,2,3$.

\section{Preliminary conclusions concerning the asymptotics}
\label{section:prel}

In the present section, we are interested in $\tn{2}$-symmetric solutions to 
Einstein's equations such that the corresponding metric admits a 
foliation of the form (\ref{eq:metric}) on $I\times\tn{3}$, where 
$I=(t_{0},\infty)$ and $t_{0}\geq 0$. For the sake of brevity, we shall below 
refer to solutions of this form as \textit{future global}, and we shall speak 
of $t_{0}$ and $t_{1}=t_{0}+2$ without further introduction. 

It is useful to begin by recalling the following
consequences of the non-negative pressure condition and the dominant
energy condition. 

\begin{lemma}\label{lemma:codecnnpr}
Consider a solution to Einstein's equations with a cosmological constant 
$\Lambda$ and a metric of the form (\ref{eq:metric}). Let $\rho$,
$P_{i}$, $J_{i}$ and $S_{ij}$, $i,j=1,2,3$, be 
defined by (\ref{eq:rhodefetc}). If the stress energy tensor satisfies
the non-negative pressure condition, then, for $i=1,2,3$, 
\begin{equation}\label{eq:nnpr}
0\leq P_{i}.
\end{equation}
If the stress energy tensor satisfies the dominant energy 
condition, then, for $i,j=1,2,3$, 
\begin{eqnarray}
0 & \leq & \rho,\label{eq:decco}\\
|P_{i}| & \leq & \rho,\label{eq:decct}\\
|J_{i}| & \leq & \rho,\label{eq:deccth}\\
|S_{ij}| & \leq & \rho.\label{eq:deccf}
\end{eqnarray}
\end{lemma}
\begin{proof}
By definition, $P_{i}=T(e_{i},e_{i})$. Since $e_{i}$ is a spacelike vector field, 
the non-negative pressure condition implies that (\ref{eq:nnpr}) holds. 
The dominant energy condition states that $T(u,v)\geq 0$ for future directed 
timelike vectors $u$ and $v$. By continuity, this inequality also holds
for future directed causal vectors. Since $e_{0}$ is future directed timelike,
$\rho=T(e_{0},e_{0})\geq 0$, so that (\ref{eq:decco}) follows. 
Note that $e_{0}\pm e_{i}$ is a future directed causal vector field. In 
particular,
\[
0\leq T(e_{0}-e_{i},e_{0}+e_{j})=\rho+J_{i}-J_{j}-S_{ij}.
\]
Since $S_{ij}$ is symmetric, adding this inequality with the one obtained by
interchanging $i$ and $j$ yields the conclusion that $S_{ij}\leq\rho$. 
Similarly, 
\[
0\leq T(e_{0}\pm e_{i},e_{0}\pm e_{j})=\rho\mp J_{i}\mp J_{j}+S_{ij}.
\]
Adding the two inequalities yields $-S_{ij}\leq\rho$. Thus (\ref{eq:deccf})
holds. The proof of (\ref{eq:decct}) is similar. Finally, 
\[
0\leq T(e_{0},e_{0}\pm e_{i})=\rho\mp J_{i},
\]
so that (\ref{eq:deccth}) holds. 
\end{proof}

Before deriving estimates describing the asymptotics of solutions, let us make 
the following remark. 

\begin{remark}
In what follows, the constants appearing in the estimates we state are
allowed to depend on the solution, unless otherwise indicated. 
\end{remark}

\begin{prop}\label{prop:abd}
Given a future global solution to Einstein's equations with a cosmological 
constant $\Lambda>0$, $\tn{2}$-symmetry and a stress energy tensor satisfying
the dominant energy condition, there 
is a constant $C>0$ such that
\begin{equation}\label{eq:baaest}
\a(t,\theta)\leq Ct^{-3}
\end{equation}
for all $(t,\theta)\in [t_{1},\infty)\times\so$. 
\end{prop}
\begin{remark}\label{remark:nlsf}
The same conclusion holds if we replace the cosmological constant with 
a non-linear scalar field with a potential with a positive lower bound;
in other words, if we set $\Lambda=0$ and consider stress energy tensors
of the form $T=T^{\mathrm{o}}+T^{\mathrm{sf}}$, where $T^{\mathrm{o}}$ is the stress 
energy tensor associated with matter fields satisfying the dominant energy 
condition, and $T^{\mathrm{sf}}$ is the stress energy tensor associated 
with a non-linear scalar field with a potential $V$ having a positive lower
bound. 
\end{remark}
\begin{proof}
Due to (\ref{eq:altltmat}) and the fact that the matter satisfies the 
dominant energy condition (so that (\ref{eq:decct}) holds), we 
conclude that $\lambda_{t}-\a_{t}/\a\geq 0$. There is thus a $c_{0}>0$ such that 
\[
(\a^{-1/2}e^{\lambda/2})(t,\theta)\geq c_{0}
\]
for all $(t,\theta)\in [t_{1},\infty)\times\so$. Combining this observation with
(\ref{eq:altateq}) and (\ref{eq:decct}),  we obtain
\[
\d_{t}\a^{-1/2}=-\frac{\a_{t}}{2\a}\a^{-1/2}\geq 
2t^{1/2}\a^{-1/2}e^{\lambda/2}\Lambda\geq c_{1}t^{1/2}
\]
for some constant $c_{1}>0$ and all $(t,\theta)\in [t_{1},\infty)\times\so$.
Integrating this inequality, we obtain the conclusion of the proposition. 
\end{proof}

In the Gowdy case, the second and third terms on the right hand side of 
(\ref{eq:altlteq}) are zero, and as a consequence, we can extract more 
information. In fact, we have the following observation. 

\begin{prop}\label{prop:llb}
Consider a future global solution to Einstein's equations with a cosmological 
constant $\Lambda>0$, $\tn{3}$-Gowdy symmetry and matter satisfying the 
non-negative pressure condition. Then there is, for every $\e>0$, a $T>t_{0}$ 
such that
\[
\l(t,\theta)\geq-3\ln t+2\ln\left(\frac{3}{4\Lambda}\right)-\e
\]
for all $(t,\theta)\in [T,\infty)\times\so$. 
\end{prop}
\begin{proof}
Let
\begin{equation}\label{eq:hldef}
\hl=\l+3\ln t-2\ln\left(\frac{3}{4\Lambda}\right).
\end{equation}
Then (\ref{eq:altlteq}) with $J=K=0$ yields
\[
\d_{t}\hl=t\left[P_{t}^{2}+\a P_{\theta}^{2}+e^{2P}(Q_{t}^{2}
+\a Q_{\theta}^{2})\right]+4t^{1/2}e^{\lambda/2}P_{1}
+\frac{3}{t}(1-e^{\hl/2}). 
\]
Since $P_{1}\geq 0$ due to the non-negative pressure condition, cf. 
(\ref{eq:nnpr}), we conclude that 
\[
\d_{t}\hat{\lambda}\geq\frac{3}{t}(1-e^{\hat{\lambda}/2}).
\]
For every $\e>0$, there is thus a $T$ such that $\hl(t,\theta)\geq-\e$
for all $(t,\theta)\in [T,\infty)\times\so$. The proposition follows.
\end{proof}

In order to proceed, it is convenient to introduce an energy:
\begin{equation}\label{eq:ebasdef}
\begin{split}
E_{\bas} = & \is t\a^{-1/2}\left(\lambda_{t}-2\frac{\a_{t}}{\a}-4t^{1/2}e^{\lambda/2}
\Lambda\right)d\theta\\
 = & \is\left(t^{2}\a^{-1/2}\left[P_{t}^{2}+\a P_{\theta}^{2}+e^{2P}(Q_{t}^{2}
+\a Q_{\theta}^{2})\right]
+\frac{\a^{-1/2}e^{\lambda/2-P}J^{2}}{t^{3/2}}
\right.\\
&\left.\phantom{\frac{1}{2}}
+\frac{\a^{-1/2}e^{\lambda/2+P}(K-QJ)^{2}}{t^{3/2}}
+4t^{3/2}\a^{-1/2}e^{\lambda/2}\rho\right)d\theta.
\end{split}
\end{equation}
Let us motivate this particular choice. 
The energy $E_{\bas}$ is quite similar to the energy defined in 
\cite[(42), p. 251]{arw}. However, there is one fundamental difference. 
The integrand in the energy defined in \cite{arw} contains a term of 
the form $\a^{-1/2}U_{t}^{2}$, where $U=(P+\ln t)/2$. Using $U$ instead of 
$P$ as a variable is convenient in global existence arguments, since some
of the formulae become less involved. However, the variable $U$ is poorly
adapted to the actual asymptotics of solutions. The reason for this is that,
in the end, it turns out that $P_{t}$ converges to zero as $t^{-2}$. The
dominant term in $U$ is thus $\ln t/2$. Using $U$ instead of $P$ in the energy,
the best estimate of $E_{\bas}$ one could hope for would be $E_{\bas}(t)\leq 
Ct^{3/2}$. Below, we prove that $E_{\bas}\leq Ct^{1/2}$; cf. 
Lemma~\ref{lemma:lasstone}. In addition to this, it is possible to 
derive a good estimate for the time derivative of $E_{\bas}$; cf. the proof
of Lemma~\ref{lemma:Ebest} below. 

On a more general level, it is natural to ask why it is necessary to use 
$L^{2}$-based energy estimates at all. Since the problem is $1+1$-dimensional,
should it not be sufficient to consider the behaviour along characteristics?
The problem in our setting is that we wish to derive detailed quantitative 
information for arbitrary initial data. In particular, we are not in a 
situation where we can use bootstrap arguments. For this reason, we need
to proceed step by step. First, it is necessary to derive not only rough,
but quite detailed, control of some of the metric components, in particular
$\lambda$. This leads, for example, to estimates of the form (\ref{eq:Jenest}) 
and (\ref{eq:KQJenest}). Only once we have estimates of this form is it 
meaningful to turn to the characteristic system; cf., e.g., the last two terms 
on the right hand side of (\ref{eq:dVods}) and the proof of 
Lemma~\ref{lemma:dVodtest}. 

In case the metric has $\lambda$-asymptotics (recall Definition~\ref{def:las}),
it turns out to be possible to estimate $E_{\bas}$. 

\begin{lemma}\label{lemma:Ebest}
Consider a future global solution to Einstein's equations with a cosmological 
constant $\Lambda>0$, $\tn{2}$-symmetry, $\lambda$-asymptotics, and a stress 
energy tensor satisfying the dominant energy condition and the non-negative 
pressure condition. Then for every $a>1/2$, there is a constant $C_{a}>0$ such
that 
\begin{equation}\label{eq:Ebasprela}
E_{\bas}(t)\leq C_{a}t^{a}
\end{equation}
for all $t\geq t_{1}$.
\end{lemma}
\begin{proof}
Due to (\ref{eq:ttpthth}), we obtain 
\begin{equation}\label{eq:denden}
\begin{split}
\d_{t}\left[t\a^{-1/2}\left(\lambda_{t}-2\frac{\a_{t}}{\a}-4t^{1/2}e^{\lambda/2}
\Lambda\right)\right] = & \d_{\theta}(t\a^{1/2}\lambda_{\theta})
+2t\a^{1/2}(P_{\theta}^{2}+e^{2P}Q_{\theta}^{2})\\
 & -\frac{3}{2}\frac{\a^{-1/2}e^{\lambda/2-P}J^{2}}{t^{5/2}}
-\frac{3}{2}\frac{\a^{-1/2}e^{\lambda/2+P}(K-QJ)^{2}}{t^{5/2}}\\
 & -2t^{5/2}\a^{-1/2}e^{\lambda/2}\Lambda [P_{t}^{2}+e^{2P}Q_{t}^{2}
+\a (P_{\theta}^{2}+e^{2P}Q_{\theta}^{2})]\\
 & +t^{1/2}\a^{-1/2}e^{\lambda/2}(3\rho+P_{1}-2P_{2}-2P_{3})\\
 & -4t^{2}\a^{-1/2}e^{\lambda}\Lambda (\rho+P_{1}).
\end{split}
\end{equation}
Since the matter satisfies the non-negative pressure condition, 
we know that $P_{i}\geq 0$, cf. (\ref{eq:nnpr}), so that 
\begin{equation*}
\begin{split}
\frac{dE_{\bas}}{dt} \leq & \is 2t\a^{1/2}(P_{\theta}^{2}+e^{2P}Q_{\theta}^{2})d\theta
-\is 2t^{5/2}\a^{1/2}e^{\lambda/2}\Lambda (P_{\theta}^{2}+e^{2P}Q_{\theta}^{2})d\theta\\
 & +\is t^{1/2}\a^{-1/2}e^{\lambda/2}(3\rho+P_{1})d\theta
-\is 4t^{2}\a^{-1/2}e^{\lambda}\Lambda (\rho+P_{1})d\theta.
\end{split}
\end{equation*}
Using the consequences of Lemma~\ref{lemma:codecnnpr} and the fact that the 
solution has $\lambda$-asymptotics, we conclude that
for every $a>1/2$, there is a $T\geq t_{1}$ such that 
\[
\frac{dE_{\bas}}{dt}\leq \frac{a}{t}E_{\bas}
\]
for all $t\geq T$. As a consequence, $E_{\bas}(t)\leq Ct^{a}$ for $t\geq t_{1}$. 
\end{proof}

Using the estimate for $E_{\bas}$ derived in Lemma~\ref{lemma:Ebest}, it 
is possible to extract more information concerning the asymptotics. 

\begin{lemma}\label{lemma:lasstone}
Consider a future global solution to Einstein's equations with a cosmological 
constant $\Lambda>0$, $\tn{2}$-symmetry, $\lambda$-asymptotics, and a stress 
energy tensor satisfying the dominant energy condition and the non-negative 
pressure condition. Then there is a constant
$C>0$ such that 
\begin{eqnarray}
\left\|\lambda(t,\cdot)+3\ln t-2\ln\frac{3}{4\Lambda}\right\|_{C^{0}}
 & \leq & Ct^{-1/2},\label{eq:lhpreest}\\
E_{\bas}(t) & \leq & Ct^{1/2}\label{eq:Ebasest}
\end{eqnarray}
for all $t\geq t_{1}$.
\end{lemma}
\begin{proof}
Due to the estimate $E_{\bas}(t)\leq C_{a}t^{a}$, the fact that 
$\a^{1/2}\leq Ct^{-3/2}$, (\ref{eq:altlteq}) and (\ref{eq:decct}), we 
conclude that 
\begin{equation}\label{eq:ldtav}
\ldr{\lambda_{t}}=-4t^{1/2}\ldr{e^{\lambda/2}}\Lambda+O(t^{a-5/2});
\end{equation}
recall that the notation $\ldr{\lambda_{t}}$ was introduced in 
Remark~\ref{remark:average}. Due to (\ref{eq:altltheq})
and (\ref{eq:deccth}), we also have
\[
|\lambda_{\theta}|\leq t\a^{-1/2}[P_{t}^{2}+\a P_{\theta}^{2}
+e^{2P}(Q_{t}^{2}+\a Q_{\theta}^{2})]+4t^{1/2}\a^{-1/2}e^{\lambda/2}\rho.
\]
Due to (\ref{eq:Ebasprela}), we thus obtain
\begin{equation}\label{eq:lvarest}
\is|\lambda_{\theta}|d\theta\leq Ct^{a-1}.
\end{equation}
Recall that $\hl$ is defined in (\ref{eq:hldef}) and note that, due to 
(\ref{eq:ldtav}),
\begin{equation}\label{eq:lavtft}
\ldr{\hl_{t}}=\frac{3}{t}(1-\ldr{e^{\hl/2}})+O(t^{a-5/2}).
\end{equation}
Let us first prove that $\ldr{\hl}$ converges to zero. Let, to this end, 
$\e>0$. Since the solution has $\lambda$-asymptotics, we know that there
is a $T$ such that $\ldr{\hl}(t)\geq -\e$ for all $t\geq T$. In order 
to prove that there is a $T$ such that $\ldr{\hl}(t)\leq \e$ for all 
$t\geq T$, let us assume that $\ldr{\hl}(t)\geq \e$ for some $t$. 
Due to (\ref{eq:lvarest}), we conclude that $\hl(t,\theta)\geq\e/2$
for all $\theta\in\so$ (assuming $t$ to be large enough). Inserting
this information into (\ref{eq:lavtft}), we conclude that 
\[
\ldr{\hl_{t}}\leq \frac{2}{t}(1-e^{\e/4}),
\]
assuming $t$ to be large enough. Since the right hand side is negative 
and non-integrable, we conclude that $\ldr{\hl}$ has to decay until
it is smaller than $\e$ (assuming the starting time $t$ to be large enough). 
Moreover, $\ldr{\hl}$ cannot exceed $\e$ at a later time. In order to 
obtain a quantitative estimate, note that  
\[
\ldr{\hl_{t}}=\frac{3}{t}(1-e^{\ldr{\hl}/2})+O(t^{a-2}),
\]
where we have used the fact that $\hl$ is bounded to the future as well as
(\ref{eq:lvarest}). As a consequence, 
\begin{equation*}
\begin{split}
\d_{t}\ldr{\hl}^{2} = & 2\ldr{\hl}\ldr{\hl_{t}}=
\frac{6}{t}\ldr{\hl}\left[1-\left(1+\frac{1}{2}\ldr{\hl}+O(\ldr{\hl}^{2})
\right)\right]+O(t^{a-2}\ldr{\hl})\\
 = & -\frac{3}{t}\ldr{\hl}^{2}
+\frac{1}{t}O(\ldr{\hl}^{3})+O(t^{a-2}\ldr{\hl}).
\end{split}
\end{equation*}
Let $0<b<1-a$ and define
\[
\me=t^{2b}\ldr{\hl}^{2}.
\]
Then
\[
\frac{d\me}{dt}=\frac{2b}{t}\me-\frac{3}{t}\me
+\frac{1}{t}O(\ldr{\hl}\me)+t^{-1}O(t^{b+a-1}\me^{1/2}).
\]
As a consequence, there is a constant $C>0$ such that $\d_{t}\me\leq 0$
when $\me\geq C$ and $t$ is large enough. In particular, $\me$ is thus bounded 
to the future. For every $0<b<1/2$, there is thus a constant $C_{b}$ such that 
\[
\left\|\lambda(t,\cdot)+3\ln t-2\ln\frac{3}{4\Lambda}\right\|_{C^{0}}
\leq C_{b}t^{-b}
\]
for all $t\geq t_{1}$. Due to this estimate, we can return to the argument presented
in the proof of Lemma~\ref{lemma:Ebest} and obtain the improvement 
$E_{\bas}(t)\leq Ct^{1/2}$ for $t\geq t_{1}$. As a consequence, we can go through the 
above arguments with $a=1/2$ and $b=1/2$. The lemma follows. 
\end{proof}

\begin{lemma}\label{lemma:pqfb}
Consider a future global solution to Einstein's equations with a cosmological 
constant $\Lambda>0$, $\tn{2}$-symmetry, $\lambda$-asymptotics, and a stress 
energy tensor satisfying the dominant energy condition and the non-negative 
pressure condition. Then there is a constant $C>0$ such that 
\begin{equation}\label{eq:PQbdgen}
t^{-3/2}\ldr{\a^{-1/2}(t,\cdot)}+
\|Q(t,\cdot)\|_{C^{0}}+\|P(t,\cdot)\|_{C^{0}}  \leq C
\end{equation}
for all $t\geq t_{1}$. 
\end{lemma}
\begin{proof}
Estimate, using (\ref{eq:altateq}), (\ref{eq:hldef}),
(\ref{eq:lhpreest}), (\ref{eq:Ebasest}) and Lemma~\ref{lemma:codecnnpr},
\begin{equation*}
\begin{split}
\d_{t}\ldr{\a^{-1/2}} = & -\frac{1}{2}
\left\langle\a^{-1/2}\frac{\a_{t}}{\a}\right\rangle
=\frac{1}{2\pi}\is\left[\frac{\a^{-1/2}e^{-P+\lambda/2}J^{2}}{2t^{5/2}}
+\frac{\a^{-1/2}e^{P+\lambda/2}(K-QJ)^{2}}{2t^{5/2}}\right]d\theta\\
& +\frac{1}{2\pi}\is[2t^{1/2}\a^{-1/2}e^{\lambda/2}\Lambda
+t^{1/2}\a^{-1/2}e^{\lambda/2}(\rho-P_{1})]d\theta\\
\leq & Ct^{-1/2}+\frac{1}{2\pi}\is 2t^{1/2}\a^{-1/2}e^{\lambda/2}\Lambda d\theta
\\
\leq & Ct^{-1/2}+\frac{3}{2t}\ldr{e^{\hl/2}\a^{-1/2}}
\leq\frac{3}{2t}\ldr{\a^{-1/2}}+Ct^{-3/2}\ldr{\a^{-1/2}}+
Ct^{-1/2}.
\end{split}
\end{equation*}
Let $A=\ldr{\a^{-1/2}}+t$. Then
\[
\frac{dA}{dt}=\d_{t}\ldr{\a^{-1/2}}+1\leq\frac{3}{2t}A+Ct^{-3/2}A.
\]
Consequently, 
\[
\ln\frac{A(t)}{A(t_{1})}\leq \frac{3}{2}\ln t+C_{0},
\]
so that $\ldr{\a^{-1/2}}\leq Ct^{3/2}$ for $t\geq t_{1}$. 
Combining this estimate with (\ref{eq:Ebasest}) yields
\begin{equation}\label{eq:Pthlo}
\is |P_{\theta}|d\theta\leq\left(\is\a^{1/2}P_{\theta}^{2}d\theta
\right)^{1/2}\left(\is\a^{-1/2}d\theta\right)^{1/2}\leq
Ct^{-3/4}t^{3/4}\leq C
\end{equation}
for $t\geq t_{1}$. On the other hand, using (\ref{eq:baaest}) and 
(\ref{eq:Ebasest}) yields 
\begin{equation}\label{eq:Ptlo}
\left|\d_{t}\ldr{P}\right|=\left|\ldr{P_{t}}\right|
\leq \frac{1}{2\pi}\is |P_{t}|d\theta\leq
\frac{1}{\sqrt{2\pi}}\left(\is P_{t}^{2}d\theta\right)^{1/2}
\leq Ct^{-3/2}.
\end{equation}
Consequently, $\ldr{P}$ is bounded to the future. Combining these two 
observations, we conclude that
\[
\|P(t,\cdot)\|_{C^{0}}\leq C
\]
for all $t\geq t_{1}$. Combining this estimate for $P$ with the bound
(\ref{eq:Ebasest}), it is possible to derive $L^{1}$ estimates for 
$Q_{\theta}$ and $Q_{t}$ analogous to (\ref{eq:Pthlo}) and (\ref{eq:Ptlo}).
Consequently, $Q$ is bounded to the future. The lemma follows. 
\end{proof}

\begin{lemma}\label{lemma:twistest}
Consider a future global solution to Einstein's equations with a cosmological 
constant $\Lambda>0$, $\tn{2}$-symmetry, $\lambda$-asymptotics, and a stress 
energy tensor satisfying the dominant energy condition and the non-negative 
pressure condition. Then there is a constant
$C>0$ such that 
\begin{eqnarray}
\left\|\frac{e^{\lambda/2-P}J^{2}}{t^{5/2}}\right\|_{C^{0}} & \leq & Ct^{-2},
\label{eq:Jenest}\\
\left\|\frac{e^{P+\lambda/2}(K-QJ)^{2}}{t^{5/2}}\right\|_{C^{0}} & \leq & Ct^{-2}
\label{eq:KQJenest}
\end{eqnarray}
for all $t\geq t_{1}$. Moreover, for $t\geq t_{1}$,
\[
\|H_{t}\|_{L^{1}}+\|G_{t}\|_{L^{1}}\leq Ct^{-3/2}.
\]
\end{lemma}
\begin{proof}
Combining (\ref{eq:Jth}), (\ref{eq:deccth}), (\ref{eq:lhpreest}), 
(\ref{eq:Ebasest}) and (\ref{eq:PQbdgen}), we conclude that 
\begin{equation}\label{eq:Jthloest}
\is|J_{\theta}|d\theta\leq Ct^{5/4}\is \a^{-1/2}e^{\lambda/4}\rho d\theta
\leq Ct^{5/4}t^{-3/2}t^{3/4}\is t^{3/2}\a^{-1/2}e^{\lambda/2}\rho d\theta
\leq Ct. 
\end{equation}
The spatial variation of $J$ is consequently not greater than $Ct$. 
Combining (\ref{eq:Jt}), (\ref{eq:deccf}), (\ref{eq:baaest}),
(\ref{eq:lhpreest}), (\ref{eq:Ebasest})
and (\ref{eq:PQbdgen}) yields
\begin{equation}\label{eq:Jtloest}
\is|J_{t}|d\theta\leq Ct^{5/4}\is e^{\lambda/4}\rho d\theta
\leq Ct^{5/4}t^{-3/2}t^{3/4}t^{-3/2}\is t^{3/2}\a^{-1/2}e^{\lambda/2}\rho d\theta
\leq Ct^{-1/2}.
\end{equation}
As a consequence, 
\begin{equation}\label{eq:ldrJest}
|\ldr{J}|\leq Ct^{1/2}.
\end{equation}
Combining (\ref{eq:Jthloest}) and (\ref{eq:ldrJest}) yields
\begin{equation}\label{eq:Jlofb}
\|J(t,\cdot)\|_{C^{0}}\leq Ct.
\end{equation}
Due to (\ref{eq:lhpreest}), (\ref{eq:PQbdgen}) and (\ref{eq:Jlofb}), we 
conclude that (\ref{eq:Jenest}) holds. 

Let us now turn to $K-QJ$. To begin with, note that the $L^{1}$-norm of 
$Q_{\theta}$ is bounded to the future. The argument to prove this statement 
is similar to (\ref{eq:Pthlo}), keeping in mind that (\ref{eq:PQbdgen}) holds. 
Moreover, 
\[
\is|K_{\theta}-Q_{\theta}J-QJ_{\theta}|d\theta\leq 
\is |K_{\theta}-QJ_{\theta}|d\theta+\is |Q_{\theta}J|d\theta
\leq \is |K_{\theta}-QJ_{\theta}|d\theta+Ct,
\]
where we have used (\ref{eq:Jlofb}) and the fact that the $L^{1}$-norm of 
$Q_{\theta}$ is bounded. On the other hand, (\ref{eq:Jth}) and 
(\ref{eq:Kth}) yield
\[
K_{\theta}-QJ_{\theta}=2t^{5/4}\a^{-1/2}e^{-P/2+\lambda/4}J_{3}.
\]
Keeping (\ref{eq:deccth}) in mind, we can thus argue as in the proof of 
(\ref{eq:Jthloest}) in order to conclude that 
\begin{equation}\label{eq:Kcombothest}
\is|K_{\theta}-Q_{\theta}J-QJ_{\theta}|d\theta\leq Ct.
\end{equation}
In particular, the spatial variation of $K-QJ$ is bounded by $Ct$. 
On the other hand, the $L^{1}$-norm of $Q_{t}$ can be be bounded by
$Ct^{-3/2}$; cf. (\ref{eq:Ptlo}) and (\ref{eq:PQbdgen}). Combining this
observation with (\ref{eq:Jlofb}) yields
\begin{equation}\label{eq:Kcombotest}
\is |K_{t}-Q_{t}J-QJ_{t}|d\theta\leq 
\is |K_{t}-QJ_{t}|d\theta+\is |Q_{t}J|d\theta\leq
\is |K_{t}-QJ_{t}|d\theta+Ct^{-1/2}.
\end{equation}
Moreover, due to (\ref{eq:Jt}) and (\ref{eq:Kt}), 
\[
K_{t}-QJ_{t}=-2t^{5/4}e^{-P/2+\lambda/4}S_{13}.
\]
Keeping (\ref{eq:deccf}) in mind, we can proceed as in (\ref{eq:Jtloest}) in 
order to obtain
\[
\is |K_{t}-QJ_{t}|d\theta\leq Ct^{-1/2}.
\]
Due to (\ref{eq:Kcombotest}) and this estimate, the mean value of $K-QJ$ 
cannot grow faster than $Ct^{1/2}$. Combining this observation with
(\ref{eq:Kcombothest}) yields
\begin{equation}\label{eq:KmQJest}
\|K-QJ\|_{C^{0}}\leq Ct.
\end{equation}
Keeping (\ref{eq:lhpreest}) and (\ref{eq:PQbdgen}) in mind, we obtain 
(\ref{eq:KQJenest}). Using the fact that (\ref{eq:JKdef}) holds, 
we conclude that 
\[
\is |H_{t}|d\theta\leq \is t^{-5/2}\a^{-1/2}e^{P+\lambda/2}|K-QJ|d\theta.
\]
Combining this inequality with (\ref{eq:lhpreest}), (\ref{eq:PQbdgen})
and (\ref{eq:KmQJest}) yields the
desired $L^{1}$-estimate for $H_{t}$. A similar argument for $G_{t}$
yields the remaining conclusion of the lemma. 
\end{proof}

\subsection{Causal structure of $\tn{3}$-Gowdy symmetric solutions}

It is of interest to note that in the $\tn{3}$-Gowdy symmetric case, it 
is sufficient to assume future global existence and energy conditions
in order to conclude that there is asymptotic silence. In fact, we have
the following result. 

\begin{prop}\label{prop:causal}
Consider a future global $\tn{3}$-Gowdy symmetric solution to Einstein's 
equations with a cosmological constant $\Lambda>0$ and a stress 
energy tensor satisfying the dominant energy condition and the non-negative 
pressure condition. Then there is a constant $C$, depending only on the 
solution, such that if 
\[
\g(s)=[s,\theta(s),x(s),y(s)]=[s,\bga(s)]
\]
is a causal curve, then 
\begin{equation}\label{eq:gadoabest}
|\dot{\bga}(s)|^{2}\leq Cs^{-3}
\end{equation}
for $s\geq t_{1}$. In particular, there is a point $\bx_{0}\in\tn{3}$ such
that 
\begin{equation}\label{eq:gasilence}
d[\bga(s),\bx_{0}]\leq Cs^{-1/2}
\end{equation}
for all $s\geq t_{1}$, where $d$ is the standard metric on $\tn{3}$. 
\end{prop}
\begin{proof}
The causality of the curve is equivalent to the estimate 
\begin{equation}\label{eq:gacau}
\a^{-1}\dot{\theta}^{2}+s^{3/2}e^{P-\lambda/2}[\dot{x}+Q\dot{y}+
(G+QH)\dot{\theta}]^{2}
+s^{3/2}e^{-P-\lambda/2}(\dot{y}+H\dot{\theta})^{2}\leq 1.
\end{equation}
Note that in the case of Gowdy symmetry, $G$ and $H$ are time-independent. 
In particular, they are thus bounded. Due to (\ref{eq:PQbdgen}) we also
know that $Q$ is bounded for $t\geq t_{1}$. On the other hand, combining
(\ref{eq:gacau}) with (\ref{eq:baaest}), (\ref{eq:lhpreest}) and 
(\ref{eq:PQbdgen}) yields
\begin{eqnarray*}
|\dot{\theta}| & \leq & Cs^{-3/2},\\
|\dot{y}+H\dot{\theta}| & \leq & Cs^{-3/2},\\
|\dot{x}+Q\dot{y}+(G+QH)\dot{\theta}| & \leq & Cs^{-3/2}
\end{eqnarray*}
for $s\geq t_{1}$. Thus (\ref{eq:gadoabest}) holds, an estimate which implies
(\ref{eq:gasilence}). 
\end{proof}

\section{Light cone estimates}\label{section:lces}

In the presence of matter of Vlasov type, it is necessary to consider
the characteristic system in parallel with the light cone estimates for
the metric components. Let us therefore begin by writing down the 
characteristic system. It is given by 
\begin{eqnarray}
\frac{d\Theta}{ds} & = & \a^{1/2}\frac{V^{1}}{V^{0}},\label{eq:dThetads}\\
\frac{dV^{1}}{ds} & = & -\frac{1}{4}\a^{1/2}\lambda_{\theta}V^{0}-\frac{1}{4}
\left(\lambda_{t}-2\frac{\a_{t}}{\a}-\frac{1}{s}\right)V^{1}
+\a^{1/2}e^{P}Q_{\theta}\frac{V^{2}V^{3}}{V^{0}}\label{eq:dVods}\\
& & -\frac{1}{2}\a^{1/2}P_{\theta}
\frac{(V^{3})^{2}-(V^{2})^{2}}{V^{0}}+\frac{e^{-P/2+\lambda/4}J}{s^{7/4}}V^{2}
+\frac{e^{P/2+\lambda/4}(K-QJ)}{s^{7/4}}V^{3},
\nonumber\\
\frac{dV^{2}}{ds} & = & -\frac{1}{2}\left(P_{t}+\frac{1}{s}\right)V^{2}
-\frac{1}{2}\a^{1/2}P_{\theta}\frac{V^{1}V^{2}}{V^{0}},
\label{eq:dVtds}\\
\frac{dV^{3}}{ds} & = & -\frac{1}{2}\left(\frac{1}{s}-P_{t}\right)V^{3}
+\frac{1}{2}\a^{1/2}P_{\theta}\frac{V^{1}V^{3}}{V^{0}}-e^{P}Q_{t}V^{2}
-\a^{1/2}e^{P}Q_{\theta}\frac{V^{1}V^{2}}{V^{0}}.\label{eq:dVthds}
\end{eqnarray}
Note that in this system of equations, functions such as $\a^{1/2}$ should 
be evaluated at $[s,\Theta(s)]$. 
In view of the Vlasov equation (\ref{vv}), it is clear that the distribution
function is constant along characteristics. It is important to note that only
in the case of $V^{1}$ is it necessary to carry out an analysis; concerning 
$V^{2}$ and $V^{3}$ we automatically obtain the following estimate. 

\begin{lemma}\label{lemma:vtthest}
Consider a $\tn{2}$-symmetric solution to the Einstein-Vlasov equations with a 
cosmological constant $\Lambda>0$ and existence interval $(t_{0},\infty)$,
where $t_{0}\geq 0$. Assume that the solution has $\lambda$-asymptotics and 
let $t_{1}=t_{0}+2$. Then there is a constant $C>0$, depending only on 
the solution, such that if $\Theta$, $V$ is a 
solution to (\ref{eq:dThetads})--(\ref{eq:dVthds}) with initial data 
$\Theta(t_{1})$, $V(t_{1})$ such that $[t_{1},\Theta(t_{1}),V(t_{1})]$ is in the 
support of $f$, then
\[
|V^{2}(s)|+|V^{3}(s)|\leq Cs^{-1/2}
\]
for all $s\geq t_{1}$.  
\end{lemma}
\begin{remark}
As mentioned in Remark~\ref{remark:cptsupp}, we tacitly assume $f(t_{1},\cdot)$
to have compact support. 
\end{remark}
\begin{proof}
Due to (\ref{eq:dVtds}) and (\ref{eq:dVthds}), it can be verified that 
\begin{equation}\label{eq:v23con}
s^{1/2}e^{P/2}V^{2},\ \ \ 
s^{1/2}Qe^{P/2}V^{2}+s^{1/2}e^{-P/2}V^{3}
\end{equation}
are conserved along characteristics. Since we know $P$ and $Q$ to be
uniformly bounded, cf. Lemma~\ref{lemma:pqfb}, we obtain the conclusion of the 
lemma. 
\end{proof}

Let us now turn to $V^{1}$. To begin with, we have the following estimate.

\begin{lemma}\label{lemma:dVodtest}
Consider a $\tn{2}$-symmetric solution to the Einstein-Vlasov equations with a 
cosmological constant $\Lambda>0$ and existence interval $(t_{0},\infty)$, 
where $t_{0}\geq 0$. Assume that the solution has 
$\lambda$-asymptotics and let $t_{1}=t_{0}+2$. Then there is 
a constant $C>0$, depending only on the solution, 
such that if $\Theta$, $V$ is a solution to 
(\ref{eq:dThetads})--(\ref{eq:dVthds})
with initial data $\Theta(t_{1})$, $V(t_{1})$ such that 
$[t_{1},\Theta(t_{1}),V(t_{1})]$ is in the support of $f$, then
\begin{equation}\label{eq:dvodsest}
\begin{split}
\frac{d(V^{1})^{2}}{ds} \leq & -\frac{1}{s}(V^{1})^{2}
+Cs^{-1/2}e^{\lambda/2}(\mQ^{1})^{2}\frac{|V^{1}|}{V^{0}}\\
&+CsF\frac{|V^{1}|}{V^{0}}+Cs^{-1}F^{1/2}\frac{|V^{1}|}{V^{0}}
+Cs^{-3/2}(V^{1})^{2}+Cs^{-2}|V^{1}|
\end{split}
\end{equation}
for all $s\geq t_{1}$, where 
\begin{equation}\label{eq:mQodef}
\mQ^{1}(t):=\sup\{|v^{1}|: (t,\theta,v^{1},v^{2},v^{3})\in \mathrm{supp} 
f\}
\end{equation}
and
\begin{equation}\label{eq:Fdef}
F(t)=\sup_{\theta\in\so}\ma_{+}(t,\theta)+\sup_{\theta\in\so}\ma_{-}(t,\theta),
\end{equation}
where $\ma_{\pm}$ is defined in (\ref{eq:dpmapm}).
\end{lemma}
\begin{proof}
Due to (\ref{eq:dVods}), we have
\begin{equation*}
\begin{split}
\frac{d(V^{1})^{2}}{ds}  = & -\frac{1}{2}\a^{1/2}\lambda_{\theta}V^{0}V^{1}
-\frac{1}{2}
\left(\lambda_{t}-2\frac{\a_{t}}{\a}-\frac{1}{s}\right)(V^{1})^{2}
+2\a^{1/2}e^{P}Q_{\theta}\frac{V^{1}V^{2}V^{3}}{V^{0}}\\
 & -\a^{1/2}P_{\theta}V^{1}
\frac{(V^{3})^{2}-(V^{2})^{2}}{V^{0}}+2\frac{e^{-P/2+\lambda/4}J}{s^{7/4}}
V^{1}V^{2}
+2\frac{e^{P/2+\lambda/4}(K-QJ)}{s^{7/4}}V^{1}V^{3}.
\end{split}
\end{equation*}
However, due to Lemmas~\ref{lemma:twistest} and \ref{lemma:vtthest}, we 
can estimate the last two terms by $Cs^{-2}|V^{1}|$. We thus have 
\[
\frac{d(V^{1})^{2}}{ds}\leq -\frac{1}{2}\a^{1/2}\lambda_{\theta}V^{0}V^{1}
-\frac{1}{2}\left(\lambda_{t}-2\frac{\a_{t}}{\a}-\frac{1}{s}\right)(V^{1})^{2}
+Cs^{-1}F^{1/2}\frac{|V^{1}|}{V^{0}}+Cs^{-2}|V^{1}|,
\]
where we have used Lemma~\ref{lemma:vtthest}. Due to 
(\ref{eq:Ezz}) and (\ref{eq:altltheq}), the sum of the first and the 
second term on the right hand side can be written
\begin{equation}\label{eq:lderterVest}
\begin{split}
&2s^{1/2}e^{\lambda/2}(J_{1}V^{0}-\rho V^{1})V^{1}
-2s^{1/2}e^{\lambda/2}\Lambda(V^{1})^{2}
-\frac{e^{-P+\lambda/2}J^{2}}{2s^{5/2}}(V^{1})^{2}
-\frac{e^{P+\lambda/2}(K-QJ)^{2}}{2s^{5/2}}(V^{1})^{2}\\
&-s\a^{1/2}(P_{t}P_{\theta}+e^{2P}Q_{t}Q_{\theta})V^{0}V^{1}
-\frac{1}{2}s\left[P_{t}^{2}+\a P_{\theta}^{2}+e^{2P}(Q_{t}^{2}
+\a Q_{\theta}^{2})\right](V^{1})^{2}+\frac{1}{2s}(V^{1})^{2}.
\end{split}
\end{equation}
Note that 
\begin{equation*}
\begin{split}
-s\a^{1/2}P_{t}P_{\theta}V^{0}V^{1}
-\frac{1}{2}s[P_{t}^{2}+\a P_{\theta}^{2}](V^{1})^{2} \leq &
s\a^{1/2}|P_{t}P_{\theta}|V^{0}|V^{1}|
-\frac{1}{2}s[P_{t}^{2}+\a P_{\theta}^{2}](V^{1})^{2}\\
 \leq & s\a^{1/2}|P_{t}P_{\theta}||V^{1}|(V^{0}-|V^{1}|)\\
&-\frac{1}{2}s(|P_{t}|-\a^{1/2}|P_{\theta}|)^{2}(V^{1})^{2}.
\end{split}
\end{equation*}
Combining this estimate with a similar estimate for $Q$, we conclude that
the second and third last terms in (\ref{eq:lderterVest}) can be estimated
by 
\[
 s\a^{1/2}(|P_{t}P_{\theta}|+e^{2P}|Q_{t}Q_{\theta}|)|V^{1}|
\frac{1+(V^{2})^{2}+(V^{3})^{2}}{V^{0}+|V^{1}|}\leq CsF\frac{|V^{1}|}{V^{0}}.
\]
Due to Lemma~\ref{lemma:lasstone}, we have 
\[
-2s^{1/2}e^{\lambda/2}\Lambda+\frac{1}{2s}=-\frac{1}{s}+O(s^{-3/2}).
\]
Combining the above observations with Lemma~\ref{lemma:twistest}, we 
conclude that 
\begin{equation}\label{eq:dVosqestpre}
\begin{split}
\frac{d(V^{1})^{2}}{ds} \leq & -\frac{1}{s}(V^{1})^{2}
+2s^{1/2}e^{\lambda/2}(J_{1}V^{0}-\rho V^{1})V^{1}\\
&+CsF\frac{|V^{1}|}{V^{0}}+Cs^{-1}F^{1/2}\frac{|V^{1}|}{V^{0}}
+Cs^{-3/2}(V^{1})^{2}+Cs^{-2}|V^{1}|.
\end{split}
\end{equation}
Let us estimate the term
\begin{equation}\label{eq:Todef}
T_{1}:=2s^{1/2}e^{\lambda/2}(J_1 V^0-\rho V^{1})V^{1}
=2s^{1/2}e^{\lambda/2}(I^{-}+I^{+})V^{1},
\end{equation}
where
\begin{eqnarray}
\displaystyle I^{-}(s) &=&\int_{\mathbb{R}^{2}}\int_{-\infty}^{0}
[v^{1}V^0(s)-v^0V^{1}(s)]
f(s,\Theta(s),v)dv^{1}dv^{2}dv^{3},
\nonumber\\
\displaystyle I^{+}(s) &=&\int_{\mathbb{R}^{2}}\int_{0}^{\infty}
[v^{1}V^0(s)-v^0V^{1}(s)]
f(s,\Theta(s),v)dv^{1}dv^{2}dv^{3}.
\nonumber
\end{eqnarray}
There are two cases to be distinguished, $V^{1}(s)\geq 0$ and $V^{1}(s)<0$. 
When $V^{1}(s)\geq 0$, $I^{-}$ is non-positive and can be dropped.
Furthermore, for $v^{1}\geq 0$,
\begin{equation*}
\begin{split}
v^{1}V^0-v^0V^{1} = & 
\frac{(v^{1})^{2}(V^0)^{2}-(v^0)^{2}(V^{1})^{2}}{v^{1}V^0+v^0V^{1}}\\
 = & \frac{(v^{1})^{2}(1+(V^{2})^{2}+(V^{3})^{2})}{v^{1}V^0+v^0V^{1}}
-\frac{(V^{1})^{2}(1+(v^{2})^{2}+(v^{3})^{2})}{v^{1}V^0+v^0V^{1}}\\
\leq & \frac{v^{1}(1+(V^{2})^{2}+(V^{3})^{2})}{V^0}.
\end{split}
\end{equation*}
Letting $E_{\vel}=\{(v^{2},v^{3}):|v^{2}|\leq Cs^{-1/2},|v^{3}|\leq Cs^{-1/2}\}$, 
where 
$C$ is the constant appearing in Lemma~\ref{lemma:vtthest}, we obtain
\begin{equation}\label{T1est}
\begin{split}
 T_{1} \leq & 2\|f(t_{1},\cdot)\|_{\infty}s^{1/2}e^{\lambda/2}V^{1}\int_{E_{\vel}}
\int_{0}^{\mQ^{1}}\frac{v^{1}
(1+(V^{2})^{2}+(V^{3})^{2})}{V^{0}}dv^{1}dv^{2}dv^{3}\\
\leq & Cs^{-1/2}e^{\lambda/2}\int_{0}^{\mQ^{1}}v^{1}dv^{1}\frac{V^{1}}{V^{0}}
\leq Cs^{-1/2}e^{\lambda/2}(\mQ^{1})^{2}\frac{V^{1}}{V^{0}},
\end{split}
\end{equation}
where $\mQ^{1}(t)$ is defined as in the statement of the lemma.  
When $V^{1}<0$, an analogous argument can be 
given and it follows that in both cases
\begin{equation}\label{eq:Tofinest}
T_{1}\leq Cs^{-1/2}e^{\lambda/2}(\mQ^{1})^{2}\frac{|V^{1}|}{V^{0}}.
\end{equation}
Combining (\ref{eq:dVosqestpre}), (\ref{eq:Todef}) and 
(\ref{eq:Tofinest}) yields the conclusion of the lemma. 
\end{proof}

\begin{lemma}\label{lemma:Qocoest}
Consider a $\tn{2}$-symmetric solution to the Einstein-Vlasov equations with a 
cosmological constant $\Lambda>0$ and existence interval $(t_{0},\infty)$, 
where $t_{0}\geq 0$. Assume that the solution has $\lambda$-asymptotics and
let $t_{1}=t_{0}+2$. Then there is a constant $C>0$, depending only on the 
solution, such that 
\[
t^{3}\|P_{t}^{2}+\a P_{\theta}^{2}+e^{2P}(Q_{t}^{2}+\a Q_{\theta}^{2})\|_{C^{0}}
+t[\mQ^{1}(t)]^{2}\leq C
\]
for all $t\geq t_{1}$. 
\end{lemma}
\begin{proof}
Let us use (\ref{leg}) to derive an estimate for $F$. Note, to begin
with, that 
\[
-\left(\frac{2}{t}-\frac{\a_{t}}{\a}\right)\leq
-\left(\frac{2}{t}+4t^{1/2}e^{\lambda/2}\Lambda\right)
\leq -\frac{5}{t}+O(t^{-3/2}),
\]
where we have used (\ref{eq:altateq}), (\ref{eq:decct}) and
(\ref{eq:lhpreest}). 
Note also that the second term on the right hand side of (\ref{leg}) can be
written
\[
\frac{1}{2t}(\ma_{\mp}-\ma_{\pm})
+\frac{2}{t}\a(P_{\theta}^{2}
+e^{2P}Q_{\theta}^{2})\leq \frac{1}{t}\ma_{\mp}+\frac{1}{2t}(\ma_{+}+\ma_{-}).
\]
Combining these estimates with Lemma~\ref{lemma:twistest}, we conclude
that 
\begin{equation*}
\begin{split}
\d_{\pm}\ma_{\mp} \leq & -\frac{4}{t}\ma_{\mp}+Ct^{-3/2}\ma_{\mp}
+\frac{1}{2t}(\ma_{+}+\ma_{-})+Ct^{-3}\ma_{\mp}^{1/2}\\
&+2t^{-1/2}e^{\lambda/2}(P_{2}-P_{3})\d_{\mp}P
+4t^{-1/2}e^{\lambda/2}S_{23}e^{P}\d_{\mp}Q.
\end{split}
\end{equation*}
Moreover, due to (\ref{eq:rhoPketcvl}) and Lemma~\ref{lemma:vtthest},
\[
|P_{k}|\leq Ct^{-2}\ln(1+\mQ^{1}),\ \ \
|S_{23}|\leq Ct^{-2}\ln(1+\mQ^{1})
\]
for $k=2,3$. As a consequence, 
\[
\d_{\pm}\ma_{\mp} \leq  -\frac{4}{t}\ma_{\mp}+Ct^{-3/2}\ma_{\mp}
+\frac{1}{2t}(\ma_{+}+\ma_{-})+Ct^{-3}\ma_{\mp}^{1/2}
+Ct^{-4}\ma_{\mp}^{1/2}\ln(1+\mQ^{1}).
\]
Defining 
\begin{equation}\label{eq:hmadef}
\hma_{\pm}=t^{4}\ma_{\pm}+t,
\end{equation}
we obtain
\[
\d_{\pm}\hma_{\mp} \leq \frac{1}{2t}(\hma_{+}+\hma_{-})+Ct^{-3/2}\hma_{\mp}
+Ct^{-2}\hma_{\mp}^{1/2}\ln(1+\mQ^{1}).
\]
Introducing 
\begin{equation}\label{eq:hFdef}
\hF(t)=\sup_{\theta\in \so}\hma_{+}(t,\theta)
+\sup_{\theta\in \so}\hma_{-}(t,\theta),
\end{equation}
we obtain 
\begin{equation}\label{eq:hFestprel}
\hF(t)\leq \hF(t_{1})+\int_{t_{1}}^{t}\left(\frac{1}{s}\hF(s)
+Cs^{-3/2}\hF(s)+Cs^{-2}\hF^{1/2}(s)\ln(1+\mQ^{1})\right)ds.
\end{equation}
Introducing 
\begin{equation}\label{eq:RohQodef}
R^{1}(s)=[s(V^{1}(s))^{2}+1]^{1/2}, \ \ \
\hQ^{1}(s)=[s(\mQ^{1}(s))^{2}+1]^{1/2},
\end{equation}
(\ref{eq:dvodsest}) implies that 
\[
\frac{d(R^{1})^{2}}{ds}\leq Cs^{-2}(\hQ^{1})^{2}+Cs^{-2}\hF+Cs^{-5/2}\hF
+Cs^{-3/2}(\hQ^{1})^{2}.
\]
Integrating this inequality from $t_{1}$ to $t$ and taking the supremum
over initial data belonging to the support of $f$, we obtain 
\begin{equation}\label{eq:hQestprel}
[\hQ^{1}(t)]^{2}\leq [\hQ^{1}(t_{1})]^{2}+\int_{t_{1}}^{t}\left(
Cs^{-2}\hF+Cs^{-3/2}(\hQ^{1})^{2}\right)ds.
\end{equation}
Adding (\ref{eq:hFestprel}) and (\ref{eq:hQestprel}) and introducing 
\begin{equation}\label{eq:mGdef}
\mG=\hF+(\hQ^{1})^{2},
\end{equation}
we obtain 
\[
\mG(t)\leq \mG(t_{1})+\int_{t_{1}}^{t}\left(\frac{1}{s}\mG(s)+Cs^{-3/2}\mG(s)\right)
ds.
\]
In particular, $\mG(t)\leq Ct$, so that $\hF(t)\leq Ct$ and $\mQ^{1}$ is bounded. 
Returning to (\ref{eq:dvodsest}) with this information in mind, we conclude
that 
\[
\frac{d(R^{1})^{2}}{ds}\leq Cs^{-3/2}(\hQ^{1})^{2}.
\]
By arguments similar to ones given above, we conclude that $\hQ^{1}$ is 
bounded. The lemma follows. 
\end{proof}

\section{Intermediate estimates}

Before proceeding, it is useful to collect the estimates that follow from
the above arguments. 

\begin{lemma}\label{lemma:interm}
Consider a $\tn{2}$-symmetric solution to the Einstein-Vlasov equations with a 
cosmological constant $\Lambda>0$ and existence interval $(t_{0},\infty)$, 
where $t_{0}\geq 0$. Assume that the solution has $\lambda$-asymptotics and
let $t_{1}=t_{0}+2$. Then there is a constant $C>0$, depending only on the 
solution, such that 
\begin{eqnarray}
\left\|\lambda(t,\cdot)+3\ln t-2\ln\frac{3}{4\Lambda}\right\|_{C^{0}}
 & \leq &  Ct^{-1},\label{eq:hlimpest}\\
\left\|\frac{\a_{t}}{\a}+\frac{3}{t}\right\|_{C^{0}}
+\left\|\lambda_{t}+\frac{3}{t}\right\|_{C^{0}}  & \leq & Ct^{-2},
\label{eq:atltsupest}\\
\left\|\a^{1/2}\lambda_{\theta}\right\|_{C^{0}} & \leq & Ct^{-2}
\label{eq:lthsupest}
\end{eqnarray}
for all $t\geq t_{1}$. Moreover, 
\begin{eqnarray}
\|J_{t}\|_{C^{0}}+\|K_{t}\|_{C^{0}} & \leq &  Ct^{-2},\label{eq:JtKtest}\\
\left\|\frac{e^{P+\lambda/2}(K-QJ)^{2}}{t^{5/2}}\right\|_{C^{0}}
+\left\|\frac{e^{-P+\lambda/2}J^{2}}{t^{5/2}}\right\|_{C^{0}}
 & \leq & Ct^{-4},\label{eq:setwistest}\\
\|J_{\theta}\|_{C^{0}}+\|K_{\theta}\|_{C^{0}} & \leq &  C,\label{eq:JthKthest}\\
\left\|\d_{\theta}\left(\frac{e^{P+\lambda/2}(K-QJ)^{2}}{t^{5/2}}\right)\right\|_{C^{0}}
+\left\|\d_{\theta}\left(\frac{e^{-P+\lambda/2}J^{2}}{t^{5/2}}\right)\right\|_{C^{0}}
 & \leq & Ct^{-4}\label{eq:dthsetwistest}
\end{eqnarray}
for all $t\geq t_{1}$. Finally, 
\begin{equation}\label{eq:czmat}
\|\rho\|_{C^{0}}+t^{1/2}\|J_{i}\|_{C^{0}}+t\|P_{i}\|_{C^{0}}
+t\|S_{im}\|_{C^{0}}  \leq Ct^{-3/2}
\end{equation}
for all $t\geq t_{1}$.
\end{lemma}
\begin{remark}
Note that as a consequence of (\ref{eq:atltsupest}), 
\[
\left\|\d_{t}(\a^{-1/2}e^{\lambda/2})\right\|_{C^{0}}\leq Ct^{-2}
\]
for all $t\geq t_{1}$. 
\end{remark}
\begin{proof}
Due to (\ref{eq:rhoPketcvl}) and Lemmas~\ref{lemma:vtthest} and 
\ref{lemma:Qocoest}, the estimate (\ref{eq:czmat}) holds. 
Combining Lemmas~\ref{lemma:twistest} and \ref{lemma:Qocoest} with 
(\ref{eq:altlteq}) and (\ref{eq:czmat}), we conclude that 
\[
\d_{t}\hl=\frac{3}{t}-\frac{3}{t}e^{\hl/2}+O(t^{-2}),
\]
where we have used the notation (\ref{eq:hldef}). Combining this observation
with Lemma~\ref{lemma:lasstone}, we conclude that there is a constant $C$
such that 
\[
\d_{t}\hl^{2}\leq -\frac{3}{t}\hl^{2}+Ct^{-2}|\hl|.
\]
Introducing $L=t^{2}\hl^{2}$, we obtain
\[
\d_{t}L\leq-\frac{1}{t}L+\frac{C}{t}L^{1/2}.
\]
In particular, it is clear that $L$ decreases once it exceeds a certain 
value. As a consequence, $L$ is bounded, and we obtain 
(\ref{eq:hlimpest}). Combining (\ref{eq:altlteq}), (\ref{eq:altateq}),
(\ref{eq:hlimpest}) and (\ref{eq:czmat})
and Lemmas~\ref{lemma:twistest} and \ref{lemma:Qocoest}, we then obtain 
(\ref{eq:atltsupest}). As a consequence of this estimate, $t^{3}\a$
converges to a strictly positive function. In particular, there are constants
$C_{i}>0$, $i=1,2$, such that 
\begin{equation}\label{eq:roughalest}
C_{1}\leq t^{3}\a(t,\theta)\leq C_{2}
\end{equation}
for all $t\geq t_{1}$. 

Due to (\ref{eq:altltheq}), 
(\ref{eq:lhpreest}), (\ref{eq:czmat}) and Lemma~\ref{lemma:Qocoest}, we obtain 
(\ref{eq:lthsupest}). 
Returning to (\ref{eq:Jt}) and (\ref{eq:Kt}), keeping (\ref{eq:lhpreest}),
(\ref{eq:PQbdgen}) and (\ref{eq:czmat}) in mind, we conclude that 
(\ref{eq:JtKtest}) holds. As a consequence, it is clear that $J$ and $K$ are 
bounded. Combining this observation with (\ref{eq:lhpreest}) and 
(\ref{eq:PQbdgen}) yields (\ref{eq:setwistest}). Due to (\ref{eq:Jth}), 
(\ref{eq:Kth}), (\ref{eq:lhpreest}),
(\ref{eq:PQbdgen}), (\ref{eq:czmat}) and (\ref{eq:roughalest}), we also obtain 
(\ref{eq:JthKthest}).
Due to (\ref{eq:PQbdgen}), (\ref{eq:roughalest}) and Lemma~\ref{lemma:Qocoest},
we know that $P_{\theta}$ and $Q_{\theta}$ are bounded for $t\geq t_{1}$. Moreover,
$\lambda_{\theta}$ is bounded for $t\geq t_{1}$ due to (\ref{eq:lthsupest}) and  
(\ref{eq:roughalest}). Combining these observations with 
(\ref{eq:hlimpest}), (\ref{eq:JthKthest}) and the fact that $J$, $K$, $Q$ and 
$P$ are bounded, we obtain (\ref{eq:dthsetwistest}).
\end{proof}

\section{Derivatives of the characteristic 
system}\label{section:charsys}

Solutions to the Vlasov equation can be expressed in terms of the initial
datum for the distribution function and appropriate solutions to the 
characteristic system (\ref{eq:dThetads})--(\ref{eq:dVthds}). In order
to see this, let us begin by introducing the notation 
$\Theta,V$ for the solution to (\ref{eq:dThetads})--(\ref{eq:dVthds}) 
corresponding to the initial data
\begin{equation}\label{eq:ThVid}
\Theta(t;t,\theta,v)=\theta,\ \ \ 
V(t;t,\theta,v)=v.
\end{equation}
Here we write $\Theta(s;t,\theta,v)$ and $V(s;t,\theta,v)$ in
order to clarify the dependence on the initial data. In particular, 
$\Theta(s;t,\theta,v)$ and $V(s;t,\theta,v)$, considered as 
functions of the parameter $s$, thus constitute a solution to 
(\ref{eq:dThetads})--(\ref{eq:dVthds}). The purpose of the variables 
$(t,\theta,v)$ appearing after the semi-colon is simply to indicate that 
the relations (\ref{eq:ThVid}) hold. We shall use the notation $d\Theta/ds$
and $dV/ds$ to indicate differentiation with respect to the first variable.
Moreover, we shall take $\d_{t}\Theta$, $\d_{\theta}\Theta$, $\d_{v^{i}}\Theta$
etc. to denote differentiation with respect to the variables appearing 
after the semi-colon. 

Given a fixed $\tau\in (t_{0},\infty)$, where $(t_{0},\infty)$ is the 
existence interval of the solution to the Einstein-Vlasov system under
consideration, we know that 
\begin{equation}\label{eq:fitoid}
f(t,\theta,v)=f[\tau,\Theta(\tau;t,\theta,v),V(\tau;t,\theta,v)].
\end{equation}
Since $f(\tau,\cdot)$ is a smooth function with compact support, it is 
sufficient to estimate the derivatives of solutions to the characteristic 
system in order to estimate the derivatives of $f$. Unfortunately, 
differentiating the characteristic system leads to second order derivatives of 
$P$, $Q$ etc.,
quantities over which we have no control. However, using the ideas 
introduced in \cite{andreasson}, this problem can be circumvented. In fact, 
let $\d$ be a shorthand for $\d_{t}$, $\d_{\theta}$ or $\d_{v^{i}}$ and let
\begin{eqnarray}
\Psi & = & \a^{-1/2}e^{\lambda/2}\d\Theta,\label{eq:Psidef}\\
Z^{1} & = & \d V^{1}+\left[\frac{1}{4}\a^{-1/2}\left(\lambda_{t}
-2\frac{\a_{t}}{\a}-4s^{1/2}e^{\lambda/2}\Lambda\right)V^{0}
-\frac{1}{2}\a^{-1/2}P_{t}V^{0}
\frac{(V^{2})^{2}-(V^{3})^{2}}{(V^{0})^{2}-(V^{1})^{2}}\right.\nonumber\\
& & +\frac{1}{2}P_{\theta}V^{1}
\frac{(V^{2})^{2}-(V^{3})^{2}}{(V^{0})^{2}-(V^{1})^{2}}
-\a^{-1/2}e^{P}Q_{t}\frac{V^{0}V^{2}V^{3}}{(V^{0})^{2}-(V^{1})^{2}}\label{eq:Zodef}\\
& & \left.+e^{P}Q_{\theta}\frac{V^{1}V^{2}V^{3}}{(V^{0})^{2}-(V^{1})^{2}}\right]
\d\Theta,\nonumber\\
Z^{2} & = & \d V^{2}+\frac{1}{2}P_{\theta}V^{2}\d\Theta,\label{eq:Ztdef}\\
Z^{3} & = & \d V^{3}-\left(\frac{1}{2}P_{\theta}V^{3}-e^{P}Q_{\theta}V^{2}
\right)\d\Theta.\label{eq:Zthdef}
\end{eqnarray}
It is then possible to derive an ODE for $(\Psi,Z^{1},Z^{2},Z^{3})$ such that the 
coefficients are controlled due to previous arguments. The definitions 
(\ref{eq:Psidef})--(\ref{eq:Zthdef}) differ slightly from those of 
\cite{andreasson}. The reason for this is that in the present context, it is
not sufficient to know that no second order derivatives of $P$, $Q$ etc. occur;
we need to analyse, in detail, all the terms that appear, and to use the 
resulting system in order to derive specific estimates for $\d\Theta$
and $\d V^{i}$. The relevant result is the following. 

\begin{lemma}\label{lemma:Jacobisys}
Consider a $\tn{2}$-symmetric solution to the Einstein-Vlasov equations with a 
cosmological constant $\Lambda>0$ and existence interval $(t_{0},\infty)$, 
where $t_{0}\geq 0$. Assume that the solution has $\lambda$-asymptotics
and let $t_{1}=t_{0}+2$. Then there is a constant $C>0$, depending only on 
the solution, such that 
\begin{eqnarray}
\frac{dZ^{1}}{ds} & = & -\frac{1}{2s}Z^{1}+c_{1,\theta}\Psi
+c_{1,j}Z^{j},\label{eq:Zomanle}\\
\frac{dZ^{2}}{ds} & = & -\frac{1}{2s}Z^{2}+c_{2,2}Z^{2},\label{eq:Ztmanle}\\
\frac{dZ^{3}}{ds} & = & -\frac{1}{2s}Z^{3}+c_{3,2}Z^{2}+c_{3,3}Z^{3},
\label{eq:Zthmanle}\\
\frac{d\Psi}{ds} & = & c_{\theta,\theta}\Psi+c_{\theta,i}Z^{i},\label{eq:Psimanle}
\end{eqnarray}
where Einstein's summation convention applies to $i$ and $j$, 
\begin{eqnarray*}
|c_{i,j}(s;t,\theta,v)|+|c_{\theta,i}(s;t,\theta,v)|+
|c_{i,\theta}(s;t,\theta,v)|+s^{1/2}|c_{\theta,\theta}(s;t,\theta,v)| 
& \leq & Cs^{-3/2},
\end{eqnarray*}
and the estimate holds for all $(t,\theta,v)\in [t_{1},\infty)\times\so\times
\rn{3}$ in the support of $f$ and for all $s\in [t_{1},t]$. 
\end{lemma}
\begin{proof}
Let us begin by noting that 
\begin{eqnarray*}
e^{P/2}s^{1/2}Z^{2} & = & \d(s^{1/2}e^{P/2}V^{2}),\\
s^{1/2}e^{-P/2}Z^{3}+e^{P/2}Qs^{1/2}Z^{2} & = & 
\d (s^{1/2}Qe^{P/2}V^{2}+s^{1/2}e^{-P/2}V^{3}).
\end{eqnarray*}
Since the quantities appearing in (\ref{eq:v23con}) are preserved along
characteristics, we obtain
\begin{eqnarray}
\frac{d}{ds}\left(e^{P/2}s^{1/2}Z^{2}\right) & = & 0,\label{eq:Ztdot}\\
\frac{d}{ds}\left(s^{1/2}e^{-P/2}Z^{3}+e^{P/2}Qs^{1/2}Z^{2}\right) & = & 0.
\label{eq:Zthdot}
\end{eqnarray}
We also have 
\begin{equation}\label{eq:Psidot}
\frac{d\Psi}{ds}=\frac{1}{2}\left(\lambda_{t}-\frac{\a_{t}}{\a}\right)\Psi
+\frac{1}{2}\a^{1/2}\lambda_{\theta}\frac{V^{1}}{V^{0}}\Psi
+e^{\lambda/2}\d\left(\frac{V^{1}}{V^{0}}\right).
\end{equation}
In the end, we shall express $\d (V^{1}/V^{0})$ in terms of $Z^{i}$ and $\Psi$.
However, there is no immediate gain in doing so here. 
The most cumbersome part of the argument is to compute
the derivative of $Z^{1}$. This calculation can be divided into several parts. 
Let us begin by considering 
\[
\frac{d}{ds}(\d V^{1})=\d\left(\frac{dV^{1}}{ds}\right).
\]
When calculating the right hand side, it is convenient to divide the 
result into terms which include a $\d V^{i}$ factor, $i=1,2,3$,
and terms which do not. Combining Lemmas~\ref{lemma:vtthest}, 
\ref{lemma:Qocoest} and \ref{lemma:interm}, 
the terms which include such a factor can be written
\[
-\frac{1}{2s}\d V^{1}+c_{i}\d V^{i},
\]
where $c_{i}(s)=O(s^{-2})$; note that the factor in front of 
$V^{1}$ in (\ref{eq:dVods}) is given by 
\[
-\frac{1}{2s}+O(s^{-2}).
\] 
It is straightforward to calculate the remaining terms, and we conclude that 
\begin{equation*}
\begin{split}
\frac{d}{ds}(\d V^{1}) = & 
-\frac{1}{4}\d_{\theta}(\a^{1/2}\lambda_{\theta})V^{0}\d\Theta-\frac{1}{4}
\d_{\theta}\left(\lambda_{t}-2\frac{\a_{t}}{\a}\right)V^{1}\d\Theta
+\d_{\theta}(\a^{1/2}e^{P}Q_{\theta})\frac{V^{2}V^{3}}{V^{0}}\d\Theta\\
& -\frac{1}{2}\d_{\theta}(\a^{1/2}P_{\theta})
\frac{(V^{3})^{2}-(V^{2})^{2}}{V^{0}}\d\Theta
+s^{-7/4}\d_{\theta}(e^{\lambda/4}e^{-P/2}J)V^{2}\d\Theta\\
& +s^{-7/4}\d_{\theta}[e^{\lambda/4}e^{P/2}(K-QJ)]V^{3}\d\Theta-\frac{1}{2s}\d V^{1}
+c_{i}\d V^{i},
\end{split}
\end{equation*}
where $c_{i}(s)=O(s^{-2})$. Combining an argument which is identical to the 
proof of (\ref{eq:dthsetwistest}) with Lemma~\ref{lemma:vtthest}, we can 
estimate the 
factors multiplying $\d\Theta$ in the third and fourth last terms on the right
hand side. In fact, we obtain 
\begin{equation*}
\begin{split}
\frac{d}{ds}(\d V^{1}) = & -\frac{1}{2s}\d V^{1}
-\frac{1}{4}\d_{\theta}(\a^{1/2}\lambda_{\theta})V^{0}\d\Theta-\frac{1}{4}
\d_{\theta}\left(\lambda_{t}-2\frac{\a_{t}}{\a}\right)V^{1}\d\Theta
+\d_{\theta}(\a^{1/2}e^{P}Q_{\theta})\frac{V^{2}V^{3}}{V^{0}}\d\Theta\\
& -\frac{1}{2}\d_{\theta}(\a^{1/2}P_{\theta})
\frac{(V^{3})^{2}-(V^{2})^{2}}{V^{0}}\d\Theta+c_{\theta}\d\Theta
+c_{i}\d V^{i},
\end{split}
\end{equation*}
where $c_{i}(s)=O(s^{-2})$ and $c_{\theta}(s)=O(s^{-3})$. As a next
step, it is of interest to consider the terms that arise when $d/ds$
hits a $V^{\a}$ in the second term in the definition of 
$Z^{1}$. Before writing down the result, let us note that, due to 
Lemmas~\ref{lemma:vtthest}, \ref{lemma:Qocoest} and \ref{lemma:interm},
\begin{equation}\label{eq:dVdscrude}
\frac{dV^{i}}{ds}=-\frac{1}{2s}V^{i}+O(s^{-2}),\ \ \
\frac{dV^{0}}{ds}=O(s^{-2}).
\end{equation}
Due to these observations, Lemmas~\ref{lemma:vtthest}, \ref{lemma:Qocoest} and 
\ref{lemma:interm}, as well as (\ref{eq:roughalest}) and the definition of 
$Z^{1}$, we conclude that when $d/ds$ hits a $V^{\a}$ in the second term in
$Z^{1}$, the resulting expression can be written $c_{\theta}\d\Theta$,
where $c_{\theta}(s)=O(s^{-2})$. 

Note that every term appearing in the second
term in the definition of $Z^{1}$ can be written in the form
\begin{equation}\label{eq:hwhatnotterm}
h(\cdot,\Theta)\psi(V)\a^{-1/2}(\cdot,\Theta)\d\Theta.
\end{equation}
We have already estimated the terms that arise when $d/ds$ hits $\psi$. Let
us therefore consider the terms that arise when the derivative hits the 
remaining factors. Omitting the arguments, we need to consider 
\begin{equation}\label{eq:ddsnotpsi}
\begin{split}
 & \left(h_{t}+\a^{1/2}h_{\theta}\frac{V^{1}}{V^{0}}\right)\psi\a^{-1/2}\d\Theta
-\frac{\a_{t}}{2\a^{3/2}}h\psi\d\Theta\\
& +h\psi\left[-\frac{\a_{\theta}}{2\a^{3/2}}\a^{1/2}\frac{V^{1}}{V^{0}}\d\Theta
+\a^{-1/2}\frac{\a_{\theta}}{2\a^{1/2}}\frac{V^{1}}{V^{0}}\d\Theta
+\d\left(\frac{V^{1}}{V^{0}}\right)\right]\\
 = & \left(\d_{t}(\a^{-1/2}h)+h_{\theta}\frac{V^{1}}{V^{0}}\right)\psi\d\Theta
+h\psi\d\left(\frac{V^{1}}{V^{0}}\right).
\end{split}
\end{equation}
In all the terms of interest, $h\psi=O(s^{-2})$, where we have used
Lemmas~\ref{lemma:vtthest}, \ref{lemma:Qocoest} and 
\ref{lemma:interm}. As a consequence, this 
expression can be written 
\[
\left(\d_{t}(\a^{-1/2}h)+h_{\theta}\frac{V^{1}}{V^{0}}\right)\psi\d\Theta
+c_{i}\d V^{i},
\]
where $c_{i}(s)=O(s^{-2})$. Adding up the 
above observations, we conclude that 
\begin{equation}\label{eq:hugedZods}
\begin{split}
\frac{dZ^{1}}{ds} = & -\frac{1}{2s}\d V^{1}
-\frac{1}{4}\d_{\theta}(\a^{1/2}\lambda_{\theta})V^{0}\d\Theta-\frac{1}{4}
\d_{\theta}\left(\lambda_{t}-2\frac{\a_{t}}{\a}\right)V^{1}\d\Theta
+\d_{\theta}(\a^{1/2}e^{P}Q_{\theta})\frac{V^{2}V^{3}}{V^{0}}\d\Theta\\
& -\frac{1}{2}\d_{\theta}(\a^{1/2}P_{\theta})
\frac{(V^{3})^{2}-(V^{2})^{2}}{V^{0}}\d\Theta
+\frac{1}{4}\d_{t}\left[\a^{-1/2}\left(\lambda_{t}
-2\frac{\a_{t}}{\a}-4s^{1/2}e^{\lambda/2}\Lambda\right)\right]V^{0}\d\Theta\\
&+\frac{1}{4}\d_{\theta}\left(\lambda_{t}
-2\frac{\a_{t}}{\a}-4s^{1/2}e^{\lambda/2}\Lambda\right)V^{1}\d\Theta
-\frac{1}{2}\d_{t}(\a^{-1/2}P_{t})V^{0}
\frac{(V^{2})^{2}-(V^{3})^{2}}{(V^{0})^{2}-(V^{1})^{2}}\d\Theta\\
&  -\frac{1}{2}P_{t\theta}V^{1}
\frac{(V^{2})^{2}-(V^{3})^{2}}{(V^{0})^{2}-(V^{1})^{2}}\d\Theta
+\frac{1}{2}P_{t\theta}V^{1}
\frac{(V^{2})^{2}-(V^{3})^{2}}{(V^{0})^{2}-(V^{1})^{2}}\d\Theta\\
 & +\frac{1}{2}\d_{\theta}
(\a^{1/2}P_{\theta})\frac{(V^{1})^{2}}{V^{0}}
\frac{(V^{2})^{2}-(V^{3})^{2}}{(V^{0})^{2}-(V^{1})^{2}}\d\Theta
-\d_{t}(\a^{-1/2}e^{P}Q_{t})\frac{V^{0}V^{2}V^{3}}{(V^{0})^{2}-(V^{1})^{2}}\d\Theta\\
&  
-\d_{\theta}(e^{P}Q_{t})\frac{V^{1}V^{2}V^{3}}{(V^{0})^{2}-(V^{1})^{2}}\d\Theta
+\d_{t}(e^{P}Q_{\theta})\frac{V^{1}V^{2}V^{3}}{(V^{0})^{2}-(V^{1})^{2}}\d\Theta\\
&  +\d_{\theta}(\a^{1/2}e^{P}Q_{\theta})\frac{V^{1}}{V^{0}}
\frac{V^{1}V^{2}V^{3}}{(V^{0})^{2}-(V^{1})^{2}}\d\Theta
+c_{\theta}\d\Theta
+c_{i}\d V^{i},
\end{split}
\end{equation}
where $c_{\theta}(s)=O(s^{-2})$ and $c_{i}(s)=O(s^{-2})$. How to interpret the
terms in this equation should be clear from (\ref{eq:hwhatnotterm}) and
(\ref{eq:ddsnotpsi}). However, there is one term which is slightly 
ambiguous, namely the sixth one on the right hand side of 
(\ref{eq:hugedZods}). For the sake of clarity, let us point out that the 
factor in front of $V^{0}\d\Theta$ in this 
term should be interpreted as the time derivative of 
\[
\frac{1}{4}\left[\a^{-1/2}\left(\lambda_{t}
-2\frac{\a_{t}}{\a}-4t^{1/2}e^{\lambda/2}\Lambda\right)\right],
\]
evaluated at $[s,\Theta(s)]$. The expression (\ref{eq:hugedZods})
can be simplified somewhat. To begin with, the terms involving $P_{t\theta}$
cancel. Moreover, 
\[
-\d_{\theta}(e^{P}Q_{t})+\d_{t}(e^{P}Q_{\theta})=e^{P}(P_{t}Q_{\theta}-P_{\theta}Q_{t}).
\]
Using Lemma~\ref{lemma:Qocoest} and (\ref{eq:roughalest}), we can estimate 
this expression in order
to conclude that the sum of the fourth and fifth last terms is of the 
form $c_{\theta}\d \Theta$, where $c_{\theta}=O(s^{-3})$ (in order to obtain 
this conclusion, we have also used Lemma~\ref{lemma:vtthest}).
Using Lemma~\ref{lemma:Qocoest}, (\ref{eq:hlimpest}), (\ref{eq:lthsupest}) and 
(\ref{eq:roughalest}), the sum of the third 
and the seventh term can be written $c_{\theta}\d \Theta$, where 
$c_{\theta}=O(s^{-2})$. Since 
\begin{eqnarray*}
-\frac{(V^{3})^{2}-(V^{2})^{2}}{V^{0}}+\frac{(V^{1})^{2}}{V^{0}}
\frac{(V^{2})^{2}-(V^{3})^{2}}{(V^{0})^{2}-(V^{1})^{2}} & = & 
V^{0}\frac{(V^{2})^{2}-(V^{3})^{2}}{(V^{0})^{2}-(V^{1})^{2}}\\
\frac{V^{2}V^{3}}{V^{0}}+\frac{V^{1}}{V^{0}}
\frac{V^{1}V^{2}V^{3}}{(V^{0})^{2}-(V^{1})^{2}} & = & 
\frac{V^{0}V^{2}V^{3}}{(V^{0})^{2}-(V^{1})^{2}},
\end{eqnarray*}
the terms involving a factor of $\d_{\theta}(\a^{1/2}P_{\theta})$
can be written
\[
\frac{1}{2}\d_{\theta}(\a^{1/2}P_{\theta})
V^{0}\frac{(V^{2})^{2}-(V^{3})^{2}}{(V^{0})^{2}-(V^{1})^{2}}\d\Theta,
\]
and the terms involving a factor of $\d_{\theta}(\a^{1/2}e^{P}Q_{\theta})$
can be written
\[
\d_{\theta}(\a^{1/2}e^{P}Q_{\theta})
\frac{V^{0}V^{2}V^{3}}{(V^{0})^{2}-(V^{1})^{2}}\d\Theta.
\]
Combining these observations yields
\begin{equation}\label{eq:dZoexinte}
\begin{split}
\frac{dZ^{1}}{ds} = & -\frac{1}{2s}\d V^{1}
+\frac{1}{4}\d_{t}\left[\a^{-1/2}\left(\lambda_{t}
-2\frac{\a_{t}}{\a}-4s^{1/2}e^{\lambda/2}\Lambda\right)\right]V^{0}\d\Theta
-\frac{1}{4}\d_{\theta}(\a^{1/2}\lambda_{\theta})V^{0}\d\Theta\\
&-\frac{1}{2}[\d_{t}(\a^{-1/2}P_{t})-\d_{\theta}(\a^{1/2}P_{\theta})]V^{0}
\frac{(V^{2})^{2}-(V^{3})^{2}}{(V^{0})^{2}-(V^{1})^{2}}\d\Theta\\
 & -[\d_{t}(\a^{-1/2}e^{P}Q_{t})-\d_{\theta}(\a^{1/2}e^{P}Q_{\theta})]
\frac{V^{0}V^{2}V^{3}}{(V^{0})^{2}-(V^{1})^{2}}\d\Theta
+c_{\theta}\d\Theta
+c_{i}\d V^{i},
\end{split}
\end{equation}
where $c_{\theta}(s)=O(s^{-2})$ and $c_{i}(s)=O(s^{-2})$. Combining
(\ref{eq:ttmthth}), (\ref{eq:roughalest}) and Lemmas~\ref{lemma:Qocoest}
and \ref{lemma:interm} yields
\[
\d_{t}(\a^{-1/2}P_{t})-\d_{\theta}(\a^{1/2}P_{\theta})=O(t^{-1}).
\]
Due to this estimate and Lemmas~\ref{lemma:Qocoest}
and \ref{lemma:vtthest}, the fourth term on the right hand side of 
(\ref{eq:dZoexinte}) can be written $c_{\theta}\d\Theta$, where 
$c_{\theta}(s)=O(s^{-2})$. Keeping (\ref{eq:Etth}) in mind, a similar 
argument yields the same conclusion concerning the fifth term on the right 
hand side of (\ref{eq:dZoexinte}). Finally, keeping (\ref{eq:denden}) in
mind, a similar argument yields the conclusion that the combination of 
the second and third terms on the right hand side of (\ref{eq:dZoexinte})
can be written $c_{\theta}\d\Theta$, where 
$c_{\theta}(s)=O(s^{-3/2})$. To conclude,
\begin{equation}\label{eq:dZodsalth}
\frac{dZ^{1}}{ds}=-\frac{1}{2s}\d V^{1}+c_{\theta}\d\Theta
+c_{i}\d V^{i},
\end{equation}
where $c_{\theta}(s)=O(s^{-3/2})$ and $c_{i}(s)=O(s^{-2})$. On the other hand,
due to (\ref{eq:hlimpest}) and (\ref{eq:roughalest}), the function
$\a^{-1/2}e^{\lambda/2}$ can be bounded from above and below by positive
constants (for $t\geq t_{1}$). In other words, $\d\Theta$ and $\Psi$
are interchangeable when deriving equations of the form
(\ref{eq:Zomanle})--(\ref{eq:Psimanle}). Moreover,
due to (\ref{eq:roughalest}) and Lemmas~\ref{lemma:vtthest}, 
\ref{lemma:Qocoest} and \ref{lemma:interm},
\begin{equation}\label{eq:ZiVirel}
Z^{i}=\d V^{i}+c_{i,\theta}\Psi,
\end{equation}
where $c_{i,\theta}(s)=O(s^{-1/2})$. We are now in a position to prove the 
lemma. To begin with, combining (\ref{eq:Ztdot}), (\ref{eq:Zthdot})
and Lemma~\ref{lemma:Qocoest} yields (\ref{eq:Ztmanle}) and 
(\ref{eq:Zthmanle}). Combining (\ref{eq:Psidot}) and (\ref{eq:ZiVirel}) with 
Lemmas~\ref{lemma:vtthest}, \ref{lemma:Qocoest} and \ref{lemma:interm}
yields (\ref{eq:Psimanle}). Finally, combining (\ref{eq:dZodsalth}) and 
(\ref{eq:ZiVirel}) yields (\ref{eq:Zomanle}).
\end{proof}

\begin{lemma}\label{leqmma:dercs}
Consider a $\tn{2}$-symmetric solution to the Einstein-Vlasov equations with a 
cosmological constant $\Lambda>0$ and existence interval $(t_{0},\infty)$, 
where $t_{0}\geq 0$. Assume that the solution has $\lambda$-asymptotics and 
let $t_{1}=t_{0}+2$. Then there is a constant 
$C>0$, depending only on the solution, such that 
\begin{equation}\label{eq:dth}
(\ln s)^{2}\left|\d_{\theta}\Theta(s;t,\theta,v)\right|
+s^{1/2}\left|\d_{\theta}V^{i}(s;t,\theta,v)\right|  \leq
C(\ln t)^{2}
\end{equation}
for all $s\in [t_{1},t]$ and $(t,\theta,v)\in [t_{1},\infty)\times\so\times
\rn{3}$ in the support of $f$.
\end{lemma}
\begin{remark}
It is possible to use arguments similar to the ones given below in order to 
derive estimates for $\d_{t}\Theta$, $\d_{v^{i}}\Theta$ etc. 
\end{remark}
\begin{proof}
Let
\begin{equation}\label{eq:hZhPsdef}
\hZ^{i}(s;t,\theta,v)=s^{1/2}Z^{i}(s;t,\theta,v),\ \ \
\hPs(s;t,\theta,v)=(\ln s)^{2}\Psi(s;t,\theta,v).
\end{equation}
Then, due to Lemma~\ref{lemma:Jacobisys}, 
\begin{eqnarray*}
\frac{d\hZ^{1}}{ds} & = & c_{1,\theta}s^{1/2}(\ln s)^{-2}\hPs
+c_{1,j}\hZ^{j},\\
\frac{d\hZ^{2}}{ds} & = & c_{2,2}\hZ^{2},\\
\frac{d\hZ^{3}}{ds} & = & c_{3,2}\hZ^{2}+c_{3,3}\hZ^{3},\\
\frac{d\hPs}{ds} & = & \frac{2}{s\ln s}\hPs
+c_{\theta,\theta}\hPs+c_{\theta,i}s^{-1/2}(\ln s)^{2}\hZ^{i}
\end{eqnarray*}
for $s\in [t_{1},t]$, 
with coefficients as in Lemma~\ref{lemma:Jacobisys}. Introducing 
\begin{equation}\label{eq:hEdef}
\hE=\sum_{i=1}^{3}(\hZ^{i})^{2}+(\hPs)^{2},
\end{equation}
we conclude that there is a constant $C>0$, depending only on 
the solution, such that 
\[
\frac{d\hE}{ds}\geq -C\frac{1}{s(\ln s)^{2}}\hE
\]
for $s\in [t_{1},t]$. As a consequence, 
\begin{equation}\label{eq:mainhEest}
\hE(s;t,\theta,v)\leq C\hE(t;t,\theta,v)
\end{equation}
for $s\in [t_{1},t]$. 
Let us now assume $\d=\d_{\theta}$. Then
\[
\hPs(t;t,\theta,v)=O[(\ln t)^{2}].
\]
Moreover, 
\[
\hZ^{i}(t;t,\theta,v)=[t^{1/2}\d_{\theta}V^{i}+O(1)\Psi](t;t,\theta,v)=O(1).
\]
As a consequence, $\hE(t;t,\theta,v)=O[(\ln t)^{4}]$. Thus
(\ref{eq:mainhEest}) implies (\ref{eq:dth}); note that the 
estimate for $\d_{\theta}\Theta$ is immediate and that 
\[
|\d_{\theta}V^{i}|\leq |Z^{i}|+Cs^{-1/2}|\d_{\theta}\Theta|.
\] 
The lemma follows. 
\end{proof}

\section{Higher order light cone estimates}\label{section:holce}


Before proceeding to the higher order light cone estimates, let us
record some consequences of the estimates obtained in Lemma~\ref{leqmma:dercs}.

\begin{lemma}\label{leqmma:dermq}
Consider a $\tn{2}$-symmetric solution to the Einstein-Vlasov equations with a 
cosmological constant $\Lambda>0$ and existence interval $(t_{0},\infty)$, 
where $t_{0}\geq 0$. Assume that the solution has $\lambda$-asymptotics and
let $t_{1}=t_{0}+2$. Then there is a constant $C>0$, depending only on the 
solution, such that
\begin{eqnarray}
\|\d_{\theta}\rho\|_{C^{0}}+t^{1/2}\|\d_{\theta}J_{i}\|_{C^{0}}
+t\|\d_{\theta}S_{ij}\|_{C^{0}}+t\|\d_{\theta}P_{k}\|_{C^{0}}
& \leq & Ct^{-3/2}(\ln t)^{2},\label{eq:math}\\
\left\|\d_{\theta}\left(\frac{\a_{t}}{\a}\right)\right\|_{C^{0}}
 & \leq & Ct^{-3/2},\label{eq:dtthlna}\\
\left\|\frac{\a_{\theta}}{\a}\right\|_{C^{0}}
 & \leq & C\label{eq:dthlna}
\end{eqnarray}
for all $t\geq t_{1}$.
\end{lemma}
\begin{remark}
Note that $\d_{\theta}(\a_{t}/\a)=\d_{t}(\a_{\theta}/\a)=\d_{t}\d_{\theta}\ln\a$.
\end{remark}
\begin{proof}
The estimate (\ref{eq:math}) follows from the fact that
(\ref{eq:fitoid}) and (\ref{eq:dth}) hold and the fact that
$|v^{i}|\leq Ct^{-1/2}$ in the support of $f(t,\cdot)$. Consider
(\ref{eq:altateq}). Since Lemmas~\ref{lemma:Qocoest} and \ref{lemma:interm}
together with (\ref{eq:roughalest}) imply that $J$, $K$, $Q$ and $P$ are 
bounded in $C^{1}$ and that $\lambda_{\theta}$ is $O(t^{-1/2})$, the first two 
terms on the right hand side of (\ref{eq:altateq}) are $O(t^{-4})$ in $C^{1}$. 
Since $\lambda_{\theta}$
is $O(t^{-1/2})$, the $\theta$-derivative of the third term on the right hand
side of (\ref{eq:altateq}) is $O(t^{-3/2})$. Due to (\ref{eq:math}), the
$\theta$-derivative of the last term is better. Thus (\ref{eq:dtthlna})
holds, so that 
\[
\left\|\d_{t}\left(\frac{\a_{\theta}}{\a}\right)\right\|_{C^{0}}
\leq\left\|\d_{\theta}\left(\frac{\a_{t}}{\a}\right)\right\|_{C^{0}}\leq Ct^{-3/2}.
\]
Integrating this estimate yields (\ref{eq:dthlna}).  
\end{proof}

In what follows, we shall proceed inductively in order to derive estimates for
higher order derivatives. Let us therefore assume that we have a 
$\tn{2}$-symmetric solution to the Einstein-Vlasov equations with a 
cosmological constant $\Lambda>0$ and existence interval $(t_{0},\infty)$, 
where $t_{0}\geq 0$. Assume, moreover, that the solution has 
$\lambda$-asymptotics and let $t_{1}=t_{0}+2$. Let us make the following 
inductive assumption. 

\begin{ind}\label{ind:step2}
For some $1\leq N\in\zo$, there are constants $0\leq m_{j}\in\zo$ and $C_{j}$, 
$j=1,\dots,N$, (depending only on $N$ and the solution) such that
\begin{eqnarray}
s^{1/2}\left|\frac{\d^{j}V}{\d\theta^{j}}(s;t,\theta,v)\right|+
\left|\frac{\d^{j}\Theta}{\d\theta^{j}}(s;t,\theta,v)\right| & \leq &
C_{j}(\ln t)^{m_{j}},\label{eq:ThetaVinds2}\\
\|P_{\theta}\|_{C^{N-1}}+\|Q_{\theta}\|_{C^{N-1}}+t^{3/2}\|P_{t}\|_{C^{N-1}}
+t^{3/2}\|Q_{t}\|_{C^{N-1}} & \leq & C_{N-1},\label{eq:PQinds2}
\end{eqnarray}
for all $j=1,\dots,N$, 
$(t,\theta,v)\in [t_{1},\infty)\times\so\times\rn{3}$ in the support of
$f$ and $s\in [t_{1},t]$.
\end{ind}

\begin{remarks}
The induction hypothesis holds for $N=1$. In what follows, $C_{j}$ and
$m_{j}$ will change from line to line. However, they are only allowed to
depend on $N$ and the solution.
\end{remarks}
In this section we prove that, given that Inductive assumption~\ref{ind:step2}
holds, then (\ref{eq:PQinds2}) holds with $N$ replaced 
by $N+1$. In the next section, we close the induction argument by proving 
that (\ref{eq:ThetaVinds2}) holds with $j$ replaced by $N+1$.

We shall need the following consequences of the inductive assumption.

\begin{lemma}\label{lemma:indhcons2}
Consider a $\tn{2}$-symmetric solution to the Einstein-Vlasov equations with a 
cosmological constant $\Lambda>0$ and existence interval $(t_{0},\infty)$, 
where $t_{0}\geq 0$. Assume that the solution has $\lambda$-asymptotics and
let $t_{1}=t_{0}+2$. Assume, moreover, that 
Inductive assumption~\ref{ind:step2} holds for some $1\leq N\in\zo$. Then
there are constants $C_{j}$, $j=0,\dots,N$, and $m_{N}$, depending only on $N$
and the solution, such that
\begin{eqnarray}
\|\rho\|_{C^{N}}+t^{1/2}\|J_{i}\|_{C^{N}}+t\|P_{i}\|_{C^{N}}
+t\|S_{im}\|_{C^{N}} & \leq & C_{N}t^{-3/2}(\ln t)^{m_{N}},\label{eq:mathetas2}\\
t^{1/2}\|\d^{l+1}_{\theta}\lambda\|_{C^{0}} & \leq & C_{l},
\label{eq:laindfins2}\\
\left\|\a^{-1}\d^{j}_{\theta}\a\right\|_{C^{0}}& \leq & C_{j},
\label{eq:alindfins2}\\
\|J_{\theta}\|_{C^{N}}+\|K_{\theta}\|_{C^{N}} & \leq & C_{N}
(\ln t)^{m_{N}},\label{eq:JKderests2}\\
\left\|\d_{\theta}^{l+1}\left(\frac{\a_{t}}{\a}\right)\right\|_{C^{0}}
 & \leq & C_{l}t^{-3/2}
\label{eq:ataltests2}
\end{eqnarray}
for $t\geq t_{1}$, $0\leq j\leq N$, $0\leq l\leq N-1$ and $i,m=1,2,3$.
\end{lemma}
\begin{proof}
For $N=1$, the conclusions follow from Lemmas~\ref{lemma:Qocoest}, 
\ref{lemma:interm} and \ref{leqmma:dermq}, (\ref{eq:roughalest})
and the equations (\ref{eq:Jth}) and (\ref{eq:Kth}).
We may thus, without loss of generality, assume that $N\geq 2$. 
An immediate consequence of the inductive assumption is that, for
$0\leq j\leq N$ and $t\geq t_{1}$,
\[
\left|\frac{\d^{j}f}{\d\theta^{j}}\right|\leq C_{j}(\ln t)^{m_{j}};
\]
cf. (\ref{eq:fitoid}) and (\ref{eq:ThetaVinds2}). As a consequence of this
estimate, we obtain (\ref{eq:mathetas2}).
In order to obtain control of the $\theta$-derivatives
of $\a$ and $\lambda$, we need to proceed inductively. Let us make the
inductive assumption that
\begin{eqnarray}
\left\|\a^{-1}\d^{j}_{\theta}\a\right\|_{C^{0}} & \leq & C_{j},
\label{eq:alind}\\
\|\d^{j}_{\theta}\lambda\|_{C^{0}} & \leq & C_{j}t^{-1/2}\label{eq:laind}
\end{eqnarray}
for $1\leq j\leq l<N$. Note that we know the inductive hypothesis to be
true for $l=1$. Differentiating (\ref{eq:altltheq}) $l$ times with respect
to $\theta$ and appealing to (\ref{eq:PQinds2}), (\ref{eq:mathetas2}),
(\ref{eq:alind}) and (\ref{eq:laind}), we conclude that (\ref{eq:laind}) holds 
with $j$ replaced by $l+1$. 
In order to improve our knowledge concerning $\a$, let us begin by 
improving our estimates for the $\theta$-derivatives for $J$ and $K$.
Differentiating (\ref{eq:Jth}) and (\ref{eq:Kth}) $0\leq j\leq l$ times and
using (\ref{eq:PQinds2}), (\ref{eq:mathetas2}), (\ref{eq:alind}) and 
(\ref{eq:laind}), we conclude that
\begin{equation}\label{eq:dthjpoJest}
\|\d_{\theta}^{j+1}J\|_{C^{0}}+
\|\d_{\theta}^{j+1}K\|_{C^{0}}\leq C_{j}(\ln t)^{m_{j}}
\end{equation}
for $t\geq t_{1}$ and $0\leq j\leq l$. 
Differentiating (\ref{eq:altateq}) $l+1$ times with respect to $\theta$, using 
(\ref{eq:PQinds2}), (\ref{eq:mathetas2}), 
(\ref{eq:dthjpoJest}) as well as the fact that (\ref{eq:laind}) holds
for $1\leq j\leq l+1$, we obtain
\begin{equation}\label{eq:dtheldaqu}
\left\|\d_{\theta}^{l+1}\left(\frac{\a_{t}}{\a}\right)\right\|_{C^{0}}
=\left\|\d_{t}\d_{\theta}^{l}\left(\frac{\a_{\theta}}{\a}\right)\right\|_{C^{0}}
\leq C_{l}t^{-3/2}
\end{equation}
for $t\geq t_{1}$. Thus
\[
\left\|\d_{\theta}^{l}\left(\frac{\a_{\theta}}{\a}\right)\right\|_{C^{0}}
\leq C_{l}
\]
for $t\geq t_{1}$. Combining this estimate with the inductive hypothesis,
we conclude that (\ref{eq:alind}) holds with $j$ replaced by $l+1$. 
Thus (\ref{eq:laind}) and (\ref{eq:alind}) hold for $1\leq j\leq N$. 
We thus conclude that (\ref{eq:laindfins2}), (\ref{eq:alindfins2}) and (\ref{eq:JKderests2}) hold.
In addition, (\ref{eq:dtheldaqu}) implies that (\ref{eq:ataltests2}) holds.
\end{proof}
We are now in a position to derive higher order light cone estimates.

\begin{lemma}\label{lemma:indstep1}
Consider a $\tn{2}$-symmetric solution to the Einstein-Vlasov equations with a 
cosmological constant $\Lambda>0$ and existence interval $(t_{0},\infty)$, 
where $t_{0}\geq 0$. Assume that the solution has $\lambda$-asymptotics and
let $t_{1}=t_{0}+2$. Assume, moreover, that 
Inductive assumption~\ref{ind:step2} holds for some $1\leq N\in\zo$.
Then there is a constant $C_{N}>0$, depending only on $N$ and the 
solution, such that
\begin{equation}\label{eq:highorder}
t^{3/2}\|\d_{\theta}^N P_{t}\|_{C^{0}}+\|\d_{\theta}^N P_{\theta}\|_{C^{0}}
+t^{3/2}\|e^{P}\d_{\theta}^N Q_{t}\|_{C^{0}}+\|e^{P}\d_{\theta}^NQ_{\theta}\|_{C^{0}}
\leq C_{N}
\end{equation}
for all $t\geq t_{1}$. As a consequence, (\ref{eq:PQinds2}) holds with $N$ 
replaced by $N+1$.
\end{lemma}
\begin{proof}
Let us begin by pointing out that if $N=1$, then the lower bound is larger 
than the upper bound in some of the sums below. In that case, the 
corresponding sum should be equated with zero. Moreover, terms which are 
bounded by $Ct^{-3}$ for $t\geq t_{1}$ will sometimes be written $O(t^{-3})$. 
Let us compute
\begin{equation*}
\begin{split}
\d_{\pm}[\d_{\theta}^{N}P_{t}\mp\d_{\theta}^{N}(\a^{1/2}P_{\theta})] = &
\d_{\theta}^N P_{tt}\mp\d_{\theta}^N\left(\frac{\a_{t}}{2\a}\a^{1/2}P_{\theta}+\a^{1/2}P_{t\theta}
\right)\pm\a^{1/2}\d_{\theta}^{N+1}P_{t}\\
 & -\a^{1/2}\d_{\theta}^{N+1}(\a^{1/2}P_{\theta})\\
 = & \d_{\theta}^N P_{tt}\mp\frac{1}{2}\sum_{j=0}^{N-1}\beta_j\d_{\theta}^{N-j}\left(\frac{\a_{t}}{\a}\right)
\d_{\theta}^{j}(\a^{1/2}P_{\theta})\mp\frac{1}{2}\frac{\a_{t}}{\a}
\d_{\theta}^N(\a^{1/2}P_{\theta})\\
& \mp\frac{N\a_{\theta}}{2\a^{1/2}}\d_{\theta}^NP_t \mp \sum_{j=0}^{N-2}\beta_j\d_{\theta}^{N-j}(\a^{1/2})\d_{\theta}^{j+1}(P_{t})-\d_{\theta}^{N}[\a^{1/2}\d_{\theta}(\a^{1/2}P_{\theta})]\\
& + \frac{N\a_{\theta}}{2\a^{1/2}}\d_{\theta}^N(\a^{1/2}P_{\theta})
+\sum_{j=0}^{N-2}\beta_j\d_{\theta}^{N-j}(\a^{1/2})\d_{\theta}^{j+1}(\a^{1/2}P_{\theta}),
\end{split}
\end{equation*}
where the $\beta_j$ are binomial coefficients. Note that all the sums
are $O(t^{-3})$ due to Inductive assumption~\ref{ind:step2}, Lemma~\ref{lemma:indhcons2}
and (\ref{eq:roughalest}). Let us use (\ref{eq:ttmthth}) in order to compute
\begin{equation*}
\begin{split}
\d_{\theta}^{N}[P_{tt}-\a^{1/2}\d_{\theta}(\a^{1/2}P_{\theta})] = &
\d_{\theta}^{N}\left(P_{tt}-\a P_{\theta\theta}-\frac{\a_{\theta}}{2}P_{\theta}\right)\\
 = & -\frac{1}{t}\d_{\theta}^NP_{t}+\frac{\a_{t}}{2\a}\d_{\theta}^N P_{t}
 +\sum_{j=0}^{N-1}\beta_j\d_{\theta}^{N-j}\left(\frac{\a_{t}}{2\a}\right)\d_{\theta}^j P_{t}\\
& +\sum_{j=0}^{N-1}\beta_j \d_{\theta}^{N-j}(e^{2P})\d_{\theta}^j(Q_{t}^{2}-\a Q_{\theta}^{2})\\
 & +e^{2P}\sum_{j=1}^{N-1}\beta_j(\d_{\theta}^{N-j}(Q_t-\alpha^{1/2}Q_{\theta}))(\d_{\theta}^j(Q_t+\alpha^{1/2}Q_{\theta}))\\
 & +
 2e^{2P}[Q_{t}\d_{\theta}^NQ_{t}-\a^{1/2}Q_{\theta}\d_{\theta}^N(\a^{1/2}Q_{\theta})]
-\d_{\theta}^N\left(\frac{e^{P+\lambda/2}(K-QJ)^{2}}{2t^{7/2}}\right)\\
 & +\d_{\theta}^N\left(\frac{e^{-P+\lambda/2}J^{2}}{2t^{7/2}}\right)
+t^{-1/2}\d_{\theta}^N[e^{\lambda/2}(P_{2}-P_{3})].
\end{split}
\end{equation*}
Due to Inductive assumption~\ref{ind:step2}, Lemmas~\ref{lemma:indhcons2} and \ref{lemma:interm},
and (\ref{eq:roughalest}), the sums are $O(t^{-3})$, as well as the last three terms on the
right hand side. We thus obtain
\begin{equation*}
\begin{split}
\d_{\pm}[\d_{\theta}^N P_{t}\mp\d_{\theta}^N(\a^{1/2}P_{\theta})] = &
-\frac{1}{t}\d_{\theta}^N P_{t}+\frac{\a_{t}}{2\a}
[\d_{\theta}^N P_{t}\mp \d_{\theta}^N(\a^{1/2}P_{\theta})]
\mp \frac{N\a_{\theta}}{2\a^{1/2}}[\d_{\theta}^NP_{t}\mp
\d_{\theta}^N(\a^{1/2}P_{\theta})]\\
 & +2e^{2P}[Q_{t}\d_{\theta}^N Q_{t}-\a^{1/2}Q_{\theta}\d_{\theta}^N(\a^{1/2}Q_{\theta})]
+O(t^{-3}).
\end{split}
\end{equation*}
Introducing
\begin{equation}\label{eq:manpopm}
\ma_{N+1,\pm}=[\d_{\theta}^N P_{t}\pm\d_{\theta}^N(\a^{1/2}P_{\theta})]^{2}
+e^{2P}[\d_{\theta}^N Q_{t}\pm\d_{\theta}^N(\a^{1/2}Q_{\theta})]^{2},
\end{equation}
we conclude that
\begin{equation*}
\begin{split}
\d_{\pm}[\d_{\theta}^N P_{t}\mp\d_{\theta}^N (\a^{1/2}P_{\theta})]^{2} \leq &
-\frac{5}{t}[\d_{\theta}^N P_{t}\mp\d_{\theta}^N (\a^{1/2}P_{\theta})]^{2}
\mp\frac{2}{t}\d_{\theta}^N(\a^{1/2}P_{\theta})
[\d_{\theta}^N P_{t}\mp\d_{\theta}^N (\a^{1/2}P_{\theta})]\\
 & +C_{N}t^{-3}\ma_{N+1,\mp}^{1/2}+C_{N}t^{-3/2}(\ma_{N+1,+}+\ma_{N+1,-}),
\end{split}
\end{equation*}
where we have used Inductive assumption~\ref{ind:step2}, Lemma~\ref{lemma:indhcons2}, 
(\ref{eq:atltsupest}) and (\ref{eq:roughalest}). Let us now consider
\begin{equation*}
\begin{split}
\d_{\pm}[\d_{\theta}^NQ_{t}\mp\d_{\theta}^N(\a^{1/2}Q_{\theta})] = &
\d_{\theta}^N Q_{tt}\mp\frac{1}{2}\sum_{j=0}^{N-1}\beta_j \d_{\theta}^{N-j}\left(\frac{\a_{t}}{\a}\right)
\d_{\theta}^j(\a^{1/2}Q_{\theta})\mp\frac{1}{2}\frac{\a_{t}}{\a}
\d_{\theta}^N(\a^{1/2}Q_{\theta})\\
&\mp\frac{N\a_{\theta}}{2\a^{1/2}}\d_{\theta}^NQ_t \mp\sum_{j=0}^{N-2}\beta_j \d_{\theta}^{N-j}(\a^{1/2})\d_{\theta}^jQ_{t\theta}
 -\d_{\theta}^N[\a^{1/2}\d_{\theta}(\a^{1/2}Q_{\theta})]\\
&+\frac{N\a_{\theta}}{2\a^{1/2}}\d_{\theta}^N(\a^{1/2}Q_{\theta}) 
+\sum_{j=0}^{N-2}\beta_j\d_{\theta}^{N-j}(\a^{1/2}) \d_{\theta}^{j+1}(\a^{1/2}Q_{\theta}).
\end{split}
\end{equation*}
As above, all the sums are $O(t^{-3})$ due to Inductive assumption~\ref{ind:step2}, 
Lemma~\ref{lemma:indhcons2} and (\ref{eq:roughalest}). Compute, using (\ref{eq:Etth}),
\begin{equation*}
\begin{split}
\d_{\theta}^N[Q_{tt}-\a^{1/2}\d_{\theta}(\a^{1/2}Q_{\theta})] = &
-\frac{1}{t}\d_{\theta}^N Q_{t}+\frac{\a_{t}}{2\a}\d_{\theta}^N Q_{t}
+\sum_{j=0}^{N-1}\beta_j \d_{\theta}^{N-j}\left(\frac{\a_{t}}{2\a}\right)\d_{\theta}^jQ_{t}
-2(\d_{\theta}^NP_{t})Q_{t}\\
& -2P_{t}(\d_{\theta}^N Q_{t})
  +2\d_{\theta}^N(\a^{1/2} P_{\theta})
\a^{1/2}Q_{\theta}+2\a^{1/2}P_{\theta}\d_{\theta}^N(\a^{1/2} Q_{\theta})\\
&-2\sum_{j=1}^{N-1}\beta_j(\d_{\theta}^{N-j}P_t)(\d_{\theta}^jQ_t)+
2\sum_{j=1}^{N-1}\beta_j(\d_{\theta}^{N-j}(\a^{1/2}P_{\theta}))(\d_{\theta}^j(\a^{1/2}Q_{\theta}))\\
& +\d_{\theta}^N\left(\frac{e^{\lambda/2-P}J(K-QJ)}{t^{7/2}}\right)
+2t^{-1/2}\d_{\theta}^N(e^{\lambda/2-P}S_{23}).
\end{split}
\end{equation*}
Due to Inductive assumption~\ref{ind:step2}, Lemmas~\ref{lemma:indhcons2} and \ref{lemma:interm},
and (\ref{eq:roughalest}), the sums are $O(t^{-3})$, as well as the last two terms on the
right hand side. Thus
\begin{equation*}
\begin{split}
\d_{\pm}[\d_{\theta}^NQ_{t}\mp\d_{\theta}^N(\a^{1/2}Q_{\theta})] = &
-\frac{1}{t}\d_{\theta}^N Q_{t}+\frac{\a_{t}}{2\a}
[\d_{\theta}^NQ_{t}\mp \d_{\theta}^N(\a^{1/2}Q_{\theta})]
\mp\frac{N\a_{\theta}}{2\a^{1/2}}[\d_{\theta}^NQ_{t}\mp \d_{\theta}^N(\a^{1/2}Q_{\theta})]\\
& -2(\d_{\theta}^NP_{t})Q_{t} -2P_{t}(\d_{\theta}^N Q_{t})
  +2\d_{\theta}^N(\a^{1/2} P_{\theta})
\a^{1/2}Q_{\theta}\\
& +2\a^{1/2}P_{\theta}\d_{\theta}^N(\a^{1/2} Q_{\theta})
+O(t^{-3}).
\end{split}
\end{equation*}
Consequently,
\begin{equation*}
\begin{split}
\d_{\pm}[\d_{\theta}^N Q_{t}\mp\d_{\theta}^N(\a^{1/2}Q_{\theta})]^{2} \leq &
-\frac{5}{t}[\d_{\theta}^N Q_{t}\mp\d_{\theta}^N(\a^{1/2}Q_{\theta})]^{2}
\mp\frac{2}{t}\d_{\theta}^N(\a^{1/2}Q_{\theta})
[\d_{\theta}^N Q_{t}\mp\d_{\theta}^N (\a^{1/2}Q_{\theta})]\\
 & +C_{N}t^{-3}\ma_{N+1,\mp}^{1/2}+C_{N}t^{-3/2}(\ma_{N+1,+}+\ma_{N+1,-}),
\end{split}
\end{equation*}
where we have used Inductive assumption~\ref{ind:step2}, Lemma~\ref{lemma:indhcons2}, 
(\ref{eq:atltsupest}) and (\ref{eq:roughalest}).
Adding up the above estimates, we conclude that
\begin{equation*}
\begin{split}
\d_{\pm}\ma_{N+1,\mp} \leq & -\frac{5}{t}\ma_{N+1,\mp}+\frac{1}{2t}(\ma_{N+1,\mp}
-\ma_{N+1,\pm})+\frac{1}{t}(\ma_{N+1,+}+\ma_{N+1,-})\\
 & +C_{N}t^{-3/2}(\ma_{N+1,+}+\ma_{N+1,-})+C_{N}t^{-3}\ma_{N+1,\mp}^{1/2}.
\end{split}
\end{equation*}
Let us introduce
\begin{equation}\label{eq:hmaNdefhFNdef}
\hma_{N+1,\pm}=t^{7/2}\ma_{N+1,\pm}+t^{1/2},\ \ \
\hF_{N+1,\pm}=\sup_{\theta\in\so}\hma_{N+1,\pm},\ \ \
\hF_{N+1}=\hF_{N+1,+}+\hF_{N+1,-}.
\end{equation}
Then
\[
\d_{\pm}\hma_{N+1,\mp}\leq\frac{1}{2t}\hma_{N+1,\pm}+C_{N}t^{-3/2}(\hma_{N+1,+}+\hma_{N+1,-}).
\]
Integrating this differential inequality, taking
the supremum etc., we obtain
\[
\hF_{N+1}(t)\leq \hF_{N+1}(t_{1})+\int_{t_{1}}^{t}\left(
\frac{1}{2s}\hF_{N+1}(s)+C_{N}s^{-3/2}\hF_{N+1}(s)\right)ds.
\]
Thus $\hF_{N+1}(t)\leq C_{N}t^{1/2}$. Combining this estimate with 
Inductive assumption~\ref{ind:step2}, Lemma~\ref{lemma:indhcons2}
and (\ref{eq:roughalest}), we obtain (\ref{eq:highorder}).
\end{proof}


\section{Higher order derivatives of the characteristic system}
\label{section:hodecs}

In the previous section we showed that (\ref{eq:PQinds2}) holds 
with $N$ replaced by $N+1$; i.e., that 
\begin{equation}\label{eq:PQinds2Nplus1}
\|P_{\theta}\|_{C^{N}}+\|Q_{\theta}\|_{C^{N}}+t^{3/2}\|P_{t}\|_{C^{N}}
+t^{3/2}\|Q_{t}\|_{C^{N}}  \leq  C_{N}
\end{equation}
holds for all $t\geq t_{1}$. We also need to prove that 
(\ref{eq:ThetaVinds2}) holds with $j$ replaced with $N+1$. Before stating the
relevant result, let us make the following preliminary observation. 

\begin{lemma}\label{lemma:indhcon}
Consider a $\tn{2}$-symmetric solution to the Einstein-Vlasov equations with a 
cosmological constant $\Lambda>0$ and existence interval $(t_{0},\infty)$, 
where $t_{0}\geq 0$. Assume that the solution has $\lambda$-asymptotics and
let $t_{1}=t_{0}+2$. Assume, moreover, that 
Inductive assumption~\ref{ind:step2} holds for some $1\leq N\in\zo$. Then
there are constants $C_{j}$ and $m_{j}$, $j=0,\dots,N$, depending only on $N$
and the solution, such that
\begin{eqnarray}
\|\rho\|_{C^{N}}+t^{1/2}\|J_{i}\|_{C^{N}}+t\|P_{i}\|_{C^{N}}
+t\|S_{im}\|_{C^{N}} & \leq & C_{N}t^{-3/2}(\ln t)^{m_{N}},\label{eq:matheta}\\
\left\|\a^{-1}\d^{j}_{\theta}\a\right\|_{C^{0}}
+t^{1/2}\|\d^{j+1}_{\theta}\lambda\|_{C^{0}} & \leq & C_{j},
\label{eq:allaindfin}\\
\|J_{\theta}\|_{C^{N}}+\|K_{\theta}\|_{C^{N}} & \leq & C_{N}
(\ln t)^{m_{N}},\label{eq:JKderest}\\
\left\|\d_{\theta}^{l+1}\left(\frac{\a_{t}}{\a}\right)\right\|_{C^{0}}
+\|\d_{\theta}^{l+1}\lambda_{t}\|_{C^{0}} & \leq & C_{l}t^{-3/2},
\label{eq:ataltest}\\
\left\|\lambda_{t}-2\frac{\a_{t}}{\a}-4t^{1/2}e^{\lambda/2}\Lambda
\right\|_{C^{N}} & \leq & C_{N}t^{-2},\label{eq:hoest}\\
\left\|\lambda_{t}-\frac{\a_{t}}{\a}\right\|_{C^{N}} & \leq & C_{N}t^{-2}
\label{eq:atltest}
\end{eqnarray}
for $t\geq t_{1}$, $0\leq j\leq N$, $0\leq l\leq N-1$ and $i,m=1,2,3$.
Moreover, using the notation
\begin{equation}\label{eq:hodchsys}
\Psi_{j}=\d^{j}_{\theta}\Psi,\ \ \
Z^{i}_{j}=\d_{\theta}^{j}Z^{i}, \ \ \
V^{i}_{j}=\d_{\theta}^{j}V^{i},\ \ \
\Theta_{j}=\d_{\theta}^{j}\Theta
\end{equation}
(where the $\d$-operator used to define $Z$ and $\Psi$ is given by $\d_{\theta}$),
there are functions $c_{i,\theta}$, $i=1,2,3$, such that the following estimates 
hold:
\begin{eqnarray}
|\Psi_{l}(s;t,\theta,v)|+
s^{1/2}|Z_{l}(s;t,\theta,v)|
 & \leq &  C_{l}(\ln t)^{m_{l}},\label{eq:ZPld}\\
|\Psi_{j}(s;t,\theta,v)-(\a^{-1/2}e^{\lambda/2})[s,\Theta(s;t,\theta,v)]
\Theta_{j+1}(s;t,\theta,v)|
 & \leq & C_{j}(\ln t)^{m_{j}},\label{eq:PjTj}\\
|Z^{i}_{j}(s;t,\theta,v)-V^{i}_{j+1}(s;t,\theta,v)-
(c_{i,\theta}\Psi_{j})(s;t,\theta,v)| & \leq & C_{j}
s^{-1/2}(\ln t)^{m_{j}},\label{eq:ZjVj}\\
|c_{i,\theta}(s;t,\theta,v)| & \leq & C_{0}s^{-1/2}\label{eq:cithest}
\end{eqnarray}
for all $(t,\theta,v)\in[t_{1},\infty)\times\so\times\rn{3}$ in the support of
$f$, $0\leq j\leq N$, $0\leq l\leq N-1$ and $s\in [t_{1},t]$.
\end{lemma}
\begin{remark}
Due to (\ref{eq:allaindfin}),
\[
|\d_{\theta}^{j}\a^{p}|\leq C_{p,j}\a^{p}
\]
for all $(t,\theta)\in [t_{1},\infty)\times\so$, $p\in\ro$ and $0\leq j\leq N$.
In particular, spatial derivatives of powers of $\a$ can thus effectively
be ignored. In the derivation of the estimates below, it is useful to keep
this observation in mind. 
\end{remark}
\begin{proof}
Combining Lemma~\ref{lemma:indhcons2} with (\ref{eq:PQinds2Nplus1}),
(\ref{eq:altlteq}) and (\ref{eq:altltheq}), we obtain 
(\ref{eq:matheta})--(\ref{eq:ataltest}); recall that $P$, $Q$, $J$ and $K$ are bounded to the 
future. The estimate (\ref{eq:hoest}) is a
consequence of Lemma~\ref{lemma:indhcons2}, (\ref{eq:PQinds2Nplus1}) and
the fact that
\begin{equation*}
\begin{split}
\lambda_{t}-2\frac{\a_{t}}{\a}-4t^{1/2}e^{\lambda/2}\Lambda
 = & t\left[P_{t}^{2}+\a P_{\theta}^{2}+e^{2P}(Q_{t}^{2}
+\a Q_{\theta}^{2})\right]+\frac{e^{-P+\lambda/2}J^{2}}{t^{5/2}}\\
& +\frac{e^{P+\lambda/2}(K-QJ)^{2}}{t^{5/2}}
+4t^{1/2}e^{\lambda/2}\rho;
\end{split}
\end{equation*}
cf. (\ref{eq:altateq}) and (\ref{eq:altltmat}).
For similar reasons, (\ref{eq:atltest}) holds; cf. (\ref{eq:altltmat}).
Turning to (\ref{eq:ZPld})--(\ref{eq:cithest}), let us note that 
$\d_{\theta}^{j}\Psi(s;t,\theta,v)$ can, up to numerical factors, be written as 
a sum of terms of the form
\[
\d_{\theta}^{k}(\a^{-1/2}e^{\lambda/2})[s,\Theta(s;t,\theta,v)]
\Theta_{i_{1}}(s;t,\theta,v)\cdots \Theta_{i_{k+1}}(s;t,\theta,v),
\]
where $i_{1}+\dots+i_{k+1}=j+1$. Since $k\leq j\leq N$, the first factor is 
bounded due to (\ref{eq:allaindfin}). Due to Inductive 
assumption~\ref{ind:step2}, the factors $\Theta_{i_{j}}$ can be 
estimated by $C(\ln t)^{m_{j}}$ if $i_{j}\leq N$. The only way a factor
$\Theta_{N+1}$ could occur is if $k=0$ and all the derivatives hit $\Theta_{1}$ 
in the definition of $\Psi$. Due to these observations, (\ref{eq:PjTj})
and the estimate
\[
|\Psi_{l}(s;t,\theta,v)| \leq  C_{l}(\ln t)^{m_{l}},
\]
$0\leq l\leq N-1$, follow. The proof of the remaining estimates is similar
in nature, but somewhat more involved. 
\end{proof}

We now finish the induction argument by proving that (\ref{eq:ThetaVinds2}) holds
with $j$ replaced by $N+1$.

\begin{lemma}\label{lemma:indstep2}
Consider a $\tn{2}$-symmetric solution to the Einstein-Vlasov equations with a 
cosmological constant $\Lambda>0$ and existence interval $(t_{0},\infty)$, 
where $t_{0}\geq 0$. Assume that the solution has $\lambda$-asymptotics and
let $t_{1}=t_{0}+2$. Assume, moreover, that 
Inductive assumption~\ref{ind:step2} holds for some $1\leq N\in\zo$.
Then (\ref{eq:ThetaVinds2}) holds with $j$ replaced by $N+1$.
\end{lemma}
\begin{proof}
The strategy of the proof is very similar to that of the proof of 
Lemma~\ref{lemma:Jacobisys}. The idea is to derive a system of 
ODE's for $Z^{i}_{N}$ and $\Psi_{N}$ analogous to 
(\ref{eq:Zomanle})--(\ref{eq:Psimanle}), and then to use arguments  
similar to those presented in the proof of Lemma~\ref{leqmma:dercs}. Deriving 
appropriate equations
for $Z^{i}_{N}$, $i=2,3$, turns out to be relatively easy; this is due to 
(\ref{eq:Ztdot}) and (\ref{eq:Zthdot}). In fact, we obtain the desired
conclusions concerning $Z^{i}_{N}$, $i=2,3$, without much effort; cf. 
(\ref{eq:Ztkest}) and (\ref{eq:Zthkest}) below. Deriving an equation
for $\Psi_{N}$ also turns out to be quite easy; cf. (\ref{eq:Psikeq})
below. Similarly to the proof of Lemma~\ref{lemma:Jacobisys}, the main
difficulty consists of deriving an equation for $Z_{N}^{1}$. Once the 
desired equation has been obtained, we rescale $Z^{1}_{N}$ and $\Psi_{N}$
according to (\ref{eq:hZoNdef}) and introduce an energy according
to (\ref{eq:hENdef}); note that these definitions are 
analogous to the ones given in the proof of Lemma~\ref{leqmma:dercs}. 
Finally, the equations imply a differential inequality for the energy 
$\hE_{N}$ which can be integrated to yield the desired estimate. 

Before proceeding to a proof of the statement of the lemma, it is of interest 
to introduce
some notation. Let $b$ be a $C^{1}$ function on $M=(t_{0},\infty)\times\so$.
Evaluating this function along a characteristic, we obtain
\[
B(s;t,\theta,v)=b[s,\Theta(s;t,\theta,v)].
\]
Differentiating $B$ with respect to $\theta$, we obtain
\begin{equation}\label{eq:Hhrel}
\frac{\d B}{\d\theta}(s;t,\theta,v)=\frac{\d b}{\d \theta}
[s,\Theta(s;t,\theta,v)]\frac{\d\Theta}{\d\theta}(s;t,\theta,v).
\end{equation}
On the other hand, distinguishing between $B$ and $b$ is quite cumbersome
in the arguments that we are about to carry out. As a consequence, we shall
write $b$ when we mean $B$. Moreover, we shall use the notation $\dvt b$ as a
shorthand for $\d_{\theta}B$, whereas $\d_{\theta}b$ should be interpreted
as the first factor on the right hand side of (\ref{eq:Hhrel}) and
$\d_{t}b$ should be interpreted as the function mapping $(s;t,\theta,v)$ to
\[
\frac{\d b}{\d t}[s,\Theta(s;t,\theta,v)].
\]
In particular, we thus have $b_{\vartheta}=b_{\theta} \Theta_{1}$. Finally,
let us point out that if in some expression a $\vartheta$-derivative hits a
$V$ or a $\Theta$, it is to be interpreted as an ordinary $\theta$-derivative.

Note, to begin with, that
\[
\frac{d\Psi}{ds}=\frac{1}{2}\left(\lambda_{t}-\frac{\a_{t}}{\a}\right)\Psi
+\frac{1}{2}\a^{1/2}\lambda_{\theta}\frac{V^{1}}{V^{0}}\Psi
+e^{\lambda/2}\frac{\d V^{1}}{V^{0}}
-e^{\lambda/2}\frac{V^{1}}{(V^{0})^{3}}\sum_{i=1}^{3}V^{i}\d V^{i}
\]
due to (\ref{eq:Psidot}). Differentiating this equality $N$ times with
respect to $\vartheta$, we obtain
\[
\frac{d\Psi_{N}}{ds}=\frac{1}{2}\left(\lambda_{t}-\frac{\a_{t}}{\a}\right)
\Psi_{N}+\frac{1}{2}\a^{1/2}\lambda_{\theta}\frac{V^{1}}{V^{0}}\Psi_{N}
+e^{\lambda/2}\frac{V_{N+1}^{1}}{V^{0}}
-e^{\lambda/2}\frac{V^{1}}{(V^{0})^{3}}\sum_{i=1}^{3}V^{i} V^{i}_{N+1}
+O[s^{-2}(\ln t)^{m_{N}}],
\]
where we have used (\ref{eq:roughalest}), Inductive assumption~\ref{ind:step2}
and Lemma~\ref{lemma:indhcon}. Due to (\ref{eq:ZjVj}), this equation can be 
written
\begin{equation}\label{eq:Psikeq}
\frac{d\Psi_{N}}{ds}=c_{\theta,\theta}^{N}\Psi_{N}
+c_{\theta,i}^{N}Z_{N}^{i}+O[s^{-2}(\ln t)^{m_{N}}],
\end{equation}
where $c_{\theta,\theta}^{N}=O(s^{-2})$, $c_{\theta,i}^{N}=O(s^{-3/2})$
and we sum over $i$ but not $N$. Turning to $Z^{2}$, we have
\[
\frac{dZ^{2}}{ds}=-\frac{1}{2s}Z^{2}-\frac{1}{2}\left(P_{t}+\a^{1/2}P_{\theta}
\frac{V^{1}}{V^{0}}\right)Z^{2};
\]
cf. (\ref{eq:Ztdot}).
Differentiating this equality $N$ times with respect to $\vartheta$, we obtain
\[
\frac{dZ^{2}_{N}}{ds}
=-\frac{1}{2s}Z^{2}_{N}-\frac{1}{2}\left(P_{t}+\a^{1/2}P_{\theta}
\frac{V^{1}}{V^{0}}\right)Z^{2}_{N}+O[s^{-2}(\ln t)^{m_{N}}],
\]
where we have used (\ref{eq:roughalest}), 
(\ref{eq:PQinds2Nplus1}), Inductive assumption~\ref{ind:step2}
and Lemma~\ref{lemma:indhcon}.
As a consequence,
\[
\frac{d}{ds}\left(s^{1/2}e^{P/2}Z^{2}_{N}\right)=O[s^{-3/2}(\ln t)^{m_{N}}].
\]
Integrating this equality from $s$ to $t$, we obtain (assuming $N\geq 1$)
\[
-\left(s^{1/2}e^{P/2}Z^{2}_{N}\right)(s;t,\theta,v)=O[(\ln t)^{m_{N}}];
\]
note that
\[
\left(t^{1/2}e^{P/2}Z^{2}_{N}\right)(t;t,\theta,v)=t^{1/2}e^{P/2}
\frac{1}{2}(\d_{\theta}^{N+1}P) v^{2}=O(1)
\]
due to (\ref{eq:PQinds2Nplus1}). In particular, we thus have
\begin{equation}\label{eq:Ztkest}
|Z^{2}_{N}(s;t,\theta,v)|\leq C_{N}s^{-1/2}(\ln t)^{m_{N}}
\end{equation}
for $s\in [t_{1},t]$. Turning to $Z^{3}$, we have
\[
\frac{d Z^{3}}{ds}=-\frac{1}{2s}Z^{3}+\frac{1}{2}\left(P_{t}+\a^{1/2}P_{\theta}
\frac{V^{1}}{V^{0}}\right)Z^{3}
-e^{P}\left(Q_{t}+\a^{1/2}Q_{\theta}
\frac{V^{1}}{V^{0}}\right)Z^{2};
\]
cf. (\ref{eq:Zthdot}) and (\ref{eq:Ztdot}).
Differentiating $N$ times with respect to $\vartheta$, we obtain
\[
\frac{d Z^{3}_{N}}{ds}=-\frac{1}{2s}Z^{3}_{N}
+\frac{1}{2}\left(P_{t}+\a^{1/2}P_{\theta}
\frac{V^{1}}{V^{0}}\right)Z^{3}_{N}
-e^{P}\left(Q_{t}+\a^{1/2}Q_{\theta}
\frac{V^{1}}{V^{0}}\right)Z^{2}_{N}+O[s^{-2}(\ln t)^{m_{N}}].
\]
Due to (\ref{eq:Ztkest}), the third term
on the right hand side is $O[s^{-2}(\ln t)^{m_{N}}]$. We can thus proceed as
in the proof of (\ref{eq:Ztkest}) in order to obtain
\begin{equation}\label{eq:Zthkest}
|Z^{3}_{N}(s;t,\theta,v)|\leq C_{N}s^{-1/2}(\ln t)^{m_{N}}
\end{equation}
for $s\in [t_{1},t]$.
Finally, we need to derive an equation for $Z^{1}_{N}$. Just as in the
derivation of the equation for $Z^{1}$, it is natural to divide the
analysis into several steps. Consider, to begin with,
\[
\dvt^{j+1}\left(\frac{dV^{1}}{ds}\right)
\]
for $0\leq j\leq N$.
All the terms appearing in $dV^{1}/ds$ can be written $h\, \psi\circ V$.
When differentiating an expression of this form, the terms that arise are (up
to numerical factors) of the form $\dvt^{k}h \d_{\theta}^{l}\psi\circ V$.
If both $k$ and $l$ are $\geq 1$, the resulting term is
$O[s^{-2}(\ln t)^{m_{j}}]$.
If all the derivatives hit $\psi$, we obtain (after summing over all the
terms appearing in $dV^{1}/ds$)
\[
-\frac{1}{2s}V^{1}_{j+1}+c_{i}^{j}V^{i}_{j+1}+O[s^{-2}(\ln t)^{m_{j}}]
\]
where $c_{i}^{j}=O(s^{-2})$ and we sum over $i$ but not $j$. If all
the derivatives hit $h$, we obtain (after summing over all the
terms appearing in $dV^{1}/ds$)
\begin{equation*}
\begin{split}
&-\frac{1}{4}\dvt^{j+1}(\a^{1/2}\lambda_{\theta})V^{0}-\frac{1}{4}
\dvt^{j+1}\left(\lambda_{t}-2\frac{\a_{t}}{\a}\right)V^{1}
+\dvt^{j+1}(\a^{1/2}e^{P}Q_{\theta})\frac{V^{2}V^{3}}{V^{0}}\\
& -\frac{1}{2}\dvt^{j+1}(\a^{1/2}P_{\theta})
\frac{(V^{3})^{2}-(V^{2})^{2}}{V^{0}}+c_{\theta}^{j}\Theta_{j+1}
+O[s^{-3}(\ln t)^{m_{j}}],
\end{split}
\end{equation*}
where $c_{\theta}^{j}=O(s^{-3})$ and we have used 
(\ref{eq:JthKthest}); note that, due to (\ref{eq:PQinds2Nplus1}) and 
Lemma~\ref{lemma:indhcon}, we
control $N+1$ $\theta$-derivatives of the first factor in each of the last
two terms appearing on the right hand side of (\ref{eq:dVods}).
Adding up, we conclude that
\begin{equation}\label{eq:dthkpoVo}
\begin{split}
\dvt^{j+1}\left(\frac{dV^{1}}{ds}\right) = &
-\frac{1}{2s}V^{1}_{j+1}+c_{i}^{j}V^{i}_{j+1}+c_{\theta}^{j}\Theta_{j+1}
-\frac{1}{4}\dvt^{j+1}(\a^{1/2}\lambda_{\theta})V^{0}-\frac{1}{4}
\dvt^{j+1}\left(\lambda_{t}-2\frac{\a_{t}}{\a}\right)V^{1}\\
&+\dvt^{j+1}(\a^{1/2}e^{P}Q_{\theta})\frac{V^{2}V^{3}}{V^{0}}
 -\frac{1}{2}\dvt^{j+1}(\a^{1/2}P_{\theta})
\frac{(V^{3})^{2}-(V^{2})^{2}}{V^{0}}
+O[s^{-2}(\ln t)^{m_{j}}],
\end{split}
\end{equation}
where $c_{i}^{j}=O(s^{-2})$ and $c_{\theta}^{j}=O(s^{-3})$ and we sum over
$i$ but not $j$.

The second term in the definition of $Z^{1}$ consists of a sum of
terms of the form
\begin{equation}\label{eq:hpsistr}
h\, \psi\circ V\, \a^{-1/2}\, \Theta_{1}.
\end{equation}
The relevant $h$'s are
\[
h_{1}=\frac{1}{4}\left(\lambda_{t}
-2\frac{\a_{t}}{\a}-4s^{1/2}e^{\lambda/2}\Lambda\right),\
h_{2}=-\frac{1}{2}P_{t},\
h_{3}=\frac{1}{2}\a^{1/2}P_{\theta},\
h_{4}=-e^{P}Q_{t},\
h_{5}=\a^{1/2}e^{P}Q_{\theta},
\]
and the relevant $\psi$'s are
\begin{equation*}
\begin{split}
\psi_{1} = & V^{0},\ \
\psi_{2}=V^{0}\frac{(V^{2})^{2}-(V^{3})^{2}}{(V^{0})^{2}-(V^{1})^{2}},\ \
\psi_{3}=V^{1}\frac{(V^{2})^{2}-(V^{3})^{2}}{(V^{0})^{2}-(V^{1})^{2}},\ \
\psi_{4}=\frac{V^{0}V^{2}V^{3}}{(V^{0})^{2}-(V^{1})^{2}},\\
\psi_{5} = & \frac{V^{1}V^{2}V^{3}}{(V^{0})^{2}-(V^{1})^{2}}.
\end{split}
\end{equation*}
We want to differentiate (\ref{eq:hpsistr}) with respect to $s$ and then
$N$ times with respect to $\vartheta$. Before turning to the details,
let us record the following estimate:
\begin{equation}\label{eq:hiest}
s^{1/2}|(\dvt^{j}h_{1})(s;t,\theta,v)|
+\sum_{i=2}^{5}|(\dvt^{j}h_{i})(s;t,\theta,v)|\leq C_{N}s^{-3/2}(\ln t)^{m_{N}}
\end{equation}
for $0\leq j\leq N$, $(t,v,\theta)\in [t_{1},\infty)\times\so\times\rn{3}$
in the support of $f$ and $s\in [t_{1},t]$. In the case of $h_{i}$,
$i=2,\dots,5$, (\ref{eq:hiest}) is an immediate consequence of the
inductive hypothesis, (\ref{eq:roughalest}), (\ref{eq:PQinds2Nplus1}) and 
Lemma~\ref{lemma:indhcon}, 
and in the case of $h_{1}$, it is a consequence of (\ref{eq:hoest}). We also 
have
\begin{equation}\label{eq:psiest}
|\psi_{1}(s;t,\theta,v)|+s\sum_{i=2}^{5}|\psi_{i}(s;t,\theta,v)|\leq C, \ \ \
\sum_{i=1}^{5}|(\d_{\theta}^{j+1}\psi_{i}\circ V)(s;t,\theta,v)|\leq
C_{j}s^{-1}(\ln t)^{m_{j}}
\end{equation}
for $0\leq j\leq N-1$, $(t,v,\theta)\in [t_{1},\infty)\times\so\times\rn{3}$
in the support of $f$ and $s\in [t_{1},t]$; this is an immediate consequence
of the inductive hypothesis.

Let us consider the term that arises when $d/ds$ hits the $\psi$-factor
in (\ref{eq:hpsistr}). Note, to this end, that 
\[
\left|\dvt^{j}\left(\frac{dV^{i}}{ds}\right)\right|\leq
C_{N}s^{-3/2}(\ln t)^{m_{j}}
\]
for all $i=1,2,3$ and all $0\leq j\leq N$; for $j=0$, the estimate is
a consequence of (\ref{eq:dVdscrude}); for $j\geq 1$ and $i=1$, it is a 
consequence of (\ref{eq:dthkpoVo}); and in the case of $i=2,3$, it follows 
immediately from (\ref{eq:dVtds}), (\ref{eq:dVthds}) and the induction 
hypothesis. Due to the above observations,
\[
\dvt^{N}\left(h\frac{d\psi\circ V}{ds}
\a^{-1/2}\Theta_{1}\right)=c_{\theta}^{N}
\Theta_{N+1}+O[s^{-2}(\ln t)^{m_{j}}],
\]
where $c_{\theta}^{N}=O(s^{-2})$ and $h$ is one of $h_{1},\dots,h_{5}$. When the 
$s$-derivative hits the remaining terms in (\ref{eq:hpsistr}) (not $\psi$), we 
obtain
\[
\dvt^{N}\left[\left(\d_{t}(\a^{-1/2}h)
+h_{\theta}\frac{V^{1}}{V^{0}}\right)\psi\Theta_{1}
+h\psi\frac{\d V^{1}}{V^{0}}
-h\psi\frac{V^{1}}{(V^{0})^{3}}\sum_{i=1}^{3}
V^{i}\d V^{i}\right];
\]
cf. (\ref{eq:ddsnotpsi}). Due to (\ref{eq:hiest}) and (\ref{eq:psiest}),
this expression can be written
\[
\dvt^{N}\left[\left(\d_{t}(\a^{-1/2}h)
+h_{\theta}\frac{V^{1}}{V^{0}}\right)\psi\Theta_{1}\right]
+\sum_{i=1}^{3}c_{i}^{N}V^{i}_{N+1}+O[s^{-2}(\ln t)^{m_{N}}],
\]
where $c_{i}^{N}=O(s^{-2})$. Differentiating the
second term appearing in the definition of $Z^{1}$ once with respect to
$s$ and $N$ times with respect to $\vartheta$, we obtain (by adding up the 
above)
\begin{equation}\label{eq:wnintstep}
\sum_{i=1}^{5}\dvt^{N}\left[\left(\d_{t}(\a^{-1/2}h_{i})
+(\d_{\theta}h_{i})\frac{V^{1}}{V^{0}}\right)\psi_{i}\Theta_{1}\right]
+\sum_{i=1}^{3}c_{i}^{N}V^{i}_{N+1}+c_{\theta}^{N}
\Theta_{N+1}+O[s^{-2}(\ln t)^{m_{N}}],
\end{equation}
where $c_{\theta}^{N}=O(s^{-2})$ and $c_{i}^{N}=O(s^{-2})$.
In order to obtain the desired equation we need to add this expression
to (\ref{eq:dthkpoVo}) with $j=N$. However, before doing so, note that
\[
-\frac{1}{4}\dvt^{N+1}(\a^{1/2}\lambda_{\theta})V^{0}
=-\frac{1}{4}\dvt^{N}[\d_{\theta}(\a^{1/2}\lambda_{\theta})V^{0}\Theta_{1}]
+O[s^{-2}(\ln t)^{m_{N}}]
\]
etc. Due to this observation, we can argue as in the proof of
Lemma~\ref{lemma:Jacobisys}. In particular, we obtain a formula 
analogous to (\ref{eq:hugedZods}): the difference is that 
$\d V^{1}$ should be replaced by $V^{1}_{N+1}$ in the first term on 
the right hand side of (\ref{eq:hugedZods}); that $\dvt^{N}$ should
be applied to all but the first and last two terms on the right hand 
side of (\ref{eq:hugedZods}); and that the last two terms should be replaced
by ones analogous to the last three terms on the right hand side of 
(\ref{eq:wnintstep}). 
Proceeding as in the proof of Lemma~\ref{lemma:Jacobisys}, the 
corresponding expression can be simplified; cf. the derivation of 
(\ref{eq:dZoexinte}). Most of the steps involved in the derivation of 
(\ref{eq:dZoexinte}) consist of algebraic manipulations. However, there
are two exceptions. The combination of the fourth and fifth last terms
on the right hand side of (\ref{eq:hugedZods}) can be written
\[
e^{P}(P_{t}Q_{\theta}-P_{\theta}Q_{t})
\frac{V^{1}V^{2}V^{3}}{(V^{0})^{2}-(V^{1})^{2}}\Theta_{1}.
\]
The analogous expression in the present setting is 
\[
\dvt^{N}\left(e^{P}(P_{t}Q_{\theta}-P_{\theta}Q_{t})
\frac{V^{1}V^{2}V^{3}}{(V^{0})^{2}-(V^{1})^{2}}\Theta_{1}\right)
=c_{\theta}^{N}
\Theta_{N+1}+O[s^{-3}(\ln t)^{m_{N}}],
\]
where $c_{\theta}^{N}=O(s^{-3})$ and we have used (\ref{eq:PQinds2Nplus1}), 
Inductive assumption~\ref{ind:step2} and Lemma~\ref{lemma:indhcon}. 
The combination of the third and the seventh term on the right hand side of 
(\ref{eq:hugedZods}) can be written
\[
-\d_{\theta}(s^{1/2}e^{\lambda/2}\Lambda) V^{1}\Theta_{1}.
\]
In the present setting, the analogous term is 
\[
-\dvt^{N}[\d_{\theta}(s^{1/2}e^{\lambda/2}\Lambda) V^{1}\Theta_{1}]=c_{\theta}^{N}
\Theta_{N+1}+O[s^{-2}(\ln t)^{m_{N}}],
\]
where $c_{\theta}^{N}=O(s^{-2})$ and we have used Inductive 
assumption~\ref{ind:step2} and Lemma~\ref{lemma:indhcon}. Summing up,
we obtain
\begin{equation}\label{eq:dZodsalfin}
\begin{split}
\frac{dZ^{1}_{N}}{ds} = & -\frac{1}{2s}V^{1}_{N+1}
+\dvt^{N}\left[\frac{1}{4}\d_{t}\left[\a^{-1/2}\left(\lambda_{t}
-2\frac{\a_{t}}{\a}-4s^{1/2}e^{\lambda/2}\Lambda\right)\right]V^{0}\Theta_{1}
-\frac{1}{4}\d_{\theta}(\a^{1/2}\lambda_{\theta})V^{0}\Theta_{1}\right.\\
&-\frac{1}{2}[\d_{t}(\a^{-1/2}P_{t})-\d_{\theta}(\a^{1/2}P_{\theta})]V^{0}
\frac{(V^{2})^{2}-(V^{3})^{2}}{(V^{0})^{2}-(V^{1})^{2}}\Theta_{1}\\
 & \left.-[\d_{t}(\a^{-1/2}e^{P}Q_{t})-\d_{\theta}(\a^{1/2}e^{P}Q_{\theta})]
\frac{V^{0}V^{2}V^{3}}{(V^{0})^{2}-(V^{1})^{2}}\Theta_{1}\right]\\
 & +c_{\theta}^{N}\Theta_{N+1}
+c^{N}_{i}V^{i}_{N+1}+O[s^{-2}(\ln t)^{m_{N}}],
\end{split}
\end{equation}
where $c^{N}_{\theta}=O(s^{-2})$, $c_{i}^{N}=O(s^{-2})$ and we sum over $i$ but
not over $N$. The term of importance is the second one on the right hand side.
If all the $\vartheta$-derivatives hit $\Theta_{1}$, the resulting term can be
dealt with as in the proof of Lemma~\ref{lemma:Jacobisys}, and we obtain a
$c_{\theta}^{N}\Theta_{N+1}$-term, where $c^{N}_{\theta}=O(s^{-3/2})$. In the case
of all the remaining terms, it is possible to use the equations (as in the
proof of Lemma~\ref{lemma:Jacobisys}) in order to obtain terms of the form 
$O[s^{-3/2}(\ln t)^{m_{N}}]$. Let us go through the argument in detail for
\[
\dvt^{N}\left[-\frac{1}{2}[\d_{t}(\a^{-1/2}P_{t})-\d_{\theta}(\a^{1/2}P_{\theta})]V^{0}
\frac{(V^{2})^{2}-(V^{3})^{2}}{(V^{0})^{2}-(V^{1})^{2}}\Theta_{1}\right].
\]
Using (\ref{eq:ttmthth}), this expression can be written
\begin{equation*}
\begin{split}
 & -\frac{1}{2}\dvt^{N}\left[\left(-\frac{1}{s}\a^{-1/2}P_{t}
+\a^{-1/2}e^{2P}(Q_{t}^{2}-\a Q_{\theta}^{2})+
\frac{\a^{-1/2}e^{\lambda/2-P}J^{2}}{2s^{7/2}}
-\frac{\a^{-1/2}e^{\lambda/2+P}(K-QJ)^{2}}{2s^{7/2}}\right.\right.\\
& \left.\left. +s^{-1/2}e^{\lambda/2}\a^{-1/2}(P_{2}-P_{3})\right)V^{0}
\frac{(V^{2})^{2}-(V^{3})^{2}}{(V^{0})^{2}-(V^{1})^{2}}\Theta_{1}\right]
=c_{\theta}^{N}\Theta_{N+1}+O[s^{-2}(\ln t)^{m_{N}}],
\end{split}
\end{equation*}
where $c_{\theta}=O(s^{-2})$ and we have used (\ref{eq:roughalest}), 
(\ref{eq:PQinds2Nplus1}), Inductive assumption~\ref{ind:step2} and 
Lemma~\ref{lemma:indhcon}. Due to this
argument, and similar ones for the remaining terms in (\ref{eq:dZodsalfin}),
we obtain
\[
\frac{dZ^{1}_{N}}{ds} =  -\frac{1}{2s}V^{1}_{N+1}+c_{\theta}^{N}\Theta_{N+1}
+c^{N}_{i}V^{i}_{N+1}+O[s^{-3/2}(\ln t)^{m_{N}}],
\]
where $c^{N}_{\theta}=O(s^{-3/2})$, $c_{i}^{N}=O(s^{-3/2})$ and we sum over $i$ but
not over $N$. Due to (\ref{eq:PjTj}) and (\ref{eq:ZjVj}), we conclude that
\[
\frac{dZ^{1}_{N}}{ds} =  -\frac{1}{2s}Z^{1}_{N}+c_{\theta}^{N}\Psi_{N}
+c^{N}_{i}Z^{i}_{N}+O[s^{-3/2}(\ln t)^{m_{N}}],
\]
where $c^{N}_{\theta}=O(s^{-3/2})$, $c_{i}^{N}=O(s^{-3/2})$ and we sum over $i$ but
not over $N$. Combining this equation with (\ref{eq:Psikeq}),
(\ref{eq:Ztkest}) and (\ref{eq:Zthkest}) yields
\begin{eqnarray*}
\frac{d\hPs_{N}}{ds} & = & \frac{2}{s\ln s}\hPs_{N}+c_{\theta,\theta}^{N}\hPs_{N}
+c_{\theta,1}^{N}s^{-1/2}(\ln s)^{2}\hZ_{N}^{1}
+O[s^{-2}(\ln s)^{2}(\ln t)^{m_{N}}],\\
\frac{d\hZ^{1}_{N}}{ds} & = & c_{1,\theta}^{N}s^{1/2}(\ln s)^{-2}\hPs_{N}
+c^{N}_{1,1}\hZ^{1}_{N}+O[s^{-1}(\ln t)^{m_{N}}],
\end{eqnarray*}
where $c^{N}_{\theta,\theta}=O(s^{-2})$, $c_{\theta,1}^{N}=O(s^{-3/2})$,
$c^{N}_{1,\theta}=O(s^{-3/2})$, $c_{1,1}^{N}=O(s^{-3/2})$, there is no
summation over $N$ and we have used the notation
\begin{equation}\label{eq:hZoNdef}
\hZ^{1}_{N}(s;t,\theta,v)=s^{1/2}Z^{1}_{N}(s;t,\theta,v),\ \ \
\hPs_{N}(s;t,\theta,v)=(\ln s)^{2}\Psi_{N}(s;t,\theta,v).
\end{equation}
Introducing the energy
\begin{equation}\label{eq:hENdef}
\hE_{N}=(\hPs_{N})^{2}+(\hZ^{1}_{N})^{2},
\end{equation}
we conclude that
\[
\frac{d\hE_{N}}{ds}\geq -\frac{C_{N}}{s(\ln s)^{2}}\hE_{N}
-C_{N}s^{-1}(\ln t)^{m_{N}}\hE_{N}^{1/2}.
\]
Letting $r_{N}$ be such that $r_{N}(t_{1})=0$ and its derivative is the first
factor in the first term on the right hand side, we obtain
\[
\frac{d\me_{N}}{ds}\geq -C_{N}s^{-1}(\ln t)^{m_{N}}\me_{N}^{1/2},
\]
where $\me_{N}=\exp(-r_{N})\hE_{N}$. Dividing by $\me^{1/2}_{N}$ and integrating
from $s$ to $t$, we obtain
\[
\me_{N}^{1/2}(s;t,\theta,v)\leq \me_{N}^{1/2}(t;t,\theta,v)
+C_{N}(\ln t)^{m_{N}}.
\]
However, the first term on the right hand side can be estimated by
$C_{N}(\ln t)^{2}$. Combining the resulting estimate with (\ref{eq:PjTj}),
(\ref{eq:ZjVj}), (\ref{eq:Ztkest}) and (\ref{eq:Zthkest}), we conclude that 
(\ref{eq:ThetaVinds2}) holds with $j=N+1$.
\end{proof}

\begin{cor}\label{cor:tderbet}
Consider a $\tn{2}$-symmetric solution to the Einstein-Vlasov equations with a 
cosmological constant $\Lambda>0$ and existence interval $(t_{0},\infty)$, 
where $t_{0}\geq 0$. Assume that the solution has $\lambda$-asymptotics and
let $t_{1}=t_{0}+2$. Let $0\leq k\in\zo$. Then there is a constant $C_{k}$, 
depending only on $k$ and the solution, such that
\begin{equation}\label{eq:tderbet}
\|P_{t}\|_{C^{k}}+\|Q_{t}\|_{C^{k}}+t^{1/2}\|\a^{1/2}\lambda_{\theta}\|_{C^{k}}
\leq C_{k}t^{-2}
\end{equation}
for all $t\geq t_{1}$. Moreover,
\[
\left\|\frac{\a_{t}}{\a}+\frac{3}{t}\right\|_{C^{k}}
+\left\|\lambda_{t}+\frac{3}{t}\right\|_{C^{k}}\leq C_{k}t^{-2}
\]
for all $t\geq t_{1}$.
\end{cor}
\begin{proof}
By combining Lemmas~\ref{lemma:indstep1} and  \ref{lemma:indstep2},
we know that Inductive assumption~\ref{ind:step2} holds for all $N$.
In particular, we thus know that the conclusions of Lemma~\ref{lemma:indhcon}
hold for all $N$. Combining this information with (\ref{eq:ttmthth}) and 
(\ref{eq:Etth}) yields
\[
\|\d_{t}(t\a^{-1/2}P_{t})\|_{C^{k}}+
\|\d_{t}(t\a^{-1/2}Q_{t})\|_{C^{k}}\leq C_{k}t^{-1/2}.
\]
As a consequence,
\[
\|t\a^{-1/2}P_{t}\|_{C^{k}}+\|t\a^{-1/2}Q_{t}\|_{C^{k}}\leq C_{k}t^{1/2}.
\]
Due to this estimate, as well as (\ref{eq:roughalest})
and (\ref{eq:allaindfin}), we can proceed inductively in order
to conclude that
\[
\|P_{t}\|_{C^{k}}+\|Q_{t}\|_{C^{k}}\leq C_{k}t^{-2}.
\]
Combining this estimate with (\ref{eq:altltheq}),
(\ref{eq:roughalest}), (\ref{eq:PQinds2Nplus1}) and
Lemma~\ref{lemma:indhcon},
we obtain (\ref{eq:tderbet}). Combining (\ref{eq:tderbet}) with
(\ref{eq:altlteq}), (\ref{eq:altateq}) and (\ref{eq:atltsupest}), we obtain
the final estimate stated in the corollary.
\end{proof}

\section{Energy estimates for the distribution function}
\label{section:eedf}

In the proof of the existence of $f_{\sca,\infty}$, cf. Theorem~\ref{thm:as},
a natural first step is to estimate $L^{2}$-based energies for $f$. In the
process of deriving such estimates, it is useful to consider equations for the
derivatives of the distribution function. Such equations take the following
general form:
\begin{equation}\label{eq:vlstru}
\frac{\d h}{\d t}+\frac{\a^{1/2}v^{1}}{v^{0}}\frac{\d h}{\d \theta}
-\frac{1}{2t}v^{i}\frac{\d h}{\d v^{i}}=R.
\end{equation}
In case $h=f$, $R$ is given by
\begin{equation}\label{eq:Rdecf}
R=L^{i}\frac{\d f}{\d v^{i}},
\end{equation}
where
\begin{eqnarray}
L^{1} & = & \frac{1}{4}\a^{1/2}\lambda_{\theta}\,v^{0}
+\frac{1}{4}\left(\lambda_{t}-\frac{2\alpha_t}{\alpha}-\frac{3}{t}\right)v^{1}
-\a^{1/2}e^{P}Q_{\theta}\frac{v^2v^3}{v^0}
+\frac{1}{2}\a^{1/2}P_{\theta}\frac{(v^{3})^{2}
-(v^{2})^{2}}{v^{0}}\nonumber\\
 & & -t^{-7/4}e^{\lambda/4}\big(e^{-P/2}Jv^2+e^{P/2}(K-Q J)v^3\big),
\label{eq:Lodef}\\
L^{2} & = & \frac{1}{2}P_{t}v^{2}
+\frac{1}{2}\a^{1/2}P_{\theta}\frac{v^{1}v^{2}}{v^{0}},\label{eq:Ltdef}\\
L^{3} & = & -\frac{1}{2}P_{t}v^{3}
-\frac{1}{2}\a^{1/2}P_{\theta}\frac{v^{1}v^{3}}{v^{0}}+e^{P}v^2
\left(Q_t+\a^{1/2}Q_{\theta}\frac{v^1}{v^0}\right).\label{eq:Lthdef}
\end{eqnarray}
The energies we shall consider are 
\begin{equation}\label{eq:ekdef}
E_{k}[h](t)=\sum_{l+|\b|\leq k}\is\irn{3}t^{-|\b|}
|\d_{\theta}^{l}\d^{\b}_{v}h(t,\theta,v)|^{2}\a^{-1/2}t^{-3/2}dvd\theta.
\end{equation}
We shall also use the notation $E=E_{0}$.
\begin{remarks}
The purpose of the factor $\a^{-1/2}t^{-3/2}$ is to simplify some of the
terms that result upon carrying out partial integrations. We could equally
well consider energies of the form
\[
H_{k}[f](t)=\sum_{l+|\b|\leq k}\is\irn{3}t^{-|\b|}\ldr{t^{1/2}v}^{2\mu+2|\b|}
|\d_{\theta}^{l}\d^{\b}_{v}f(t,\theta,v)|^{2}dvd\theta
\]
for $\mu\geq 0$; cf. \cite{stab}.
However, there is a constant $C>1$, depending only on the solution, $\mu$ and
$\b$, such that
\[
C^{-1}\leq \ldr{t^{1/2}v}^{2\mu+2|\b|}\leq C
\]
for $t\geq t_{1}$ (where $t_{1}$ is defined as in the statement of previous 
lemmas) and $(t,\theta,v)$ in the support of $f$. As a consequence,
the corresponding weight is of no practical importance.
\end{remarks}

\begin{lemma}\label{lemma:vlfisten}
Consider a $\tn{2}$-symmetric solution to the Einstein-Vlasov equations with a 
cosmological constant $\Lambda>0$ and existence interval $(t_{0},\infty)$, 
where $t_{0}\geq 0$. Assume that the solution has $\lambda$-asymptotics and
let $t_{1}=t_{0}+2$. Let $h$ be a smooth
solution to (\ref{eq:vlstru}) (where the function $\a$ is the object appearing
in the Einstein-Vlasov equations and $R$ is some function) which has compact
support when restricted to compact time intervals. Then there is a constant
$C>0$, depending only on the solution to the Einstein-Vlasov
equations, such that
\[
\frac{dE[h]}{dt}\leq -\frac{3}{2t}E[h]
+2\is\irn{3}hR\a^{-1/2}t^{-3/2}dv d\theta+Ct^{-2}E[h]
\]
for all $t\geq t_{1}$.
\end{lemma}
\begin{remark}
It is important to note that the constant $C$ does not depend on $h$.
Moreover, $R$ should be thought of as being defined by (\ref{eq:vlstru}).
In particular, due to the assumptions concerning $h$, the function $R$ is
smooth and has compact support when restricted to compact time intervals.
\end{remark}
\begin{proof}
Differentiating $E$ with respect to time, we obtain
\begin{equation}\label{eq:difzeenvl}
\frac{dE}{dt} = 2\is\irn{3}h\d_{t}h\a^{-1/2}t^{-3/2}dv d\theta
+\is\irn{3}h^{2}\left(-\frac{3}{2t}-\frac{\a_{t}}{2\a}
\right)\a^{-1/2}t^{-3/2}dv d\theta.
\end{equation}
Due to (\ref{eq:atltsupest}), we can estimate the second term on
the right hand side. Consider the first term on the right hand side
of (\ref{eq:difzeenvl}). Using (\ref{eq:vlstru}), it can be
written
\[
2\is\irn{3}h\left(-\frac{\a^{1/2}v^{1}}{v^{0}}
\frac{\d h}{\d \theta}+\frac{1}{2t}v^{i}\frac{\d h}{\d v^{i}}
+R\right)\a^{-1/2}t^{-3/2}dv d\theta.
\]
The term involving $\d_{\theta}h$ can be integrated to zero. The term involving
$R$ we leave as it is. What remains is to estimate the term
\[
\frac{1}{2t}\is\irn{3}v^{i}\frac{\d h^{2}}{\d v^{i}}
\a^{-1/2}t^{-3/2}dv d\theta =
-\frac{3}{2t}\is\irn{3}h^{2}
\a^{-1/2}t^{-3/2}dv d\theta.
\]
The lemma follows.
\end{proof}

Let us turn to the higher order derivatives of the distribution function.

\begin{lemma}
Consider a $\tn{2}$-symmetric solution to the Einstein-Vlasov equations with a 
cosmological constant $\Lambda>0$ and existence interval $(t_{0},\infty)$, 
where $t_{0}\geq 0$. Assume that the solution has $\lambda$-asymptotics and
let $t_{1}=t_{0}+2$. Fix $0\leq k\in\zo$. 
Then there is a constant $C_{k}>0$, depending only on $k$ and the 
solution to the Einstein-Vlasov equations, such that
\[
\frac{dE_{k}[f]}{dt}\leq -\frac{3}{2t}E_{k}[f]+
C_{k}t^{-3/2}E_{k}[f]
\]
for all $t\geq t_{1}$. In particular, $t^{3/2}E_{k}[f]$ is bounded to the
future.
\end{lemma}
\begin{proof}
Differentiating (\ref{eq:vlstru}) with $h=f$, we obtain
\begin{equation}\label{eq:diffvl}
\frac{\d f_{\b,l}}{\d t}+\frac{\a^{1/2}v^{1}}{v^{0}}\frac{\d f_{\b,l}}{\d \theta}
-\frac{1}{2t}v^{i}\frac{\d f_{\b,l}}{\d v^{i}}=\d^{\b}_{v}\d^{l}_{\theta}R+
\left[\frac{\a^{1/2}v^{1}}{v^{0}}\d_{\theta},\d^{\b}_{v}\d^{l}_{\theta}\right]f
-\frac{1}{2t}\left[v^{i}\d_{v^{i}},\d^{\b}_{v}\d^{l}_{\theta}\right]f,
\end{equation}
where we use the notation $f_{\b,l}=\d^{\b}_{v}\d^{l}_{\theta}f$ and assume that
$|\b|+l\leq k$.
Let us denote the right hand side of (\ref{eq:diffvl}) by $\sfR_{\b,l}$.
Due to Lemma~\ref{lemma:vlfisten}, it is of interest to estimate
\begin{equation}\label{eq:inhocont}
2\is\irn{3}t^{-|\b|}f_{\b,l}\sfR_{\b,l}\a^{-1/2}t^{-3/2}dv d\theta.
\end{equation}
By an inductive argument, it can be proven that the third term on the
right hand side of (\ref{eq:diffvl}) is given by $|\b|f_{\b,l}/2t$.
The corresponding contribution to (\ref{eq:inhocont}) is thus
\[
\frac{|\b|}{t}t^{-|\b|}E[f_{\b,l}].
\]
Turning to the second term on the right hand side of (\ref{eq:diffvl}),
it can (up to numerical factors) be written as a sum of terms of the form
\[
\d^{\b_{1}}_{v}\d^{l_{1}}_{\theta}\left(\frac{\a^{1/2}v^{1}}{v^{0}}\right)
\d^{\b_{2}}_{v}\d^{l_{2}+1}_{\theta}f,
\]
where $\b_{1}+\b_{2}=\b$, $l_{1}+l_{2}=l$ and $|\b_{1}|+l_{1}\geq 1$. Note that
the first factor can always be estimated by $C_{k}t^{-3/2}$. In case $\b_{1}=0$,
it can be estimated by $C_{k}t^{-2}$ (on the support of $f$). Due to these 
observations, we have
\[
t^{-|\b|/2}\left|
\left[\frac{\a^{1/2}v^{1}}{v^{0}}\d_{\theta},\d^{\b}_{v}\d^{l}_{\theta}\right]f
\right|\leq C_{k}t^{-2}\sum_{l_{1}+|\b_{1}|\leq k}t^{-|\b_{1}|/2}|f_{\b_{1},l_{1}}|.
\]
The corresponding contribution to (\ref{eq:inhocont}) can thus be estimated
by
\[
C_{k}t^{-2}E_{k}[f].
\]
Finally, let us consider the first term on the right hand side of
(\ref{eq:diffvl}). Since $R$ is given by (\ref{eq:Rdecf}), the
expression $\d^{\b}_{v}\d^{l}_{\theta}R$ is given by the sum of
\begin{equation}\label{eq:alldhf}
L^{i}\d_{v^{i}}\d^{\b}_{v}\d^{l}_{\theta}f
\end{equation}
and terms which (up to numerical factors) can be written
\begin{equation}\label{eq:nalldhf}
(\d^{\b_{1}}_{v}\d^{l_{1}}_{\theta}L^{i})
\d_{v^{i}}\d^{\b_{2}}_{v}\d^{l_{2}}_{\theta}f,
\end{equation}
where $|\b_{1}|+l_{1}\geq 1$.
The contribution to (\ref{eq:inhocont}) from (\ref{eq:alldhf}) can be written
\[
\is\irn{3}t^{-|\b|}L^{i}(\d_{v^{i}}f_{\b,l}^{2})\a^{-1/2}t^{-3/2}dv d\theta.
\]
Integrating partially with respect to $v^{i}$ and keeping in mind that
the $L^{i}$ are given by (\ref{eq:Lodef})--(\ref{eq:Lthdef}),
we conclude that this expression can be estimated by
\begin{equation}\label{eq:slasken}
C_{k}t^{-3/2}E_{k}[f],
\end{equation}
where we have used Lemma~\ref{lemma:indhcon} and 
Corollary~\ref{cor:tderbet}.
Let us now consider the contribution arising from terms of the form
(\ref{eq:nalldhf}). It is natural to divide these terms into two different
categories. Either $\b_{1}=0$, or $\b_{1}\neq 0$. In case $\b_{1}=0$,
the expression (\ref{eq:nalldhf}) can be estimated by
\[
C_{k}t^{-3/2}t^{-1/2}|\d_{v^{i}}\d^{\b}_{v}\d^{l_{2}}_{\theta}f|.
\]
In case $\b_{1}\neq 0$, the expression (\ref{eq:nalldhf}) can be estimated by
\[
C_{k}t^{-3/2}|\d_{v^{i}}\d^{\b_{2}}_{v}\d^{l_{2}}_{\theta}f|.
\]
In order to obtain these estimates, we have appealed to 
Lemma~\ref{lemma:indhcon} and Corollary~\ref{cor:tderbet}.
As a consequence, the contribution to (\ref{eq:inhocont}) from terms of the
form (\ref{eq:nalldhf}) can be estimated by (\ref{eq:slasken}). Adding up
the above observations, we conclude that
\begin{equation*}
\begin{split}
\frac{d}{dt}\left(t^{-|\b|}E[f_{\b,l}]\right) = & -\frac{|\b|}{t}t^{-|\b|}E[f_{\b,l}]
+t^{-|\b|}\frac{d E[f_{\b,l}]}{dt}\leq
-\frac{|\b|}{t}t^{-|\b|}E[f_{\b,l}]
-\frac{3}{2t}t^{-|\b|}E[f_{\b,l}]\\
 & +2\is\irn{3}t^{-|\b|}\sfR_{\b,l}f_{\b,l}\a^{-1/2}t^{-3/2}dv d\theta
+Ct^{-2}t^{-|\b|}E[f_{\b,l}]\\
\leq & -\frac{3}{2t}t^{-|\b|}E[f_{\b,l}]+C_{k}t^{-3/2}E_{k}[f].
\end{split}
\end{equation*}
Summing over $\b$ and $l$, we obtain
\[
\frac{d E_{k}[f]}{dt}\leq -\frac{3}{2t}E_{k}[f]+C_{k}t^{-3/2}E_{k}[f].
\]
The lemma follows.
\end{proof}

In order to obtain a better understanding for the asymptotics, it is
convenient to rescale the distribution function according to
\[
f_{\sca}(t,\theta,v)=f(t,\theta,t^{-1/2}v).
\]
We have the following conclusions concerning $f_{\sca}$.

\begin{lemma}\label{lemma:fas}
Consider a $\tn{2}$-symmetric solution to the Einstein-Vlasov equations with a 
cosmological constant $\Lambda>0$ and existence interval $(t_{0},\infty)$, 
where $t_{0}\geq 0$. Assume that the solution has $\lambda$-asymptotics and
let $t_{1}=t_{0}+2$. Fix 
$0\leq k\in\zo$. Then there is a constant $C$, depending only on
the solution, such that in order for $(t,\theta,v)\in
[t_{1},\infty)\times\so\times\rn{3}$ to be in the support of $f_{\sca}$,
$v$ has to satisfy $|v|\leq C$. Moreover, there is a constant $C_{k}>0$,
depending only on $k$ and the solution to the Einstein-Vlasov
equations, such that
\[
\|\d_{t}f_{\sca}(t,\cdot)\|_{C^{k}(\so\times\rn{3})}
\leq C_{k}t^{-2}
\]
for all $t\geq t_{1}$. In particular, there is thus a smooth, non-negative
function with compact support, say $f_{\sca,\infty}$, on $\so\times\rn{3}$
such that
\[
\|f_{\sca}(t,\cdot)-f_{\sca,\infty}\|_{C^{k}(\so\times\rn{3})}
\leq C_{k}t^{-1}
\]
for all $t\geq t_{1}$.
\end{lemma}
\begin{proof}
The statement concerning the support is an immediate consequence of
Lemmas~\ref{lemma:vtthest} and \ref{lemma:Qocoest}. In order to derive 
the desired estimates, let us compute
\[
(\d_{t}f_{\sca})(t,\theta,v)=(\d_{t}f)(t,\theta,t^{-1/2}v)
-\frac{1}{2t}t^{-1/2}v^{i}(\d_{v^{i}}f)(t,\theta,t^{-1/2}v).
\]
Using the Vlasov equation, we conclude that
\[
\d_{t}f_{\sca}=-\frac{\a^{1/2}t^{-1/2}v^{1}}{\ldr{t^{-1/2}v}}
\d_{\theta}f_{\sca}+R_{\sca},
\]
where
\[
R_{\sca}(t,\theta,v)=
L^{i}(t,\theta,t^{-1/2}v)(\d_{v^{i}}f)(t,\theta,t^{-1/2}v)
=t^{1/2}L^{i}(t,\theta,t^{-1/2}v)(\d_{v^{i}}f_{\sca})(t,\theta,v)
\]
and the $L^{i}$ are defined in (\ref{eq:Lodef})--(\ref{eq:Lthdef}). Introducing
\[
L^{i}_{\sca}(t,\theta,v)=t^{1/2}L^{i}(t,\theta,t^{-1/2}v),
\]
we thus have
\begin{equation}\label{eq:fscaeq}
\d_{t}f_{\sca}=-\frac{\a^{1/2}t^{-1/2}v^{1}}{\ldr{t^{-1/2}v}}
\d_{\theta}f_{\sca}+L^{i}_{\sca}\d_{v^{i}}f_{\sca}.
\end{equation}
Due to the properties of the support of $f_{\sca}$, 
Lemma~\ref{lemma:indhcon} and Corollary~\ref{cor:tderbet}, there 
is a constant $C_{k}$ for each $0\leq k\in\zo$ such that
\begin{equation}\label{eq:Lisceq}
\sum_{l+|\b|\leq k}|(\d^{\b}_{v}\d^{l}_{\theta}L^{i}_{\sca})(t,\theta,v)|
\leq C_{k}t^{-2}
\end{equation}
for all $i=1,2,3$ and $(t,\theta,v)\in
[t_{1},\infty)\times\so\times\rn{3}$ in the support of $f_{\sca}$.
In order to estimate the derivatives of $f_{\sca}$ in $C^{k}$, it is
convenient to translate the estimate $E_{k}\leq C_{k}t^{-3/2}$ into an estimate
for $f_{\sca}$. However,
\begin{equation*}
\begin{split}
\is\irn{3}|(\d_{\theta}^{l}\d^{\b}_{v}f_{\sca})(t,\theta,v)|^{2}d\theta dv
 = & \is\irn{3}t^{-|\b|}
|(\d_{\theta}^{l}\d^{\b}_{v}f)(t,\theta,t^{-1/2}v)|^{2}d\theta dv\\
 = & t^{3/2}\is\irn{3}t^{-|\b|}
|(\d_{\theta}^{l}\d^{\b}_{v}f)(t,\theta,v)|^{2}d\theta dv
\leq Ct^{3/2}E_{k}(t)\leq C_{k},
\end{split}
\end{equation*}
assuming $l+|\b|\leq k$. Due to this estimate and Sobolev embedding, we
conclude that all derivatives of $f_{\sca}$ are bounded for $t\geq t_{1}$.
Combining this observation with (\ref{eq:fscaeq}) and (\ref{eq:Lisceq}),
we conclude that
\[
\sum_{l+|\b|\leq k}|\d^{\b}_{v}\d^{l}_{\theta}\d_{t}f_{\sca}|
\leq C_{k}t^{-2}
\]
for $t\geq t_{1}$. The lemma follows.
\end{proof}

\section{Proof of the main theorems}\label{section:stab}

Finally, we are in a position to prove the main theorems. Let us 
begin with Theorem~\ref{thm:as}.

\begin{proof}[Theorem~\ref{thm:as}]
To begin with, the conclusions concerning the distribution function 
are direct consequences of Lemma~\ref{lemma:fas}. Turning to $H$ and 
$G$, we have 
\begin{equation}\label{eq:HtGt}
H_{t}=-t^{-5/2}\a^{-1/2}e^{P+\lambda/2}(K-QJ),\ \ \
G_{t}=-QH_{t}-t^{-5/2}\a^{-1/2}e^{-P+\lambda/2}J;
\end{equation}
cf. (\ref{eq:JKdef}). In order to estimate $H_{t}$ and $G_{t}$, it is 
useful to note that $P$ and $Q$ are bounded in every $C^{k}$-norm for
$t\geq t_{0}$; this follows by integrating (\ref{eq:tderbet}). 
Combining this observation with (\ref{eq:Jt}), (\ref{eq:Kt}),
(\ref{eq:lhpreest}) and Lemma~\ref{lemma:indhcon} yields
\[
\|J_{t}\|_{C^{N}}+\|K_{t}\|_{C^{N}}\leq C_{N}t^{-2}(\ln t)^{m_{N}}.
\]
Thus $J$ and $K$ are uniformly bounded in $C^{N}$. Combining 
this observation with (\ref{eq:HtGt}), (\ref{eq:roughalest}),
(\ref{eq:lhpreest}), the bound on $P$ and $Q$ in every
$C^{N}$ norm and Lemma~\ref{lemma:indhcon}, we conclude that (\ref{eq:HGest}) 
holds. Due to Corollary~\ref{cor:tderbet}, we know that (\ref{eq:PQest}) and 
(\ref{eq:altest}) hold. Combining (\ref{eq:altest}) with 
(\ref{eq:lhpreest}) and (\ref{eq:roughalest}) yields (\ref{eq:alest}). 
Let us turn to the second fundamental form. By definition,
\[
\bk_{ij}=\bk(\d_{i},\d_{j})=\langle \nabla_{\d_{i}}e_{0},\d_{j}\rangle
=\langle \nabla_{\d_{i}}(t^{1/4}e^{-\lambda/4}\d_{t}),\d_{j}\rangle
=t^{1/4}e^{-\lambda/4}\langle \nabla_{\d_{i}}\d_{t},\d_{j}\rangle
=\frac{1}{2}t^{1/4}e^{-\lambda/4}\d_{t}g_{ij},
\]
where we have used the fact that $\d_{t}$ and $\d_{i}$ are perpendicular. 
In what follows, we would like to prove that 
\begin{equation}\label{eq:kmHg}
\|\bk_{ij}-\mH\bg_{ij}\|_{C^{N}}\leq C_{N},
\end{equation}
where $\mH=(\Lambda/3)^{1/2}$. Consider the spatial components of the metric
(\ref{eq:metric}). If a time derivative hits one of $P$, $Q$, $G$ or $H$ in
such a component, then the resulting expression is bounded in $C^{N}$ after
it has been multiplied by $t^{1/4}e^{-\lambda/4}/2$; this is due to 
(\ref{eq:HGest}), (\ref{eq:PQest}) and (\ref{eq:alest}). As a consequence, 
what we need to consider are the components of the tensor field
\begin{equation*}
\begin{split}
& \frac{1}{2}t^{1/4}e^{-\lambda/4}\left[
\left(-\frac{1}{2t}+\frac{1}{2}\lambda_{t}-\frac{\a_{t}}{\a}\right)
t^{-1/2}e^{\lambda/2}\a^{-1}d\theta^{2}+
e^{P}[dx+Qdy+(G+QH)d\theta]^{2}\right.\\
& \left.\phantom{\frac{1}{2}t^{5/4}}+
e^{-P}(dy+Hd\theta)^{2}\right]=\frac{1}{2}t^{-3/4}e^{-\lambda/4}\bar{g}
+\frac{1}{2}t^{-1/4}\a^{-1}e^{\lambda/4}\left[\frac{1}{2}\left(\lambda_{t}+\frac{3}{t}
\right)-\left(\frac{\a_{t}}{\a}+\frac{3}{t}\right)\right]d\theta^{2},
\end{split}
\end{equation*}
where $\bar{g}$ is the spatial part of the metric. Note that the components of 
the second term on the right hand side are bounded in $C^{N}$; this is a 
consequence of (\ref{eq:altest}) and (\ref{eq:alest}). Moreover,
\[
\frac{1}{2}t^{-3/4}e^{-\lambda/4}=\mH e^{-\hl/4},
\]
where $\hl$ is defined in (\ref{eq:hldef}). In order to prove (\ref{eq:kmHg}), 
it is thus sufficient to demonstrate that 
\[
\|(e^{-\hl/4}-1)\bg_{ij}\|_{C^{N}}\leq C_{N}.
\]
However, $t(e^{-\hl/4}-1)$ is bounded in $C^{N}$ due to (\ref{eq:alest})
and $t^{-1}\bg_{ij}$ is bounded in $C^{N}$ due to  (\ref{eq:HGest}),
(\ref{eq:PQest}) and (\ref{eq:alest}). Thus (\ref{eq:kmHg}) holds. Let us 
define $\bg_{\infty}$ by (\ref{eq:bginfdef}); note that this is a smooth 
Riemannian metric on $\tn{3}$. Moreover,
\begin{equation}\label{eq:rsmmlm}
\|t^{-1}\bar{g}_{ij}(t,\cdot)-\bar{g}_{\infty,ij}\|_{C^{N}}\leq C_{N}t^{-1}
\end{equation}
due to  (\ref{eq:HGest}), (\ref{eq:PQest}) and (\ref{eq:alest}).
Combining this estimate with (\ref{eq:kmHg}), we obtain (\ref{eq:bgbkest}).
The proof of future causal geodesic completeness is not very complicated, given
the above estimates. It is, e.g., possible to proceed as in the proof of 
\cite[Propositions~3 and 4, pp.~189--191]{mig4}. However, we shall not 
write down the details, since the result follows from the proof of 
Theorem~\ref{thm:stab}.
\end{proof}

Let us now turn to the proof of the cosmic no-hair conjecture. 

\begin{proof}[Theorem~\ref{thm:conoha}]
We need to verify that the conditions stated in Definition~\ref{def:fasdes}
are fulfilled. To begin with, note that $\Sigma_{t}=\{t\}\times\tn{3}$
is a Cauchy hypersurface for each $t\in (t_{0},\infty)$. An argument is 
required in order to justify this statement, but since the details are 
quite standard (cf., e.g., the proof of \cite[Proposition~20.3, p. 215]{minbok},
in particular \cite[p. 217]{minbok}), we have omitted the details. 
Let $\g=(\g^{0},\bga)$ be a future directed
and inextendible causal curve, defined on $I_{\g}=(s_{-},s_{+})$. Reparametrising
the curve, if necessary, we can assume that $\g^{0}(s)=s$ and that 
$I_{\g}=(t_{0},\infty)$. Due to the causality of the curve, we know that 
\[
\bg_{ij}[\g(t)]\dot{\bga}^{i}(t)\dot{\bga}^{j}(t)\leq -g_{00}[\g(t)]\leq
Ct^{-2}
\]
for $t\geq t_{1}$, where $t_{1}=t_{0}+2$ and we have used (\ref{eq:alest}). 
Combining this 
estimate with (\ref{eq:bgbkest}), we conclude that there is a 
constant $K_{0}>1$ (independent of the curve $\g$, as long as 
$\g^{0}(t)=t$) such that 
\[
\bg_{\infty,ij}[\bga(t)]\dot{\bga}^{i}(t)\dot{\bga}^{j}(t)\leq 
\frac{1}{4}K_{0}^{2}\mH^{-2}t^{-3}
\]
for all $t\geq t_{1}$. In particular, there is thus 
an $\bx_{0}\in\tn{3}$ such that 
$d_{\infty}[\bga(t),\bx_{0}]\leq K_{0}\mH^{-1}t^{-1/2}$ for all $t\geq t_{1}$, 
where $d_{\infty}$ is the topological metric on $\tn{3}$ induced by $\bg_{\infty}$.
Let $\einj>0$ denote the injectivity radius of $(\tn{3},\bg_{\infty})$; 
the injectivity radius of a point $p$ of a Riemannian manifold, denoted 
$\mathrm{inj}(p)$, is defined in \cite[Definition~9.2, p. 142]{petersen},
and the injectivity radius of a Riemannian manifold is the infimum of 
the injectivity radii of the points of the manifold; that $\einj>0$ follows
from the continuity of $\mathrm{inj}$, cf. \cite[p. 178]{petersen};
readers interested in a more quantitative bound on the injectivity radius
are referred to
\cite[Lemma~51, p.~319]{petersen}. Then, given $\bx\in\tn{3}$, there are 
normal coordinates on $B_{\einj}(\bx)$, where distances are computed 
using $d_{\infty}$; cf. \cite[pp.~72--73]{oneill} for a definition of
normal coordinates. Fix $t_{-}>K_{0}^{2}\mH^{-2}\einj^{-2}+1$
(note that $t_{-}$ is independent of the curve). Due to the above arguments
and definitions, 
\begin{equation}\label{eq:Jmgpsi}
J^{-}(\g)\cap J^{+}(\Sigma_{t_{-}})\subseteq\{(t,\bx)\in I\times\tn{3} :
t\geq t_{-}, \ d_{\infty}(\bx,\bx_{0})\leq K_{0}\mH^{-1}t^{-1/2}\}
\end{equation}
and the closed ball of radius $K_{0}\mH^{-1}t_{-}^{-1/2}$ (with respect to 
$d_{\infty}$) and 
centre $\bx_{0}$ is contained in the domain of definition of  
normal coordinates $\bsfx$ with centre at $\bx_{0}$. Denote the set appearing 
on the right hand side of (\ref{eq:Jmgpsi}) by $D_{t_{-},K_{0},\bx_{0}}$. Let us 
define
\[
\psi(\tau,\bxi)=[e^{2\mH\tau},\bsfx^{-1}(\bxi)].
\]
Then 
\[
\psi^{-1}(D_{t_{-},K_{0},\bx_{0}})=\{(\tau,\bxi)\in J\times\rn{3}:\tau\geq T_{0},\ \
|\bxi|\leq K_{0}\mH^{-1}e^{-\mH\tau}\},
\]
where $T_{0}=\mH^{-1}\ln t_{-}/2$, and $J=(\tau_{0},\infty)$, where 
$\tau_{0}=\mH^{-1}\ln t_{0}/2$; if $t_{0}=0$, then $\tau_{0}=-\infty$. Letting 
$T$ be slightly smaller than $T_{0}$
and $K$ be slightly larger than $K_{0}$, the map $\psi$ is still defined
on $C_{\Lambda,K,T}$; cf. (\ref{eq:clkt}). In analogy with 
Definition~\ref{def:fasdes}, let $D=\psi(C_{\Lambda,K,T})$ and 
$R(\tau)=K\mH^{-1}e^{-\mH\tau}$. Due to the above arguments, we have already 
verified all of the requirements of Definition~\ref{def:fasdes} 
(with $\Sigma=\Sigma_{t_{-}}$ etc.) but
the last one; i.e., that (\ref{eq:asdS}) hold.

In order to proceed, let $\bsfg_{\infty,ij}$ denote the components of 
$\bg_{\infty}$ with respect to the coordinates $\bsfx$. Let 
$\bsfg_{ij}(\tau,\cdot)$ and $\bsfk_{ij}(\tau,\cdot)$ denote the 
components of $\bg(e^{2\mH\tau},\cdot)$ and $\bk(e^{2\mH\tau},\cdot)$, respectively, 
with respect to the coordinates $\bsfx$. Moreover, consider 
$\bsfg_{\infty,ij}$, $\bsfg_{ij}(\tau,\cdot)$ and $\bsfk_{ij}(\tau,\cdot)$
to be functions on the image of $\bsfx$; i.e., on $B_{\einj}(0)$, with 
the origin corresponding to $\bx_{0}$. Note that the estimates 
(\ref{eq:kmHg}) and (\ref{eq:rsmmlm}) hold with $\bg_{ij}$ replaced by 
$\bsfg_{ij}$ etc., assuming the domain on which the $C^{N}$ norm is computed 
is suitably restricted. In particular, letting $S_{\tau}$ be as in 
Definition~\ref{def:fasdes}, the following estimate holds:
\[
\|e^{-2\mH\tau}\bsfk_{ij}(\tau,\cdot)-\mH\bsfg_{\infty,ij}\|_{C^{N}(S_{\tau})}+
\|e^{-2\mH\tau}\bsfg_{ij}(\tau,\cdot)-\bsfg_{\infty,ij}\|_{C^{N}(S_{\tau})}
\leq C_{N}e^{-2\mH\tau}
\]
for all $\tau\geq T$. Note that 
\[
\bsfg_{\infty,ij}(0)=\de_{ij},\ \ \
(\d_{l}\bsfg_{\infty,ij})(0)=0
\]
by the definition of the coordinates $\bsfx$. As a consequence, if $\bxi\in
S_{\tau}$, then 
\[
|(\d_{l}\bsfg_{\infty,ij})(\bxi)|=
\left|\int_{0}^{1}\frac{d}{ds}[(\d_{l}\bsfg_{\infty,ij})(s\bxi)]ds
\right|\leq Ce^{-\mH\tau}.
\]
Moreover, 
\[
|\bsfg_{\infty,ij}(\bxi)-\de_{ij}|\leq Ce^{-2\mH\tau}
\]
for $\tau\geq T$ and $\bxi\in S_{\tau}$. In particular, we thus have 
\[
\|e^{-2\mH\tau}\bsfk_{ij}(t,\cdot)-\mH\de_{ij}\|_{C^{0}(S_{\tau})}+
\|e^{-2\mH\tau}\bsfg_{ij}(t,\cdot)-\de_{ij}\|_{C^{0}(S_{\tau})}\leq Ce^{-2\mH\tau}
\]
for all $\tau\geq T$. Letting $\bg_{\dS}(\tau,\cdot)$ and 
$\bk_{\dS}(\tau,\cdot)$ be defined as in Definition~\ref{def:fasdes},
we conclude, in particular, that 
\[
\|\bg_{\dS}(\tau,\cdot)-\bg(\tau,\cdot)\|_{C^{0}_{\dS}(S_{\tau})}
+\|\bk_{\dS}(\tau,\cdot)-\bk(\tau,\cdot)\|_{C^{0}_{\dS}(S_{\tau})}\leq Ce^{-2\mH\tau}
\]
for all $\tau\geq T$. In fact, due to the above estimates, we have 
\[
\|\bg_{\dS}(\tau,\cdot)-\bg(\tau,\cdot)\|_{C^{N}_{\dS}(S_{\tau})}
+\|\bk_{\dS}(\tau,\cdot)-\bk(\tau,\cdot)\|_{C^{N}_{\dS}(S_{\tau})}\leq C_{N}e^{-2\mH\tau}
\]
for all $\tau\geq T$. The theorem follows. 
\end{proof}

Finally, we are in a position to prove Theorem~\ref{thm:stab}.

\begin{proof}[Theorem~\ref{thm:stab}]
The idea of the proof is to demonstrate that for late enough $t$, 
there is a neighbourhood of each point in $\{t\}\times\tn{3}$ such that 
Theorem~\ref{thm:main} applies in the neighbourhood; 
combining this observation with Cauchy stability, cf. 
Theorem~\ref{thm:caus}, then yields the desired result. 

Fixing $N$, there is an $\e>0$ and a constant $C_{N}$ such that for 
every $\bar{x}\in\tn{3}$, there are normal coordinates 
$\bar{\mathsf{x}}$ on $W=B_{\e}(\bar{x})$ with respect to $\bar{g}_{\infty}$,
where distances on $\tn{3}$ are measured using the topological metric induced
by $\bar{g}_{\infty}$. Moreover, if $\bar{g}_{\infty,ij}$ are the components of 
$\bar{g}_{\infty}$ with respect to $\bar{\mathsf{x}}$ and $\bar{g}_{\infty}^{ij}$ 
are the components of the inverse, then the derivatives of 
$\bar{g}_{\infty,ij}$ and $\bar{g}_{\infty}^{ij}$  up to order $N$ with respect to 
the coordinates $\bar{\mathsf{x}}$ on $W$ are bounded by $C_{N}$. 
Moreover, the derivatives of $\bar{\mathsf{x}}$, considered as functions
of $(\theta,x,y)$, up to order $N+1$ are bounded by $C_{N}$. Similarly, the
derivatives of $\bar{\mathsf{x}}^{-1}$ up to order $N+1$ are bounded by $C_{N}$.
The arguments required in order to prove the above statements are similar to 
those presented in the proof of \cite[Lemma~34.9, p.~650]{stab}. The
important point is that we obtain uniform bounds which hold regardless of the 
base point.

Define $K$ by the condition $e^{K}=4/\mH$ and define the coordinates 
$\bsfy=e^{-K}t^{1/2}\bsfx$ on $W$. Note that the range of $\bsfy$ is 
$B_{e^{-K}t^{1/2}\e}(0)$. For $t$ large enough (the bound being independent of the 
base point $\bx$), we consequently have $e^{-K}t^{1/2}\e>1$. From now on, we 
assume $t$ to be large enough that this is the case. Moreover, we assume the 
coordinates $\bsfy$ to be defined on the image of $B_{1}(0)$ under $\bsfy^{-1}$. 
Let $\bsfg_{ij}$ denote the components of $\bg(t,\cdot)$ with respect to the 
coordinates $\bsfy$. Moreover, let $\bsfg_{\infty,ij}$ denote the components 
of $\bg_{\infty}$ with respect to the coordinates $\bsfx$. Due to 
(\ref{eq:rsmmlm}), we have 
\[
|\d^{\a}_{\bxi}[(e^{-2K}\bsfg_{ij}-\bsfg_{\infty,ij})\circ\bsfy^{-1}](\bxi)|
\leq C_{N}t^{-1-|\a|/2}
\]
for $\bxi\in B_{1}(0)$ and $|\a|\leq N$; note that 
\[
\bsfy^{-1}(\bxi)=\bsfx^{-1}(e^{K}t^{-1/2}\bxi).
\]
Since $\bsfg_{\infty,ij}\circ\bsfy^{-1}(0)=\de_{ij}$, 
\begin{equation}\label{eq:firstver}
|e^{-2K}\bsfg_{ij}\circ\bsfy^{-1}-\de_{ij}|\leq C_{N}t^{-1/2}
\end{equation}
on $B_{1}(0)$; in particular, (\ref{eq:gijdeijsubest}) holds with a margin
for $t$ large enough. Similarly, 
\[
|\d_{m}(\bsfg_{\infty,ij}\circ\bsfy^{-1})(u)|=
\left|\sum_{l=1}^{n}\int_{0}^{1}\d_{l}\d_{m}(\bsfg_{\infty,ij}\circ\bsfy^{-1})(su)
u^{l}ds
\right|\leq Ct^{-1}
\]
on $B_{1}(0)$. To conclude, we thus have 
\begin{equation}\label{eq:secondver}
\|e^{-2K}(\d_{k}\bsfg_{ij}\circ\bsfy^{-1})\|_{C^{N-1}[B_{1}(0)]}\leq C_{N}t^{-1}.
\end{equation}
Due to (\ref{eq:kmHg}), we also have 
\begin{equation}\label{eq:thirdver}
\|(\bsfk_{ij}-\mH\bsfg_{ij})\circ\bsfy^{-1}\|_{C^{N}[B_{1}(0)]}\leq C_{N}t^{-1},
\end{equation}
where $\bsfk_{ij}$ denotes the components of the second fundamental form
$\bk$ calculated using the coordinates $\bsfy$. In the end, we shall 
choose $\krov=\ln t/2$. As a consequence, (\ref{eq:secondver}) and 
(\ref{eq:thirdver}) imply that (\ref{eq:maincon}) holds with a margin (note 
that the $\tn{2}$-symmetric background solution is such that $\bphi_{i}$, 
$i=0,1$, vanish), assuming $N\geq 5$; note that in order
to prove Theorem~\ref{thm:stab}, it is sufficient to apply
Theorem~\ref{thm:main} with $k_{0}=4$. 

Let us turn to the distribution function. First of all, recall that if 
$\Sigma$ is a spacelike hypersurface in a Lorentz manifold, and $\mff$ is a 
distribution function defined on the mass shell, then the initial datum for
the distribution function (denoted $\bmff$ and defined on $T\Sigma$) induced by 
$\mff$ on $\Sigma$ is given by 
\[
\bmff=\mff\circ\pros{\Sigma}^{-1},
\]
where $\pros{\Sigma}$ is the projection from the mass shell over $\Sigma$
to $T\Sigma$; in other words, if $p\in\mP_{r}$ for some $r\in\Sigma$
and $N_{r}$ is the future directed unit normal to $\Sigma$ at $r$, then
$\pros{\Sigma}(p)$ is the element of $T_{r}\Sigma$ corresponding
to $p+g(p,N_{r})N_{r}$.
In our case, we are interested in the hypersurface 
$\Sigma=\{ t\}\times\tn{3}$. If $z=(t,\theta,x,y)$, 
\[
\bmff\left(\bp^{i}\left. e_{i}\right|_{z}\right)=
\mff\left(p^{\a}\left. e_{\a}\right|_{z}\right)=
f(t,\theta,\bp),
\]
where $\bp=(\bp^{1},\bp^{2},\bp^{3})$, 
\[
p^{0}=[1+(\bp^{1})^{2}+(\bp^{2})^{2}+(\bp^{3})^{2}]^{1/2}
\]
and $p^{i}=\bp^{i}$. However, in the application of 
Theorem~\ref{thm:main}, we need to express
$\bmff$ with respect to the coordinates $\bsfy$. Consequently, we are 
interested in 
\[
\fb(\bz,\bp)=\bmff\left(\bp^{i}\left. \d_{\bsfy^{i}}\right|_{z}\right)
=\bmff\left(\bp^{i}A_{i}^{\phantom{i}j}(\bz)\left. e_{j}\right|_{z}\right)
=f[t,\theta,v(\bz,\bp)],
\]
where $z=(t,\bz)$, 
\[
v(\bz,\bp)=(\bp^{i}A_{i}^{\phantom{i}1}(\bz),\bp^{i}A_{i}^{\phantom{i}2}(\bz),
\bp^{i}A_{i}^{\phantom{i}3}(\bz))
\]
and $A_{i}^{\phantom{i}j}$ is defined by the requirement that 
\[
\left. \d_{\bsfy^{i}}\right|_{z}=
A_{i}^{\phantom{i}j}(\bz)\left. e_{j}\right|_{z}.
\]
Thus
\[
A_{i}^{\phantom{i}j}(\bz)=\left\langle\left. \d_{\bsfy^{i}}\right|_{z},
\left. e_{j}\right|_{z}\right\rangle
=e^{K}t^{-1/2}\left\langle\left. \d_{\bsfx^{i}}\right|_{z},
\left. e_{j}\right|_{z}\right\rangle
=e^{K}t^{-1/2}
\frac{\d\bsfz^{l}}{\d\bsfx^{i}}(z)
\left\langle\left. \d_{\bsfz^{l}}\right|_{z},
\left. e_{j}\right|_{z}\right\rangle,
\]
where $\bsfz$ correspond to the standard coordinates on the torus (which 
are locally well defined). In particular, $\d_{\bsfz^{1}}=\d_{\theta}$,
$\d_{\bsfz^{2}}=\d_{x}$ and $\d_{\bsfz^{3}}=\d_{y}$. Due to the observations
made at the beginning of the proof, (\ref{eq:ONframe}) and 
Theorem~\ref{thm:as}, it is clear that all derivatives of 
$A_{i}^{\phantom{i}j}$ up to order $N$ are uniformly bounded on the domain of 
$\bsfy$, the bound being independent of the base point $\bx$ and the time
$t$ (assuming $t$ to be sufficiently large). What we need to estimate is 
\[
\sum_{|\a|+|\b|\leq k_{0}}\irn{3}\int_{\bsfy(U)}(e^{-w})^{2|\b|}\langle
e^{w}\bp \rangle^{2\mu+2|\b|}|\d^{\a}_{\bxi}
\d^{\b}_{\bp}\sffb_{\bsfy}|^{2}(\bxi,\bp)d\bxi d\bp,
\]
cf. (\ref{eq:flocnormmth}) and (\ref{eq:fbcond}), where 
\[
\sffb_{\bsfy}(\bxi,\bp)=\fb[\bsfy^{-1}(\bxi),\bp]
=\fb[\bsfx^{-1}(e^{K}t^{-1/2}\bxi),\bp]
=f[t,\bsfz^{1}\circ\bsfx^{-1}(e^{K}t^{-1/2}\bxi),
v(\bsfx^{-1}(e^{K}t^{-1/2}\bxi),\bp)]
\]
and the constant $w$ remains to be specified. 
Note that all derivatives of $\bsfx^{-1}$ and $\bsfz^{1}\circ\bsfx^{-1}$
up to order $N$ are uniformly bounded. As a consequence, 
\[
\d^{\b}_{\bp}\sffb_{\bsfy}(\bxi,\bp)
=\sum_{|\g|=|\b|}\d^{\g}_{v}f[t,\bsfz^{1}\circ\bsfx^{-1}(e^{K}t^{-1/2}\bxi),
v(\bsfx^{-1}(e^{K}t^{-1/2}\bxi),\bp)]\psi_{\g}(t^{-1/2}\bxi)
\]
for functions $\psi_{\g}$ with bounded derivatives; note that 
\[
\frac{\d v^{i}}{\d \bp^{j}}(\bsfx^{-1}(e^{K}t^{-1/2}\bxi),\bp)=
A_{j}^{\phantom{j}i}\circ\bsfx^{-1}(e^{K}t^{-1/2}\bxi).
\]
As a consequence, $\d^{\a}_{\bxi}\d^{\b}_{\bp}\sffb_{\bsfy}(\bxi,\bp)$
consists of sums of terms of the form 
\[
t^{-|\a|/2}
\d_{\theta}^{l}\d^{\g+\de}_{v}f[t,\bsfz^{1}\circ\bsfx^{-1}(e^{K}t^{-1/2}\bxi),
v(\bsfx^{-1}(e^{K}t^{-1/2}\bxi),\bp)]\phi_{\g,\de,l}(t^{-1/2}\bxi)\bp^{\lambda},
\]
where $|\lambda|=|\de|$, $|\g|=|\b|$, $l+|\de|\leq |\a|$, $\phi_{\g,\de,l}$ are 
bounded functions and $\bp^{\lambda}=(\bp^{1})^{\lambda_{1}}(\bp^{2})^{\lambda_{2}}
(\bp^{3})^{\lambda_{3}}$. On the other hand, we know that 
\[
f(t,\theta,v)=f_{\sca}(t,\theta,t^{1/2}v),
\]
where $f_{\sca}$ converges to a smooth function with compact support with 
respect to every every $C^{k}$-norm. Moreover, we know that $f_{\sca}$ has
uniformly compact support. Note also that
\[
\d_{\theta}^{l}\d^{\g+\de}_{v}f(t,\theta,v)
=t^{(|\g|+|\de|)/2}
\d_{\theta}^{l}\d^{\g+\de}_{v}f_{\sca}(t,\theta,t^{1/2}v).
\]
Since there is a uniform constant $C>1$ (independent of $t$ (large enough)
and the base point $\bx$) such that 
\[
C^{-1}|\bp|\leq |v(\bsfx^{-1}(e^{K}t^{-1/2}\bxi),\bp)|\leq C|\bp|,
\]
we conclude that 
\begin{equation*}
\begin{split}
&t^{-|\a|/2}
|\d_{\theta}^{l}\d^{\g+\de}_{v}f[t,\bsfz^{1}\circ\bsfx^{-1}(e^{K}t^{-1/2}\bxi),
v(\bsfx^{-1}(e^{K}t^{-1/2}\bxi),\bp)]\phi_{\g,\de,l}(t^{-1/2}\bxi)\bp^{\lambda}|\\
\leq & C_{\a,\b}t^{(|\g|-|\a|)/2}\chi(t^{1/2}\bp),
\end{split}
\end{equation*}
where $\chi$ is a smooth function with compact support. As a consequence, 
\[
|\d^{\a}_{\bxi}\d^{\b}_{\bp}\sffb_{\bsfy}(\bxi,\bp)|\leq 
C_{\a,\b}t^{(|\b|-|\a|)/2}\chi(t^{1/2}\bp)
\]
for $\bxi\in B_{1}(0)$. Let us now define $w=K+\krov$, where 
$\krov=\ln t/2$. Then 
\begin{equation*}
\begin{split}
 & \irn{3}\int_{\bsfy(U)}(e^{-w})^{2|\b|}\langle
e^{w}\bp \rangle^{2\mu+2|\b|}|\d^{\a}_{\bxi}
\d^{\b}_{\bp}\sffb_{\bsfy}|^{2}(\bxi,\bp)d\bxi d\bp \\
\leq & C_{\a,\b}^{2}\irn{3}\int_{B_{1}(0)}e^{-2|\b|K}t^{-|\b|}
\langle e^{K}t^{1/2}\bp \rangle^{2\mu+2|\b|}
t^{|\b|-|\a|}\chi^{2}(t^{1/2}\bp)d\bxi d\bp\leq
C_{\mu,\a,\b}t^{-|\a|-3/2}.
\end{split}
\end{equation*}
The square root of the right hand side of this expression should be compared 
with the right hand side of (\ref{eq:fbcond}):
\[
\mH^{2}\varepsilon^{5/2}e^{-3K/2-\krov}
=\mH^{2}\varepsilon^{5/2}e^{-3K/2}t^{-1/2}.
\]
Clearly, we have a margin. As a consequence, for $t$ large enough there is,
for every $\bx\in\tn{3}$, a neighbourhood of $\bx$ such that (\ref{eq:fbcond})
holds with a margin. Due to (\ref{eq:firstver}), (\ref{eq:secondver}) and 
(\ref{eq:thirdver}), we also know that (\ref{eq:gijdeijsubest}) and 
(\ref{eq:maincon}) hold with a margin. We can thus apply 
Theorem~\ref{thm:main} with the given $\mu>5/2$ and $k_{0}=4$. 
In addition, the covering of $\tn{3}$ obtained by taking the neighbourhoods
$\bsfy^{-1}[B_{1/4}(0)]$ corresponding to varying base points $\bx$ has
a finite subcovering. Appealing to Cauchy stability, Theorem~\ref{thm:caus},
we conclude that there is an $\e>0$ with the properties stated in the theorem. 
\end{proof}

\section*{Acknowledgements}

The second author would like to acknowledge the support of the G\"{o}ran 
Gustafsson Foundation, the G\"{o}ran Gustafsson Foundation for Research in 
Natural Sciences and Medicine, and the Swedish Research Council. This article 
was in part written during a stay of both authors at the Erwin Schr\"{o}dinger 
Institute in Vienna during the programme 'Dynamics of General Relativity: 
Black Holes and Asymptotics'.

\appendix

\section{Derivation of the equations}\label{section:derofequations}

The purpose of this appendix is to compute the Einstein tensor associated with 
the metric (\ref{eq:metric}), and to derive an expression for the Vlasov 
equation. Let us begin by expressing $g$ with respect to suitable one-form 
fields. Let 
\begin{eqnarray*}
\xi^{0} & = & t^{-1/4}e^{\lambda/4}dt,\\
\xi^{1} & = & t^{-1/4}e^{\lambda/4}\a^{-1/2}d\theta,\\
\xi^{2} & = & t^{1/2}e^{P/2}(dx+Qdy+(G+QH)d\theta),\\
\xi^{3} & = & t^{1/2}e^{-P/2}(dy+Hd\theta).
\end{eqnarray*}
With respect to these one-form fields, the metric can be written 
\[
g=-\xi^{0}\otimes\xi^{0}+\sum_{i=1}^{3}\xi^{i}\otimes\xi^{i}.
\]
Using the orthonormal frame $\{e_{\a}\}$ introduced in (\ref{eq:ONframe}),
it can be verified that $\xi^{\a}(e_{\b})=\de^{\a}_{\b}$. 

\subsection{Commutators}

Let us compute the commutators, in other words the functions $\g^{\a}_{\b\zeta}$
such that 
\[
[e_{\b},e_{\zeta}]=\g^{\a}_{\b\zeta}e_{\a}.
\]
Clearly, $\g^{\a}_{\b\zeta}=-\g^{\a}_{\zeta\b}$. Consequently, it is sufficient to 
compute $\g^{\a}_{\b\zeta}$ for $\b<\zeta$. By a straightforward computation, 
\begin{eqnarray}
\g^{0}_{01} & = & \frac{1}{4}t^{1/4}e^{-\lambda/4}\a^{1/2}\lambda_{\theta},\\
\g^{1}_{01} & = & -\frac{1}{4}t^{1/4}e^{-\lambda/4}
\left(\lambda_{t}-2\frac{\a_{t}}{\a}-\frac{1}{t}\right),\\
\g^{2}_{01} & = & t^{-3/2}e^{-P/2}J,\label{eq:gtzo}\\
\g^{3}_{01} & = & t^{-3/2}e^{P/2}(K-QJ),\label{eq:gthzo}
\end{eqnarray}
where we have used the notation (\ref{eq:JKdef}). 
Turning to $\g^{\a}_{02}$, the only $\a$ for which this object is non-zero 
is $\a=2$, and we then have 
\[
\g^{2}_{02}=-\frac{1}{2}t^{1/4}e^{-\lambda/4}\left(t^{-1}+P_{t}\right).
\]
Turning to $\g^{\a}_{03}$, the only $\a$'s for which this object is non-zero 
is $\a=2$ and $\a=3$, and we then have 
\begin{eqnarray*}
\g^{2}_{03} & = & -t^{1/4}e^{-\lambda/4}e^{P}Q_{t},\\
\g^{3}_{03} & = & -\frac{1}{2}t^{1/4}e^{-\lambda/4}\left(t^{-1}-P_{t}\right).
\end{eqnarray*}
Turning to $\g^{\a}_{12}$, the only $\a$ for which this object is non-zero 
is $\a=2$, and we then have 
\[
\g^{2}_{12}=-\frac{1}{2}t^{1/4}e^{-\lambda/4}\a^{1/2}P_{\theta}.
\]
Turning to $\g^{\a}_{13}$, the only $\a$'s for which this object is non-zero 
is $\a=2$ and $\a=3$, and we then have 
\begin{eqnarray*}
\g^{2}_{13} & = & -t^{1/4}e^{-\lambda/4}\a^{1/2}e^{P}Q_{\theta},\\
\g^{3}_{13} & = & \frac{1}{2}t^{1/4}e^{-\lambda/4}\a^{1/2}P_{\theta}.
\end{eqnarray*}
Finally, $\g^{\a}_{23}=0$. For future reference, let us make the following 
observations:
\begin{eqnarray*}
\g^{A}_{0A} & = & -t^{-3/4}e^{-\lambda/4},\\
\g^{i}_{0i} & = & \g^{1}_{01}-t^{-3/4}e^{-\lambda/4},\\
\g^{\a}_{0\a} & = & \g^{1}_{01}-t^{-3/4}e^{-\lambda/4},\\
\g^{A}_{1A} & = & 0,\\
\g^{i}_{1i} & = & 0,\\
\g^{\a}_{1\a} & = & -\g^{0}_{01},
\end{eqnarray*}
where Greek indices range from $0$ to $3$, lower case Latin indices range 
from $1$ to $3$ and capital Latin indices range from $2$ to $3$; moreover, 
Einstein's summation convention is enforced. 

\textbf{Non-zero components.} Note that in order for $\g^{\a}_{\b\zeta}$ to 
be non-zero, one of the following conditions has to be satisfied:
\begin{itemize}
\item $\{\b,\zeta\}=\{0,1\}$,
\item one of $\b$ and $\zeta$ is in $\{0,1\}$; the other is in $\{2,3\}$;
and $\a$ is in $\{2,3\}$.
\end{itemize}

\subsection{Connection coefficients}

Define the connection coefficients $\G^{\a}_{\b\zeta}$ by the relation
\[
\nabla_{e_{\b}}e_{\zeta}=\G^{\a}_{\b\zeta}e_{\a},
\]
where $\nabla$ is the Levi-Civita connection associated with the metric $g$.
Note that 
\[
\G^{0}_{\b\zeta}=-\langle\nabla_{e_{\b}}e_{\zeta},e_{0}\rangle,\ \ \
\G^{i}_{\b\zeta}=\langle\nabla_{e_{\b}}e_{\zeta},e_{i}\rangle,
\]
since the frame is orthonormal. Let us begin by recording some symmetries of 
this object. 

\textbf{Symmetries of the connection coefficients.} Due to the metricity
of the connection and the fact that the basis is orthonormal, we have 
\[
\G^{\a}_{\b\a}=0\ \ \mathrm{(no\ summation\ over}\ \a).
\]
For similar reasons, $\G^{i}_{\a j}$ is anti-symmetric
in the indices $i$ and $j$. Moreover, since $[e_{i},e_{j}]$ is perpendicular 
to $e_{0}$, we have 
\[
\G^{i}_{j0}=\langle\nabla_{e_{j}}e_{0},e_{i}\rangle
=-\langle e_{0},\nabla_{e_{j}}e_{i}\rangle
=-\langle e_{0},\nabla_{e_{i}}e_{j}\rangle
=\langle\nabla_{e_{i}}e_{0},e_{j}\rangle=\G^{j}_{i0}.
\]
Thus, $\G^{i}_{j0}$ is symmetric in $i$ and $j$. Note that the computation
also demonstrates that $\G^{i}_{j0}=\G^{0}_{ij}$. Similarly, since 
$[e_{A},e_{B}]=0$, we have 
\[
\G^{A}_{B1}=\langle\nabla_{e_{B}}e_{1},e_{A}\rangle
=-\langle e_{1},\nabla_{e_{B}}e_{A}\rangle
=-\langle e_{1},\nabla_{e_{A}}e_{B}\rangle
=\langle\nabla_{e_{A}}e_{1},e_{B}\rangle=\G^{B}_{A1}.
\]
In particular, $\G^{A}_{B1}$ is symmetric in the indices $A$ and $B$, and 
$\G^{A}_{B1}=-\G^{1}_{AB}$. 

\textbf{Connection coefficients including two or more zero indices.} That 
$\G^{0}_{\a 0}=0$ follows from the above. Moreover, using the Koszul formula, 
it can be computed that 
\[
\G^{0}_{01}=\G^{1}_{00}=\g^{0}_{01}
\]
and that the remaining components satisfy $\G^{0}_{0A}=\G^{A}_{00}=0$. 

\textbf{Connection coefficients including exactly one zero index.}
As already mentioned, $\G^{i}_{0j}$ is antisymmetric and 
$\G^{i}_{j0}=\G^{0}_{ij}$ is symmetric. It is thus sufficient to 
compute $\G^{i}_{0j}$ for $i<j$ and $\G^{i}_{j0}$ for $i\leq j$. 
We have 
\begin{eqnarray*}
\G^{1}_{0A} & = & -\frac{1}{2}\g^{A}_{01},\\
\G^{2}_{03} & = & \frac{1}{2}\g^{2}_{03}.
\end{eqnarray*}
Moreover, 
\begin{eqnarray*}
\G^{i}_{i0} & = & -\g^{i}_{0i}\ \ \ \mathrm{(no\ summation\ over}\ i),\\
\G^{1}_{A0} & = & -\frac{1}{2}\g^{A}_{01},\\
\G^{2}_{30} & = & -\frac{1}{2}\g^{2}_{03}.
\end{eqnarray*}

\textbf{Connection coefficients including no zero index.}
Note that, due to the Koszul formula and the properties of the 
commutators, the only $\G^{i}_{jk}$'s which are non-zero are the 
ones that have one index equalling $1$ and two indices in $\{2,3\}$.  
Moreover, since $\G^{B}_{A1}=-\G^{1}_{AB}$ is symmetric and $\G^{A}_{1B}$ is 
antisymmetric, it is sufficient to calculate that 
\begin{eqnarray*}
\G^{2}_{13} & = & \frac{1}{2}\g^{2}_{13},\\
\G^{A}_{A1} & = & -\g^{A}_{1A}\ \ \ \mathrm{(no\ summation\ over}\ A),\\
\G^{2}_{31} & = & -\frac{1}{2}\g^{2}_{13}.
\end{eqnarray*}

For future reference, it is of interest to note that 
\begin{eqnarray}
\G^{\a}_{\a 1} & = & \g^{0}_{01},\label{eq:gaao}\\
\G^{i}_{i 1} & = & 0,\label{eq:giio}\\
\G^{\a}_{\a 0} & = & -\g^{1}_{01}+t^{-3/4}e^{-\lambda/4},\label{eq:gaaz}\\
\G^{\a}_{\a A} & = & 0.\label{eq:gaaA}
\end{eqnarray}

\subsection{Twist quantities}\label{ssection:twqua}

The quantities $J$ and $K$ have been defined in two different
ways; one definition is given in (\ref{eq:twq}) and another is given in
(\ref{eq:JKdef}). In the present subsection, we wish
to verify that these two definitions yield the same result. In the proof of 
this statement, it is useful to introduce different notation for the different
definitions. Let us therefore denote the $J$ and the $K$ defined in 
(\ref{eq:twq}) by $J_{\tw}$ and $K_{\tw}$ respectively. The quantities defined 
in (\ref{eq:JKdef}) we, however, still refer to as
$J$ and $K$. Due to the fact that we have calculated the connection 
coefficients using the orthonormal frame $\{e_{\a}\}$, it is convenient
to carry out the computations relative to this frame. Note, 
to this end, that 
\begin{eqnarray*}
X & = & t^{1/2}e^{P/2}e_{2},\\
Y & = & t^{1/2}e^{-P/2}e_{3}+t^{1/2}e^{P/2}Qe_{2},
\end{eqnarray*}
where $X=\d_{x}$ and $Y=\d_{y}$. In particular, if $X^{\a}$ and $Y^{\a}$ denote
the components of $X$ and $Y$ with respect to the frame $\{e_{\a}\}$, then
\begin{itemize}
\item $X^{2}=t^{1/2}e^{P/2}$ and $X^{\a}=0$ for $\a\neq 2$,
\item $Y^{2}=t^{1/2}e^{P/2}Q$, $Y^{3}=t^{1/2}e^{-P/2}$ and $Y^{0}=Y^{1}=0$. 
\end{itemize}
Consequently, 
\begin{equation}\label{eq:Jtw}
J_{\tw}=\e_{\a\b\zeta\de}X^{\a}Y^{\b}\nabla^{\zeta}X^{\de}
=\e_{23\zeta\de}X^{2}Y^{3}\nabla^{\zeta}X^{\de}=t(\nabla^{0}X^{1}-\nabla^{1}X^{0})
=-t(\nabla_{0}X^{1}+\nabla_{1}X^{0}),
\end{equation}
where the indices are frame indices, and we assume the orientation of 
$M$ to be such that $\e_{0123}=1$. Similarly,
\begin{equation}\label{eq:Ktw}
K_{\tw}=-t(\nabla_{0}Y^{1}+\nabla_{1}Y^{0}).
\end{equation}
What remains is thus to calculate $\nabla_{0}X^{1}$ etc. However, 
\[
\nabla_{\a}X^{\b}=\xi^{\b}(\nabla_{e_{\a}}X)
=\xi^{\b}[e_{\a}(X^{\zeta})e_{\zeta}+X^{\zeta}\nabla_{e_{\a}}e_{\zeta}]
=e_{\a}(X^{\b})+X^{\zeta}\Gamma_{\a\zeta}^{\b}.
\]
The calculation for $Y$ is the same. In particular, 
\[
\nabla_{0}X^{1}+\nabla_{1}X^{0}=X^{\zeta}(\Gamma_{0\zeta}^{1}+\Gamma_{1\zeta}^{0})
=X^{2}(\Gamma_{02}^{1}+\Gamma_{12}^{0})=-X^{2}\g^{2}_{01}.
\]
Combining this calculation with (\ref{eq:gtzo}) and (\ref{eq:Jtw})  yields 
\[
J_{\tw}=tX^{2}\g^{2}_{01}=t^{3/2}e^{P/2}t^{-3/2}e^{-P/2}J=J.
\]
Thus $J=J_{\tw}$. Next, let us calculate
\[
\nabla_{0}Y^{1}+\nabla_{1}Y^{0}=Y^{\zeta}(\Gamma_{0\zeta}^{1}+\Gamma_{1\zeta}^{0})
=Y^{2}(\Gamma_{02}^{1}+\Gamma_{12}^{0})
+Y^{3}(\Gamma_{03}^{1}+\Gamma_{13}^{0})=-Y^{2}\g^{2}_{01}-Y^{3}\g^{3}_{01}.
\]
Combining this observation with (\ref{eq:gtzo}), (\ref{eq:gthzo}) and 
(\ref{eq:Ktw}) yields
\[
K_{\tw}=tY^{2}\g^{2}_{01}+tY^{3}\g^{3}_{01}
=t^{3/2}e^{P/2}Qt^{-3/2}e^{-P/2}J+t^{3/2}e^{-P/2}t^{-3/2}e^{P/2}(K-QJ)=K.
\]
To conclude, we thus have $J_{\tw}=J$ and $K_{\tw}=K$.

\subsection{Auxiliary computations}

In order to simplify future calculations, let us make some observations
concerning the derivatives of the connection coefficients. Note, to begin
with, that all $\g^{i}_{0i}$ (no summation) can be written 
\[
\g^{i}_{0i}=-t^{1/4}e^{-\lambda/4}f_{i}
\]
for suitably chosen functions $f_{i}$. Thus, it can be computed that 
\begin{equation*}
\begin{split}
-e_{0}(\g^{i}_{0i}) = & -t^{1/4}e^{-\lambda/4}\d_{t}\left(-\exp\left(-\frac{1}{4}
\lambda+\frac{1}{2}\ln\a-\frac{3}{4}\ln t\right)t\a^{-1/2}f_{i}\right)\\
 = & \frac{1}{4}t^{1/4}e^{-\lambda/4}
\left(\lambda_{t}-2\frac{\a_{t}}{\a}-\frac{1}{t}\right)
\g^{i}_{0i}+t^{1/4}e^{-\lambda/4}t^{-1}\g^{i}_{0i}+t^{-1/2}e^{-\lambda/2}
\a^{1/2}\d_{t}(t\a^{-1/2}f_{i}),
\end{split}
\end{equation*}
where we do not sum over $i$. However, the factor in front of $\g^{i}_{0i}$
in the first term is given by $-\g^{1}_{01}$ and the factor in front of 
$\g^{i}_{0i}$ in the second term is given by $-\g^{2}_{02}-\g^{3}_{03}$. 
Consequently,
\[
e_{0}(\G^{0}_{ii})+\G^{\a}_{\a 0}\G^{0}_{ii}=
t^{-1/2}e^{-\lambda/2}\a^{1/2}\d_{t}(t\a^{-1/2}f_{i}),
\]
where we sum over $\a$ but not over $i$. In particular, we thus have
\begin{eqnarray}
e_{0}(\G^{0}_{11})+\G^{\a}_{\a 0}\G^{0}_{11} & = & 
\frac{1}{4}t^{-1/2}e^{-\lambda/2}\a^{1/2}\d_{t}[t\a^{-1/2}(\lambda_{t}
-2\a^{-1}\a_{t}-t^{-1})],\label{eq:ezla}\\
e_{0}(\G^{0}_{22})+\G^{\a}_{\a 0}\G^{0}_{22} & = & 
\frac{1}{2}t^{-1/2}e^{-\lambda/2}\a^{1/2}\d_{t}[t\a^{-1/2}(t^{-1}+P_{t})],
\label{eq:ezpp}\\
e_{0}(\G^{0}_{33})+\G^{\a}_{\a 0}\G^{0}_{33} & = & 
\frac{1}{2}t^{-1/2}e^{-\lambda/2}\a^{1/2}\d_{t}[t\a^{-1/2}(t^{-1}-P_{t})].\label{eq:ezpm}
\end{eqnarray}
Next, note that $\g^{0}_{01}$, $\g^{2}_{12}$ and $\g^{3}_{13}$ all can be 
written as $h_{i}=t^{1/4}e^{-\lambda/4}f_{i}$ for suitably chosen functions $f_{i}$. 
Moreover,
\[
e_{1}(h_{i})=t^{1/4}e^{-\lambda/4}\a^{1/2}
\d_{\theta}(t^{1/4}e^{-\lambda/4}f_{i})=t^{1/2}e^{-\lambda/2}\a^{1/2}\d_{\theta}f_{i}
-\g^{0}_{01}h_{i}.
\]
Since $\G^{\a}_{\a 1}=\g^{0}_{01}$, we conclude that 
\begin{eqnarray}
e_{1}(\G^{1}_{00})+\G^{\a}_{\a 1}\G^{1}_{00} & = & 
\frac{1}{4}t^{-1/2}e^{-\lambda/2}\a^{1/2}\d_{\theta}(t\a^{1/2}\lambda_{\theta}),
\label{eq:eola}\\
e_{1}(\G^{1}_{22})+\G^{\a}_{\a 1}\G^{1}_{22} & = & 
-\frac{1}{2}t^{-1/2}e^{-\lambda/2}\a^{1/2}\d_{\theta}(t\a^{1/2}P_{\theta}),
\label{eq:eott}\\
e_{1}(\G^{1}_{33})+\G^{\a}_{\a 1}\G^{1}_{33} & = & 
\frac{1}{2}t^{-1/2}e^{-\lambda/2}\a^{1/2}\d_{\theta}(t\a^{1/2}P_{\theta}).
\label{eq:eothth}
\end{eqnarray}
The expressions $\g^{2}_{03}$ and $\g^{2}_{13}$ require a somewhat different
treatment. However, similar arguments yield
\begin{eqnarray}
e_{1}(\G^{1}_{23})+\G^{1}_{23}\G^{\a}_{\a 1} & = & 
-\frac{1}{2}t^{-1/2}e^{-\lambda/2-P}\a^{1/2}\d_{\theta}(t\a^{1/2}
e^{2P}Q_{\theta})-\frac{1}{2}(\g^{3}_{13}-\g^{2}_{12})\g^{2}_{13},
\label{eq:eogotth}\\ 
e_{0}(\G^{0}_{23}) +\G^{0}_{23}\G^{\a}_{\a 0} & = & 
\frac{1}{2}t^{-1/2}e^{-\lambda/2-P}\a^{1/2}
\d_{t}(t\a^{-1/2}e^{2P}Q_{t})
+\frac{1}{2}(\g^{3}_{03}-\g^{2}_{02})\g^{2}_{03}.
\label{eq:ezgztth}
\end{eqnarray}

\subsection{Ricci curvature}

The Ricci curvature is given by 
\begin{equation*}
\begin{split}
\Ric(e_{\b},e_{\zeta}) = & \sum_{\a}\e_{\a}\langle R_{e_{\a}e_{\b}}e_{\zeta},e_{\a}\rangle
=\sum_{\a}\e_{\a}\langle \nabla_{e_{\a}}\nabla_{e_{\b}}e_{\zeta}
-\nabla_{e_{\b}}\nabla_{e_{\a}}e_{\zeta}-\nabla_{[e_{\a},e_{\b}]}e_{\zeta},e_{\a}\rangle\\
 = & \sum_{\a}\e_{\a}\langle \nabla_{e_{\a}}(\G_{\b\zeta}^{\de}e_{\de})
-\nabla_{e_{\b}}(\G_{\a\zeta}^{\de}e_{\de})-\g^{\de}_{\a\b}\G_{\de\zeta}^{\lambda}e_{\lambda}
,e_{\a}\rangle\\
 = & \sum_{\a}\e_{\a}\langle e_{\a}(\G_{\b\zeta}^{\de})e_{\de}
+\G_{\b\zeta}^{\de}\G_{\a\de}^{\lambda}e_{\lambda}
-e_{\b}(\G_{\a\zeta}^{\de})e_{\de}-\G_{\a\zeta}^{\de}\G_{\b\de}^{\lambda}e_{\lambda}
-\g^{\de}_{\a\b}\G_{\de\zeta}^{\lambda}e_{\lambda},e_{\a}\rangle\\
 = & e_{\a}(\G_{\b\zeta}^{\a})+\G_{\b\zeta}^{\de}\G_{\a\de}^{\a}
-e_{\b}(\G_{\a\zeta}^{\a})-\G_{\a\zeta}^{\de}\G_{\b\de}^{\a}
-\g^{\de}_{\a\b}\G_{\de\zeta}^{\a},
\end{split}
\end{equation*}
where $\e_{0}=-1$ and $\e_{i}=1$. Let us begin by computing the $00$ 
component:
\begin{equation*}
\begin{split}
\Ric(e_{0},e_{0}) = & e_{\a}(\G_{00}^{\a})+\G_{00}^{\de}\G_{\a\de}^{\a}
-e_{0}(\G_{\a 0}^{\a})-\G_{\a 0}^{\de}\G_{0\de}^{\a}
-\g^{\de}_{\a 0}\G_{\de 0}^{\a}\\
 = & e_{1}(\G_{00}^{1})+\G_{00}^{1}\G_{\a 1}^{\a}
-e_{0}(\G_{\a0}^{\a})-\G_{\a0}^{\b}\G_{0\b}^{\a}
-\g^{\b}_{\a 0}\G_{\b 0}^{\a}.
\end{split}
\end{equation*}
In order to simplify the expression, let us note that 
\[
e_{1}(\G_{00}^{1})+\G_{00}^{1}\G_{\a 1}^{\a}
=\frac{1}{4}t^{-1/2}e^{-\lambda/2}\a^{1/2}\d_{\theta}(t\a^{1/2}\lambda_{\theta}),
\]
where we have used (\ref{eq:eola}). Moreover,
due to (\ref{eq:gaaz}),
\[
-e_{0}(\G_{\a0}^{\a})=-e_{0}(-\g^{1}_{01}+t^{-3/4}e^{-\lambda/4})
=e_{0}(\g^{1}_{01})+\frac{1}{4}t^{-1/2}e^{-\lambda/2}(\lambda_{t}+3t^{-1}).
\]
Compute, using the symmetries of the connection coefficients, 
\[
-\G_{\a0}^{\b}\G_{0\b}^{\a}=-\G_{i0}^{0}\G_{00}^{i}
-\G_{00}^{i}\G_{0i}^{0}-\G_{i0}^{j}\G_{0j}^{i}=
-\G_{00}^{1}\G_{01}^{0}=-(\g^{0}_{01})^{2}.
\]
Finally, 
\begin{equation*}
\begin{split}
-\g^{\b}_{\a 0}\G_{\b 0}^{\a} = &-\g^{\b}_{i 0}\G_{\b 0}^{i}
=\g^{0}_{0i}\G_{0 0}^{i}+\g^{j}_{0i}\G_{j 0}^{i}
=(\g^{0}_{01})^{2}+\g^{1}_{01}\G_{1 0}^{1}+\g^{A}_{01}\G_{A0}^{1}
+\g^{B}_{0A}\G_{B0}^{A}\\
 = & (\g^{0}_{01})^{2}-(\g^{1}_{01})^{2}+\g^{A}_{01}\G_{A0}^{1}
+\g^{B}_{0A}\G_{B0}^{A}\\
 = & (\g^{0}_{01})^{2}-(\g^{1}_{01})^{2}-\frac{1}{2}(\g^{2}_{01})^{2}
-\frac{1}{2}(\g^{3}_{01})^{2}-(\g^{2}_{02})^{2}-(\g^{3}_{03})^{2}
-\frac{1}{2}(\g^{2}_{03})^{2}-\frac{1}{2}\g^{2}_{03}\g^{3}_{02}.
\end{split}
\end{equation*}
Since $\g^{3}_{02}=0$, we obtain 
\begin{equation*}
\begin{split}
\Ric(e_{0},e_{0}) = & \frac{1}{4}t^{-1/2}e^{-\lambda/2}\a^{1/2}
\d_{\theta}(t\a^{1/2}\lambda_{\theta})
+\frac{1}{4}t^{-1/2}e^{-\lambda/2}(\lambda_{t}+3t^{-1})
+e_{0}(\g^{1}_{01})\\
 & -(\g^{1}_{01})^{2}-\frac{1}{2}(\g^{2}_{01})^{2}
-\frac{1}{2}(\g^{3}_{01})^{2}-(\g^{2}_{02})^{2}-(\g^{3}_{03})^{2}
-\frac{1}{2}(\g^{2}_{03})^{2}.
\end{split}
\end{equation*}
In order to simplify this expression, it is useful combine 
(\ref{eq:gaaz}) and (\ref{eq:ezla}) in order to conclude
\[
e_{0}(\g^{1}_{01})-(\g^{1}_{01})^{2}
=-\frac{1}{4}t^{-1/2}e^{-\lambda/2}\a^{1/2}\d_{t}[t\a^{-1/2}(\lambda_{t}
-2\a^{-1}\a_{t}-t^{-1})]
+\frac{1}{4}t^{-1/2}e^{-\lambda/2}(\lambda_{t}
-2\a^{-1}\a_{t}-t^{-1}).
\]
Let us now compute
\begin{equation*}
\begin{split}
 & -\frac{1}{2}(\g^{2}_{01})^{2}
-\frac{1}{2}(\g^{3}_{01})^{2}-(\g^{2}_{02})^{2}-(\g^{3}_{03})^{2}
-\frac{1}{2}(\g^{2}_{03})^{2}\\
 = & -\frac{1}{2}t^{1/2}e^{-\lambda/2}
\left(P_{t}^{2}+e^{2P}Q_{t}^{2}+\frac{e^{\lambda/2-P}J^{2}}{t^{7/2}}
+\frac{e^{\lambda/2+P}(K-QJ)^{2}}{t^{7/2}}\right)
-\frac{1}{2}t^{-3/2}e^{-\lambda/2}.
\end{split}
\end{equation*}
Thus
\begin{equation*}
\begin{split}
\Ric(e_{0},e_{0}) = & \frac{1}{4}t^{-1/2}e^{-\lambda/2}\a^{1/2}
\left(\d_{\theta}(t\a^{1/2}\lambda_{\theta})
-\d_{t}[t\a^{-1/2}(\lambda_{t}
-2\a^{-1}\a_{t}-t^{-1})]\right)\\
 & -\frac{1}{2}t^{1/2}e^{-\lambda/2}
\left(P_{t}^{2}+e^{2P}Q_{t}^{2}+\frac{e^{\lambda/2-P}J^{2}}{t^{7/2}}
+\frac{e^{\lambda/2+P}(K-QJ)^{2}}{t^{7/2}}\right)\\
 & +\frac{1}{2}t^{-1/2}e^{-\lambda/2}\left(\lambda_{t}-\frac{\a_{t}}{\a}\right).
\end{split}
\end{equation*}
Let us compute, using (\ref{eq:gaao}), (\ref{eq:gaaA}) and the antisymmetry of 
$\G^{i}_{\a j}$ in $i$ and $j$, that 
\begin{equation*}
\begin{split}
\Ric(e_{0},e_{1}) = &  e_{\a}(\G^{\a}_{01})
+\G^{\b}_{01}\G^{\a}_{\a\b}-e_{0}(\G^{\a}_{\a 1})
-\G^{\b}_{\a 1}\G_{0\b}^{\a}
-\g^{\b}_{\a 0}\G^{\a}_{\b 1}\\
= &  e_{0}(\G^{0}_{01})+\G^{0}_{01}\G^{\a}_{\a 0}
-e_{0}(\g^{0}_{01})
-\G^{\b}_{\a 1}\G_{0\b}^{\a}
-\g^{\b}_{\a 0}\G^{\a}_{\b 1}\\
 = & \G^{0}_{01}\G^{\a}_{\a 0}-\G^{\b}_{\a 1}\G_{0\b}^{\a}
-\g^{\b}_{\a 0}\G^{\a}_{\b 1}.
\end{split}
\end{equation*}
Compute, using (\ref{eq:gaaz}), 
\[
\G^{0}_{01}\G^{\a}_{\a 0}=\g^{0}_{01}(-\g^{1}_{01}+t^{-3/4}e^{-\lambda/4})
=-\g^{0}_{01}\g^{1}_{01}+\g^{0}_{01}t^{-3/4}e^{-\lambda/4}.
\]
Moreover, using the symmetries of the connection coefficients, 
\[
-\G^{\b}_{\a 1}\G_{0\b}^{\a}=-\G^{0}_{11}\G^{1}_{00}-\G^{i}_{j1}\G^{j}_{0i}
=\g^{0}_{01}\g^{1}_{01}.
\]
Thus
\[
\G^{0}_{01}\G^{\a}_{\a 0}-\G^{\b}_{\a 1}\G_{0\b}^{\a}
=\g^{0}_{01}t^{-3/4}e^{-\lambda/4}
=\frac{1}{4}t^{-1/2}e^{-\lambda/2}\a^{1/2}\lambda_{\theta}.
\]
Finally, 
\begin{equation*}
\begin{split}
-\g^{\b}_{\a 0}\G^{\a}_{\b 1} = & \g^{\b}_{0i}\G^{i}_{\b 1}=
\g^{0}_{01}\G^{1}_{01}+\g^{j}_{0i}\G^{i}_{j 1}=
\g^{2}_{02}\G^{2}_{2 1}+\g^{2}_{03}\G^{3}_{2 1}+\g^{3}_{02}\G^{2}_{3 1}+
\g^{3}_{03}\G^{3}_{31}\\
 = & -\g^{2}_{02}\g^{2}_{12}-\frac{1}{2}\g^{2}_{03}\g^{2}_{13}
-\g^{3}_{03}\g^{3}_{13}
=-\frac{1}{2}t^{1/2}e^{-\lambda/2}\a^{1/2}
(P_{t}P_{\theta}+e^{2P}Q_{t}Q_{\theta}).
\end{split}
\end{equation*}
Thus
\begin{equation*}
\begin{split}
\Ric(e_{0},e_{1}) = &
\frac{1}{4}t^{-1/2}e^{-\lambda/2}\a^{1/2}[\lambda_{\theta}
-2t(P_{t}P_{\theta}+e^{2P}Q_{t}Q_{\theta})].
\end{split}
\end{equation*}
Let us compute, using (\ref{eq:gaaA}),
\begin{equation*}
\begin{split}
\Ric(e_{0},e_{A}) = & e_{\a}(\G^{\a}_{0A})
+\G^{\b}_{0A}\G_{\a \b}^{\a}
-e_{0}(\G_{\a A}^{\a})-\G_{\a A}^{\b}\G_{0\b}^{\a}
-\g^{\b}_{\a 0}\G_{\b A}^{\a}\\
 = & -\frac{1}{2}e_{1}(\g^{A}_{01})
+\G^{1}_{0A}\G_{\a 1}^{\a}-\G_{1 A}^{0}\G_{00}^{1}-\G_{0 A}^{1}\G_{01}^{0}
-\G_{j A}^{i}\G_{0i}^{j}
+\g^{0}_{01}\G_{0 A}^{1}+\g^{i}_{0j}\G_{i A}^{j}
\\
 = & -\frac{1}{2}e_{1}(\g^{A}_{01})
+(\G^{A}_{B1}+\G^{B}_{1A})\G^{B}_{01}+\g^{B}_{01}\G_{B A}^{1}.
\end{split}
\end{equation*}
Let us begin by considering the case $A=2$. Compute
\begin{eqnarray*}
-\frac{1}{2}e_{1}(\g^{2}_{01}) & = & 
-\frac{1}{2}t^{-5/4}e^{-P/2-\lambda/4}\a^{1/2}J_{\theta}
-\frac{1}{2}\g^{2}_{12}\g^{2}_{01},\\
(\G^{2}_{B1}+\G^{B}_{12})\G^{B}_{01} & = & -\frac{1}{2}\g^{2}_{12}\g^{2}_{01}
-\frac{1}{2}\g^{2}_{13}\g^{3}_{01},\\
\g^{B}_{01}\G_{B 2}^{1} & = & \g^{2}_{01}\g^{2}_{12}+\frac{1}{2}\g^{3}_{01}
\g^{2}_{13}.
\end{eqnarray*}
Thus
\[
\Ric(e_{0},e_{2})=-\frac{1}{2}t^{-5/4}e^{-P/2-\lambda/4}\a^{1/2}J_{\theta}.
\]
Compute
\begin{eqnarray*}
-\frac{1}{2}e_{1}(\g^{3}_{01}) & = & -\frac{1}{2}t^{-5/4}e^{P/2-\lambda/4}
\a^{1/2}(K_{\theta}-QJ_{\theta})+\frac{1}{2}\g^{2}_{12}\g^{3}_{01}
-\frac{1}{2}\g^{2}_{13}\g^{2}_{01},\\
(\G^{3}_{B1}+\G^{B}_{13})\G^{B}_{01} & = & -\frac{1}{2}\g^{3}_{13}\g^{3}_{01},\\
\g^{B}_{01}\G_{B 3}^{1} & = & \frac{1}{2}\g^{2}_{01}\g^{2}_{13}+\g^{3}_{01}
\g^{3}_{13}.
\end{eqnarray*}
Keeping in mind that $\g^{2}_{12}=-\g^{3}_{13}$, we conclude that 
\[
\Ric(e_{0},e_{3})=-\frac{1}{2}t^{-5/4}e^{P/2-\lambda/4}
\a^{1/2}(K_{\theta}-QJ_{\theta}).
\]
Compute, using the symmetries of the connection coefficients,
\begin{equation*}
\begin{split}
\Ric(e_{1},e_{1}) = & e_{\a}(\G_{11}^{\a})+\G_{11}^{\de}\G_{\a\de}^{\a}
-e_{1}(\G_{\a 1}^{\a})-\G_{\a 1}^{\de}\G_{1\de}^{\a}
-\g^{\de}_{\a 1}\G_{\de 1}^{\a}\\
 = & e_{0}(\G_{11}^{0})+\G_{11}^{0}\G_{\a 0}^{\a}
-e_{1}(\g^{0}_{01})-\G_{0 1}^{A}\G_{1A}^{0}
-\G_{i 1}^{0}\G_{1 0}^{i}-\G_{i 1}^{j}\G_{1j}^{i}
-\g^{0}_{0 1}\G_{0 1}^{0}-\g^{1}_{0 1}\G_{1 1}^{0}\\
 & -\g^{A}_{0 1}\G_{A 1}^{0}-\g^{B}_{A 1}\G_{B 1}^{A}\\
 = & e_{0}(\G_{11}^{0})+\G_{11}^{0}\G_{\a 0}^{\a}
-e_{1}(\g^{0}_{01})+\frac{1}{4}\g^{A}_{01}\g^{A}_{01}
-\frac{1}{4}\g^{A}_{01}\g^{A}_{01}
+\g^{1}_{01}\G^{0}_{11}
-(\g^{0}_{0 1})^{2}\\
 & -\g^{1}_{0 1}\G_{1 1}^{0}-\g^{A}_{0 1}\G_{A 1}^{0}-\g^{B}_{A 1}\G_{B 1}^{A}\\
 = & e_{0}(\G_{11}^{0})+\G_{11}^{0}\G_{\a 0}^{\a}
-e_{1}(\g^{0}_{01})-(\g^{0}_{0 1})^{2}
-\g^{A}_{0 1}\G_{A 1}^{0}-\g^{B}_{A 1}\G_{B 1}^{A}.
\end{split}
\end{equation*}
Due to (\ref{eq:ezla}) and (\ref{eq:eola}), we know that the sum of 
the first four terms is given by 
\[
\frac{1}{4}t^{-1/2}e^{-\lambda/2}\a^{1/2}
\left(\d_{t}[t\a^{-1/2}(\lambda_{t}
-2\a^{-1}\a_{t}-t^{-1})]-\d_{\theta}(t\a^{1/2}\lambda_{\theta})\right).
\]
It can also be computed that 
\begin{eqnarray*}
-\g^{A}_{0 1}\G_{A 1}^{0} & = & \frac{1}{2}t^{-3}[e^{-P}J^{2}+e^{P}(K-QJ)^{2}],\\
-\g^{B}_{A 1}\G_{B 1}^{A} & = & -\frac{1}{2}t^{1/2}e^{-\lambda/2}\a
(P_{\theta}^{2}+e^{2P}Q_{\theta}^{2}).
\end{eqnarray*}
Thus
\begin{equation*}
\begin{split}
\Ric(e_{1},e_{1}) = & \frac{1}{4}t^{-1/2}e^{-\lambda/2}\a^{1/2}
\left(\d_{t}[t\a^{-1/2}(\lambda_{t}
-2\a^{-1}\a_{t}-t^{-1})]-\d_{\theta}(t\a^{1/2}\lambda_{\theta})\right)\\
 & +\frac{1}{2}t^{1/2}e^{-\lambda/2}
\left(\frac{e^{\lambda/2-P}J^{2}}{t^{7/2}}+\frac{e^{\lambda/2+P}(K-QJ)^{2}}{t^{7/2}}
\right)-\frac{1}{2}t^{1/2}e^{-\lambda/2}\a(P_{\theta}^{2}+e^{2P}Q_{\theta}^{2}).
\end{split}
\end{equation*}
Let us turn to 
\begin{equation*}
\begin{split}
\Ric(e_{1},e_{A}) = & e_{\a}(\G_{1A}^{\a})+\G_{1A}^{\de}\G_{\a\de}^{\a}
-e_{1}(\G^{\a}_{\a A})-\G_{\a A}^{\de}\G_{1\de}^{\a}
-\g^{\de}_{\a 1}\G_{\de A}^{\a}\\
 = & e_{0}(\G_{1A}^{0})+\G_{1A}^{0}\G_{\a 0}^{\a}+\G_{1A}^{1}\G_{\a 1}^{\a}
-\G_{i A}^{0}\G_{1 0}^{i}-\G_{0 A}^{i}\G_{1i}^{0}
-\G_{j A}^{i}\G_{1i}^{j}-\g^{0}_{i 1}\G_{0 A}^{i}\\
 & -\g^{i}_{0 1}\G_{i A}^{0}-\g^{i}_{j 1}\G_{i A}^{j}\\
 = & -\frac{1}{2}e_{0}(\g^{A}_{01})+\G^{1}_{A0}\G^{B}_{B0}-\G^{B}_{A0}\G^{1}_{B0}
-\G^{1}_{0A}\G^{1}_{10}-\G^{B}_{0A}\G^{1}_{B0}
-\g^{1}_{0 1}\G_{1 A}^{0}-\g^{B}_{0 1}\G_{B A}^{0}.
\end{split}
\end{equation*}
Compute that 
\begin{eqnarray*}
-\frac{1}{2}e_{0}(\g^{2}_{01}) & = & -\frac{1}{2}t^{-5/4}e^{-\lambda/4-P/2}J_{t}
-\frac{1}{2}\g^{3}_{03}\g^{2}_{01}-\g^{2}_{02}\g^{2}_{01},\\
-\G^{1}_{02}\G^{1}_{10}-\g^{1}_{0 1}\G_{1 2}^{0} & = & 0,\\
-\G^{B}_{20}\G^{1}_{B0}-\G^{B}_{02}\G^{1}_{B0} & = & -\frac{1}{2}\g^{2}_{02}\g^{2}_{01}
-\frac{1}{2}\g^{2}_{03}\g^{3}_{01},\\
\G^{1}_{20}\G^{B}_{B0}-\g^{B}_{01}\G^{0}_{B2} & = & \frac{3}{2}\g^{2}_{01}\g^{2}_{02}
+\frac{1}{2}\g^{2}_{01}\g^{3}_{03}+\frac{1}{2}\g^{3}_{01}\g^{2}_{03}.
\end{eqnarray*}
Adding up, we obtain 
\[
\Ric(e_{1},e_{2})=-\frac{1}{2}t^{-5/4}e^{-\lambda/4-P/2}J_{t}.
\]
Compute that 
\begin{eqnarray*}
-\frac{1}{2}e_{0}(\g^{3}_{01}) & = & -\frac{1}{2}t^{-5/4}e^{P/2-\lambda/4}
(K_{t}-QJ_{t})-\frac{1}{2}\g^{2}_{03}\g^{2}_{01}-\frac{1}{2}\g^{2}_{02}\g^{3}_{01}
-\g^{3}_{03}\g^{3}_{01},\\
-\G^{1}_{03}\G^{1}_{10}-\g^{1}_{0 1}\G_{1 3}^{0} & = & 0,\\
-\G^{B}_{30}\G^{1}_{B0}-\G^{B}_{03}\G^{1}_{B0} & = & -\frac{1}{2}\g^{3}_{03}\g^{3}_{01},\\
\G^{1}_{30}\G^{B}_{B0}-\g^{B}_{01}\G^{0}_{B3} & = & \frac{3}{2}\g^{3}_{01}\g^{3}_{03}
+\frac{1}{2}\g^{3}_{01}\g^{2}_{02}+\frac{1}{2}\g^{2}_{01}\g^{2}_{03}.
\end{eqnarray*}
Adding up, we obtain 
\[
\Ric(e_{1},e_{3})=-\frac{1}{2}t^{-5/4}e^{P/2-\lambda/4}
(K_{t}-QJ_{t}).
\]
Compute 
\begin{equation*}
\begin{split}
\Ric(e_{A},e_{B}) = & e_{\a}(\G_{AB}^{\a})+\G_{AB}^{\de}\G_{\a\de}^{\a}
-\G_{\a B}^{\de}\G_{A\de}^{\a}
-\g^{\de}_{\a A}\G_{\de B}^{\a}\\
= & e_{0}(\G_{AB}^{0})+e_{1}(\G_{AB}^{1})+\G_{AB}^{0}\G_{\a 0}^{\a}
+\G_{AB}^{1}\G_{\a 1}^{\a}
-\G_{\a B}^{\de}\G_{A\de}^{\a}
-\g^{\de}_{\a A}\G_{\de B}^{\a}.
\end{split}
\end{equation*}
Note that if $A=B=2$, the first four terms can be written
\[
\frac{1}{2}t^{-1/2}e^{-\lambda/2}\a^{1/2}\left(\d_{t}[t\a^{-1/2}
(t^{-1}+P_{t})]-\d_{\theta}(t\a^{1/2}P_{\theta})\right);
\]
cf. (\ref{eq:ezpp}) and (\ref{eq:eott}). Let us therefore turn to 
\begin{equation*}
\begin{split}
-\G_{\a 2}^{\de}\G_{2\de}^{\a}
-\g^{\de}_{\a 2}\G_{\de 2}^{\a} = & -\G_{i 2}^{0}\G_{20}^{i}
-\G_{0 2}^{i}\G_{2i}^{0}-\G_{j 2}^{i}\G_{2i}^{j}
-\g^{2}_{0 2}\G_{2 2}^{0}-\g^{2}_{1 2}\G_{22}^{1}\\
= & -\G_{20}^{i}(\G_{20}^{i}+\G_{0 2}^{i})-\G_{A2}^{1}\G_{21}^{A}
-\G_{1 2}^{A}\G_{2A}^{1}
-\g^{2}_{0 2}\G_{2 2}^{0}-\g^{2}_{1 2}\G_{22}^{1}.
\end{split}
\end{equation*}
However, it can be computed that 
\begin{eqnarray*}
-\G_{20}^{i}(\G_{20}^{i}+\G_{0 2}^{i}) & = & -\frac{1}{2}t^{-3}e^{-P}J^{2}
-\frac{1}{2}t^{1/2}e^{-\lambda/2}e^{2P}Q_{t}^{2}-(\g^{2}_{02})^{2},\\
-\G_{A2}^{1}\G_{21}^{A}
-\G_{1 2}^{A}\G_{2A}^{1} & = & \frac{1}{2}t^{1/2}e^{-\lambda/2}\a e^{2P}Q_{\theta}^{2}
+(\g^{2}_{12})^{2},\\
-\g^{2}_{0 2}\G_{2 2}^{0}-\g^{2}_{1 2}\G_{22}^{1} & = & 
(\g^{2}_{02})^{2}-(\g^{2}_{12})^{2}.
\end{eqnarray*}
Adding up,
\begin{equation*}
\begin{split}
\Ric(e_{2},e_{2}) = &
\frac{1}{2}t^{-1/2}e^{-\lambda/2}\a^{1/2}\left(\d_{t}[t\a^{-1/2}
(t^{-1}+P_{t})]-\d_{\theta}(t\a^{1/2}P_{\theta})\right)
-\frac{1}{2}t^{1/2}e^{-\lambda/2}\frac{e^{\lambda/2-P}J^{2}}{t^{7/2}}\\
 & -\frac{1}{2}t^{1/2}e^{-\lambda/2}
e^{2P}(Q_{t}^{2}-\a Q_{\theta}^{2}).
\end{split}
\end{equation*}
Next, consider
\begin{equation*}
\begin{split}
\Ric(e_{2},e_{3}) = & e_{0}(\G_{23}^{0})+e_{1}(\G_{23}^{1})+\G_{23}^{0}\G_{\a 0}^{\a}
+\G_{23}^{1}\G_{\a 1}^{\a}
-\G_{\a 3}^{\de}\G_{2\de}^{\a}
-\g^{\de}_{\a 2}\G_{\de 3}^{\a}.
\end{split}
\end{equation*}
Due to (\ref{eq:eogotth}) and (\ref{eq:ezgztth}), the sum of the first four
terms can be written
\[
\frac{1}{2}t^{-1/2}e^{-\lambda/2-P}\a^{1/2}
\left[\d_{t}(t\a^{-1/2}e^{2P}Q_{t})-\d_{\theta}(t\a^{1/2}
e^{2P}Q_{\theta})\right]-\frac{1}{2}(\g^{3}_{13}-\g^{2}_{12})\g^{2}_{13}
+\frac{1}{2}(\g^{3}_{03}-\g^{2}_{02})\g^{2}_{03}.
\]
Compute
\begin{equation*}
\begin{split}
-\G_{\a 3}^{\de}\G_{2\de}^{\a}
-\g^{\de}_{\a 2}\G_{\de 3}^{\a} = & 
-\G_{i3}^{0}\G_{20}^{i}-\G_{0 3}^{i}\G_{2i}^{0}
-\G_{j3}^{i}\G_{2i}^{j}-\g^{2}_{02}\G_{2 3}^{0}
-\g^{2}_{12}\G_{2 3}^{1}\\
 = & -(\G_{30}^{i}+\G_{0 3}^{i})\G_{20}^{i}
-\G_{A3}^{1}\G_{21}^{A}-\G_{13}^{A}\G_{2A}^{1}-\g^{2}_{02}\G_{2 3}^{0}
-\g^{2}_{12}\G_{2 3}^{1}.
\end{split}
\end{equation*}
However, 
\begin{eqnarray*}
-(\G_{30}^{i}+\G_{0 3}^{i})\G_{20}^{i} & = & -\frac{1}{2}\g^{3}_{01}\g^{2}_{01}
-\frac{1}{2}\g^{3}_{03}\g^{2}_{03},\\
-\G_{A3}^{1}\G_{21}^{A}-\G_{13}^{A}\G_{2A}^{1} & = & \frac{1}{2}\g^{3}_{13}
\g^{2}_{13},\\
-\g^{2}_{02}\G_{2 3}^{0}-\g^{2}_{12}\G_{2 3}^{1} & = & \frac{1}{2}\g^{2}_{02}
\g^{2}_{03}-\frac{1}{2}\g^{2}_{12}\g^{2}_{13}.
\end{eqnarray*}
Adding up, we obtain 
\[
\Ric(e_{2},e_{3})=\frac{1}{2}t^{-1/2}e^{-\lambda/2-P}\a^{1/2}
\left[\d_{t}(t\a^{-1/2}e^{2P}Q_{t})-\d_{\theta}(t\a^{1/2}
e^{2P}Q_{\theta})\right]-\frac{1}{2}t^{-3}J(K-QJ).
\]
Finally, let us consider
\[
\Ric(e_{3},e_{3}) = e_{0}(\G_{33}^{0})+e_{1}(\G_{33}^{1})+\G_{33}^{0}\G_{\a 0}^{\a}
+\G_{33}^{1}\G_{\a 1}^{\a}
-\G_{\a 3}^{\de}\G_{3\de}^{\a}
-\g^{\de}_{\a 3}\G_{\de 3}^{\a}.
\]
Due to (\ref{eq:ezpm}) and (\ref{eq:eothth}), the sum of the first four 
terms is 
\[
\frac{1}{2}t^{-1/2}e^{-\lambda/2}\a^{1/2}\left(
\d_{t}[t\a^{-1/2}(t^{-1}-P_{t})]
+\d_{\theta}(t\a^{1/2}P_{\theta})\right).
\]
Let us therefore compute
\begin{equation*}
\begin{split}
-\G_{\a 3}^{\de}\G_{3\de}^{\a}
-\g^{\de}_{\a 3}\G_{\de 3}^{\a} = &  -\G_{i 3}^{0}\G_{30}^{i}
-\G_{0 3}^{i}\G_{3i}^{0}-\G_{j 3}^{i}\G_{3i}^{j}
-\g^{A}_{03}\G_{A 3}^{0}-\g^{A}_{13}\G_{A3}^{1}\\
= &  -(\G_{30}^{i}+\G_{0 3}^{i})\G_{30}^{i}-\G_{A3}^{1}\G_{31}^{A}
-\G_{13}^{A}\G_{3A}^{1}
-\g^{A}_{03}\G_{A 3}^{0}-\g^{A}_{13}\G_{A3}^{1}.
\end{split}
\end{equation*}
On the other hand
\begin{eqnarray*}
-(\G_{30}^{i}+\G_{0 3}^{i})\G_{30}^{i} & = & -\frac{1}{2}(\g^{3}_{01})^{2}
-(\g^{3}_{03})^{2},\\
-\G_{A3}^{1}\G_{31}^{A}-\G_{13}^{A}\G_{3A}^{1} & = & (\g^{3}_{13})^{2},\\
-\g^{A}_{03}\G_{A3}^{0}-\g^{A}_{13}\G_{A3}^{1} & = & \frac{1}{2}(\g^{2}_{03})^{2}
+(\g^{3}_{03})^{2}-\frac{1}{2}(\g^{2}_{13})^{2}-(\g^{3}_{13})^{2}.
\end{eqnarray*}
Adding up, we obtain 
\[
-\G_{\a 3}^{\de}\G_{3\de}^{\a}
-\g^{\de}_{\a 3}\G_{\de 3}^{\a}=
-\frac{1}{2}t^{-3}e^{P}(K-QJ)^{2}+\frac{1}{2}t^{1/2}e^{-\lambda/2}
e^{2P}(Q_{t}^{2}-\a Q_{\theta}^{2}).
\]
Thus
\begin{equation*}
\begin{split}
\Ric(e_{3},e_{3}) = & \frac{1}{2}t^{-1/2}e^{-\lambda/2}\a^{1/2}\left(
\d_{t}[t\a^{-1/2}(t^{-1}-P_{t})]
+\d_{\theta}(t\a^{1/2}P_{\theta})\right)\\
 & -\frac{1}{2}t^{1/2}e^{-\lambda/2}\frac{e^{\lambda/2+P}(K-QJ)^{2}}{t^{7/2}}
+\frac{1}{2}t^{1/2}e^{-\lambda/2}e^{2P}(Q_{t}^{2}-\a Q_{\theta}^{2}).
\end{split}
\end{equation*}

\subsection{The Einstein tensor}

Adding up the above, we conclude that the scalar curvature $S$ is given by 
\begin{equation*}
\begin{split}
S = & \frac{1}{2}t^{-1/2}e^{-\lambda/2}\a^{1/2}
\left(\d_{t}[t\a^{-1/2}(\lambda_{t}
-2\a^{-1}\a_{t}-t^{-1})]
-\d_{\theta}(t\a^{1/2}\lambda_{\theta})\right)\\
 & +\frac{1}{2}t^{1/2}e^{-\lambda/2}
\left[P_{t}^{2}+e^{2P}Q_{t}^{2}-\a (P_{\theta}^{2}+e^{2P}Q_{\theta}^{2})
\right]\\
& +\frac{1}{2}t^{1/2}e^{-\lambda/2}
\left(\frac{e^{\lambda/2-P}J^{2}}{t^{7/2}}
+\frac{e^{\lambda/2+P}(K-QJ)^{2}}{t^{7/2}}\right)
-\frac{1}{2}t^{-1/2}e^{-\lambda/2}\lambda_{t}.
\end{split}
\end{equation*}
As a consequence, if $\Ein_{\a\b}=\Ein(e_{\a},e_{\b})$, then 
\begin{eqnarray*}
\Ein_{00} & = & -\frac{1}{4}t^{1/2}e^{-\lambda/2}\left[
P_{t}^{2}+\a P_{\theta}^{2}+e^{2P}(Q_{t}^{2}+\a Q_{\theta}^{2})
+\frac{e^{\lambda/2-P}J^{2}}{t^{7/2}}
+\frac{e^{\lambda/2+P}(K-QJ)^{2}}{t^{7/2}}\right]\\
 & & +\frac{1}{4}t^{-1/2}e^{-\lambda/2}
\left(\lambda_{t}-2\frac{\a_{t}}{\a}\right),\\
\Ein_{11} & = & -\frac{1}{4}t^{1/2}e^{-\lambda/2}\left[
P_{t}^{2}+\a P_{\theta}^{2}+e^{2P}(Q_{t}^{2}+\a Q_{\theta}^{2})
-\frac{e^{\lambda/2-P}J^{2}}{t^{7/2}}
-\frac{e^{\lambda/2+P}(K-QJ)^{2}}{t^{7/2}}\right]\\
& & +\frac{1}{4}t^{-1/2}e^{-\lambda/2}\lambda_{t},\\
\Ein_{22} & = & \frac{1}{2}t^{-1/2}e^{-\lambda/2}\a^{1/2}
\left(\d_{t}\left[t\a^{-1/2}\left(P_{t}-\frac{1}{2}\lambda_{t}
+\frac{\a_{t}}{\a}+\frac{3}{2t}\right)\right]
-\d_{\theta}\left[t\a^{1/2}\left(P_{\theta}-\frac{1}{2}\lambda_{\theta}\right)
\right]\right)\\
& & -\frac{1}{4}t^{1/2}e^{-\lambda/2}
\left[P_{t}^{2}+3e^{2P}Q_{t}^{2}-\a (P_{\theta}^{2}+3e^{2P}Q_{\theta}^{2})
\right]\\
& & -\frac{1}{4}t^{1/2}e^{-\lambda/2}
\left(3\frac{e^{\lambda/2-P}J^{2}}{t^{7/2}}
+\frac{e^{\lambda/2+P}(K-QJ)^{2}}{t^{7/2}}\right)
+\frac{1}{4}t^{-1/2}e^{-\lambda/2}\lambda_{t}.
\end{eqnarray*}
The last diagonal component is given by 
\begin{equation*}
\begin{split}
\Ein_{33} = & -\frac{1}{2}t^{-1/2}e^{-\lambda/2}\a^{1/2}
\left(\d_{t}\left[t\a^{-1/2}\left(P_{t}+\frac{1}{2}\lambda_{t}
-\frac{\a_{t}}{\a}-\frac{3}{2t}\right)\right]
-\d_{\theta}\left[t\a^{1/2}\left(P_{\theta}+\frac{1}{2}\lambda_{\theta}\right)
\right]\right)\\
 & -\frac{1}{4}t^{1/2}e^{-\lambda/2}
\left[P_{t}^{2}-e^{2P}Q_{t}^{2}-\a (P_{\theta}^{2}-e^{2P}Q_{\theta}^{2})
\right]\\
 & -\frac{1}{4}t^{1/2}e^{-\lambda/2}
\left(\frac{e^{\lambda/2-P}J^{2}}{t^{7/2}}
+3\frac{e^{\lambda/2+P}(K-QJ)^{2}}{t^{7/2}}\right)
+\frac{1}{4}t^{-1/2}e^{-\lambda/2}\lambda_{t}.
\end{split}
\end{equation*}
It is of interest to note that 
\begin{equation*}
\begin{split}
\Ein_{22}-\Ein_{33} = & t^{-1/2}e^{-\lambda/2}\a^{1/2}
\left(\d_{t}(t\a^{-1/2}P_{t})
-\d_{\theta}(t\a^{1/2}P_{\theta})
\right)-t^{1/2}e^{-\lambda/2}e^{2P}(Q_{t}^{2}-\a Q_{\theta}^{2})\\
 & -\frac{1}{2}t^{1/2}e^{-\lambda/2}
\left(\frac{e^{\lambda/2-P}J^{2}}{t^{7/2}}
-\frac{e^{\lambda/2+P}(K-QJ)^{2}}{t^{7/2}}\right).
\end{split}
\end{equation*}
We also have 
\begin{equation*}
\begin{split}
\Ein_{22}+\Ein_{33} = & -\frac{1}{2}t^{-1/2}e^{-\lambda/2}\a^{1/2}
\left(\d_{t}\left[t\a^{-1/2}\left(\lambda_{t}
-2\frac{\a_{t}}{\a}-\frac{3}{t}\right)\right]
-\d_{\theta}\left(t\a^{1/2}\lambda_{\theta}\right)\right)\\
 & -\frac{1}{2}t^{1/2}e^{-\lambda/2}
\left[P_{t}^{2}+e^{2P}Q_{t}^{2}-\a (P_{\theta}^{2}+e^{2P}Q_{\theta}^{2})
\right]\\
 & -t^{1/2}e^{-\lambda/2}
\left(\frac{e^{\lambda/2-P}J^{2}}{t^{7/2}}
+\frac{e^{\lambda/2+P}(K-QJ)^{2}}{t^{7/2}}\right)
+\frac{1}{2}t^{-1/2}e^{-\lambda/2}\lambda_{t}.
\end{split}
\end{equation*}
The remaining components of the Einstein tensor equal the corresponding
components of the Ricci tensor, and have consequently already been 
computed. Using the above calculations, the expressions 
(\ref{eq:Ezz})--(\ref{eq:Etth}) for Einstein's equations, $\Ein+\Lambda g=T$, 
follow.

\subsection{The Vlasov equation}\label{ssection:thevlasoveq}

The distribution function $f$ characterising the Vlasov matter is defined on 
the mass shell. The mass shell, in its turn, is given by the future directed
unit timelike vectors. Since a tangent vector in this set can be written 
$v^{\a}e_{\a}$, where 
\[
v^{0}=[1+(v^{1})^{2}+(v^{2})^{2}+(v^{3})^{2}]^{1/2},
\]
we can think of $f$ as depending on $v^{i}$, $i=1,2,3$, and the base point. 
However, due to the symmetry requirements, the distribution function only
depends on the $t\theta$-coordinates of the base point. As a consequence, 
the distribution function can be considered to be a function of $(t,\theta,v)$,
where $v=(v^{1}, v^{2}, v^{3})$. In order to derive an equation for $f$, 
recall that the Vlasov equation is equivalent to $f$ being constant along 
future directed unit timelike geodesics. Consider, therefore, a future 
directed unit timelike geodesic 
\[
\g(s)=[t(s),\theta(s),x(s),y(s)]
\]
in a $\tn{2}$-symmetric spacetime. Define the functions $v^{\a}(s)$
by the equality
\[
\dot{\g}(s)=v^{\a}(s)\left. e_{\a}\right|_{\g(s)}.
\]
Note that 
\begin{eqnarray}
\frac{dt}{ds}(s) & = & t^{1/4}(s)(e^{-\lambda/4})\circ\g(s)\, v^{0}(s),
\label{eq:dtds}\\
\frac{d\theta}{ds}(s) & = & t^{1/4}(s)(e^{-\lambda/4}
\a^{1/2})\circ\g(s)\, v^{1}(s).\label{eq:dthds}
\end{eqnarray}
Let 
\[
\nu(s)=[t(s),\theta(s),v(s)],\ \ \
h=f\circ\nu,
\]
where $v(s)=[v^{1}(s), v^{2}(s), v^{3}(s)]$. The requirement that $f$ be 
constant along geodesics is equivalent to the requirement that $dh/ds=0$
regardless of the choice of future directed unit timelike geodesic $\g$. On 
the other hand, 
\[
\frac{dh}{ds}=\frac{\d f}{\d t}\circ\nu\, \frac{dt}{ds}
+\frac{\d f}{\d \theta}\circ\nu\, \frac{d\theta}{ds}
+\sum_{i=1}^{3}\frac{\d f}{\d v^{i}}\circ\nu\, \frac{dv^{i}}{ds}.
\]
Keeping (\ref{eq:dtds}) and (\ref{eq:dthds}) in mind, the requirement 
that $dh/ds=0$ is equivalent to the requirement that 
\[
\frac{\d f}{\d t}\circ\nu
+\a^{1/2}\circ\g\frac{v^{1}}{v^{0}}
\frac{\d f}{\d \theta}\circ\nu
+\sum_{i=1}^{3}\frac{\dot{v}^{i}}{t^{1/4}e^{-\lambda\circ\g/4}v^{0}}
\frac{\d f}{\d v^{i}}\circ\nu=0.
\]
In order to derive an expression for $\dot{v}$, note that 
\[
0=\ddot{\g}=\frac{d}{ds}(v^{\a}e_{\a})=\dot{v}^{\a}e_{\a}+
v^{\b}\nabla_{\dot{\g}}e_{\b}=\dot{v}^{\a}e_{\a}+
v^{\b}v^{\mu}\nabla_{e_{\mu}}e_{\b}=(\dot{v}^{\a}
+v^{\b}v^{\mu}\G^{\a}_{\mu\b})e_{\a}.
\]
The geodesic equation can thus be written 
\[
\dot{v}^{\a}=-v^{\b}v^{\mu}\G^{\a}_{\b\mu}.
\]
Using this formula, it can be calculated that 
\[
\dot{v}^{1} =-\g^{0}_{01}v^{0}v^{0}+\g^{1}_{01}v^{0}v^{1}+\g^{2}_{01}v^{0}v^{2}
+\g^{3}_{01}v^{0}v^{3}-\g^{2}_{12}v^{2}v^{2}-\g^{3}_{13}v^{3}v^{3}-\g^{2}_{13}v^{2}v^{3}.
\]
Using the formulae for $\g^{\a}_{\lambda\mu}$, we conclude that 
\begin{equation*}
\begin{split}
\frac{\dot{v}^{1}}{t^{1/4}e^{-\lambda/4}v^{0}} = & 
-\frac{1}{4}\a^{1/2}\lambda_{\theta}\,v^{0}
-\frac{1}{4}\left(\lambda_{t}-2\frac{\alpha_t}{\alpha}-\frac{1}{t}\right)v^{1}
+\a^{1/2}e^{P}Q_{\theta}\frac{v^2v^3}{v^0}
\\
& -\frac{1}{2}\a^{1/2}P_{\theta}\frac{(v^{3})^{2}
-(v^{2})^{2}}{v^{0}}
+t^{-7/4}e^{\lambda/4}\big(e^{-P/2}Jv^2+e^{P/2}(K-Q J)v^3\big),
\end{split}
\end{equation*}
where we have omitted composition with $\g$ for the sake of brevity.
We also have 
\[
\dot{v}^{2}=\g^{2}_{02}v^{0}v^{2}+\g^{2}_{12}v^{1}v^{2},
\]
so that 
\[
\frac{\dot{v}^{2}}{t^{1/4}e^{-\lambda/4}v^{0}}=
-\frac{1}{2}\left(P_{t}+\frac{1}{t}\right)v^{2}
-\frac{1}{2}\a^{1/2}P_{\theta}\frac{v^{1}v^{2}}{v^{0}}.
\]
Finally, 
\[
\dot{v}^{3}=\g^{2}_{03}v^{0}v^{2}+\g^{3}_{03}v^{0}v^{3}+\g^{2}_{13}v^{1}v^{2}
+\g^{3}_{13}v^{1}v^{3},
\]
so that 
\[
\frac{\dot{v}^{3}}{t^{1/4}e^{-\lambda/4}v^{0}}=
-\frac{1}{2}\left(\frac{1}{t}-P_{t}\right)v^{3}
+\frac{1}{2}\a^{1/2}P_{\theta}\frac{v^{1}v^{3}}{v^{0}}-e^{P}v^2
\left(Q_t+\a^{1/2}Q_{\theta}\frac{v^1}{v^0}\right).
\]
Adding up the above computations, we conclude that the Vlasov equation is 
equivalent to the requirement that (\ref{vv}) holds.

\section{Notation}\label{section:notation}

\textbf{Metric variables, twist quantities, cosmological constant,
frame, manifold.} 
\begin{enumerate}
\item $\a>0$, $\lambda$, $P$, $Q$, $G$ and $H$. These are the functions
characterising the metric; cf. (\ref{eq:metric}). 
\item $\hl$ is defined in (\ref{eq:hldef}).
\item $J$ and $K$. These are the twist quantities defined in (\ref{eq:twq}). 
They also satisfy (\ref{eq:JKdef}). 
\item $\Lambda$ and $\mH$. $\Lambda$ is the positive cosmological constant 
and $\mH=(\Lambda/3)^{1/2}$. 
\item $\{ e_{\a}\}$ is the orthonormal frame defined in (\ref{eq:ONframe}). 
\item $t_{0}$, $t_{1}$. In the situations of interest in this paper, the metric
(\ref{eq:metric}) is defined on $(t_{0},\infty)\times\tn{3}$, where $t_{0}\geq 0$.
When speaking of a $\tn{2}$-symmetric solution, we take it for granted that 
$t_{0}$ is defined in this way. Moreover, $t_{1}=t_{0}+2$. 
\end{enumerate}

\textbf{Variables for the characteristic system.} 
\begin{enumerate}
\item $\Theta$, $V^{1}$, $V^{2}$, $V^{3}$ are the basic variables of the  
characteristic system (\ref{eq:dThetads})--(\ref{eq:dVthds}). The symbols
$\Theta(s;t,\theta,v)$, $V(s;t,\theta,v)$ denote a solution to the 
characteristic system corresponding to initial data $(t,\theta,v)$. In
other words, $\Theta(s;t,\theta,v)$, $V(s;t,\theta,v)$, considered as 
functions of $s$, are solutions to (\ref{eq:dThetads})--(\ref{eq:dVthds}).
Moreover, $\Theta(t;t,\theta,v)=\theta$, $V(t;t,\theta,v)=v$. 
\item $\Psi$ and $Z^{i}$, $i=1,2,3$. Given a choice of derivative 
($\d_{t}$, $\d_{\theta}$ or $\d_{v^{i}}$), say $\d$, the variables $\Psi$ and 
$Z=(Z^{1},Z^{2},Z^{3})$ are defined by (\ref{eq:Psidef})--(\ref{eq:Zthdef}). 
\item $\hPs$ and $\hZ$ are the rescaled versions of $\Psi$ and $Z$, and
they are defined in (\ref{eq:hZhPsdef}). 
\item $\Psi_{j}$, $Z^{i}_{j}$, $V^{i}_{j}$ and $\Theta_{j}$ are the higher 
order derivatives of $\Psi$, $Z$, $V$ and $\Theta$. They are defined in 
(\ref{eq:hodchsys}). 
\item $\hZ^{1}_{N}$ and $\hPs_{N}$ are the rescaled versions of 
$Z^{1}_{N}$ and $\Psi_{N}$, and they are defined in (\ref{eq:hZoNdef}).
\end{enumerate}

\textbf{Notation, matter quantities.} 
\begin{enumerate}
\item $\rho$, $J_{i}$, $P_{i}$ and $S_{ij}$. The quantities $\rho$, $J_{i}$, 
$P_{i}$ and $S_{ij}$ are defined in general in (\ref{eq:rhodefetc}). 
In the case of Vlasov matter, they are defined in (\ref{eq:rhoPketcvl}). 
\end{enumerate}

\textbf{Notation, Vlasov matter.} 
\begin{enumerate}
\item $\mP$ denotes the mass shell (the set of future directed
unit timelike vectors). 
\item $f$ denotes the distribution function. For $\tn{2}$-symmetric solutions, 
$f$, however, denotes the symmetry reduced version of the distribution 
function. In other words, $f$ is considered to be a function of $t$, $\theta$, $v^{1}$, $v^{2}$
and $v^{3}$, where, if $p$ is an element of the mass shell, $(t,\theta,x,y)$ is the 
base point of $p$, and $v^{\a}$ are the components of $p$ relative to the orthonormal
frame $\{e_{\a}\}$. 
\item $f_{\sca}$ is the rescaled distribution function. It is  
given by $f_{\sca}(t,\theta,v)=f(t,\theta,t^{-1/2}v)$. 
\item $T^{\roV}_{\a\b}$ is the stress energy tensor associated with the Vlasov 
matter. It is given by (\ref{eq:vset}). The components of the stress energy 
tensor with respect to the frame $\{e_{\a}\}$ are given by (\ref{eq:rhodefetc})
and (\ref{eq:rhoPketcvl}); see also (\ref{eq:setvl}). 
\item $L^{i}$. The functions $L^{i}$ are given by 
(\ref{eq:Lodef})--(\ref{eq:Lthdef}).
\end{enumerate}

\textbf{The initial value formulation.} 
\begin{enumerate}
\item $\pros{\Sigma}$ is the projection defined in Remark~\ref{remark:prosdef}.
\item $\rho^{\roV}$ and $\current^{\roV}$ are defined in (\ref{eq:rhoini}) and 
(\ref{eq:jiini}). 
\end{enumerate}

\textbf{Auxiliary notation.} 
\begin{enumerate}
\item $\ldr{h}$. If $h$ is a scalar function, its mean value
is denoted $\ldr{h}$; cf. (\ref{eq:ldrhdef}). 
\item $\ldr{\bp}$. If $\bp$ is a vector, $\ldr{\bp}=(1+|\bp|^{2})^{1/2}$. 
\end{enumerate}

\textbf{Energies controlling the metric variables, the distribution function
and solutions to the characteristic system.} 
\begin{enumerate}
\item The $H^{l}_{\roV,\mu}$-norm is defined in (\ref{eq:fidnorm}). 
\item $\d_{\pm}$ and $\ma_{\pm}$ are defined in (\ref{eq:dpmapm}). They are 
given by 
\[
\d_{\pm}=\d_{t}\pm\a^{1/2}\d_{\theta}, \ \ \
\ma_{\pm}=(\d_{\pm}P)^{2}+
e^{2P}(\d_{\pm}Q)^{2}.
\]
\item $E_{\bas}$ is the $L^{2}$-based energy introduced in (\ref{eq:ebasdef}). 
\item $\mQ^{1}$ controls the size of the support of $f$ in the 
$v^{1}$-direction; cf. (\ref{eq:mQodef}). It is given by 
\[
\mQ^{1}(t):=\sup\{|v^{1}|: (t,\theta,v^{1},v^{2},v^{3})\in \mathrm{supp}f\}.
\]
\item $F$. The sup-norm energy $F$ is introduced in (\ref{eq:Fdef}). It is 
given by 
\[
F(t)=\sup_{\theta\in\so}\ma_{+}(t,\theta)+\sup_{\theta\in\so}\ma_{-}(t,\theta).
\]
\item $\hma_{\pm}$ and $\hF$ are introduced in (\ref{eq:hmadef}) and 
(\ref{eq:hFdef}) respectively. They are given by 
\[
\hma_{\pm}=t^{4}\ma_{\pm}+t,\ \ \
\hF(t)=\sup_{\theta\in \so}\hma_{+}(t,\theta)
+\sup_{\theta\in \so}\hma_{-}(t,\theta).
\]
\item $R^{1}$ and $\hQ^{1}$ are defined in (\ref{eq:RohQodef}). They are given 
by 
\[
R^{1}(s)=[s(V^{1}(s))^{2}+1]^{1/2}, \ \ \
\hQ^{1}(s)=[s(\mQ^{1}(s))^{2}+1]^{1/2}.
\]
\item $\mG$ is defined in (\ref{eq:mGdef}). 
\item $\hE$ is introduced in (\ref{eq:hEdef}). It is given by 
\[
\hE=\sum_{i=1}^{3}(\hZ^{i})^{2}+(\hPs)^{2}.
\]
\item $\ma_{N+1,\pm}$ is introduced in (\ref{eq:manpopm}). It is given by
\[
\ma_{N+1,\pm}=[\d_{\theta}^N P_{t}\pm\d_{\theta}^N(\a^{1/2}P_{\theta})]^{2}
+e^{2P}[\d_{\theta}^N Q_{t}\pm\d_{\theta}^N(\a^{1/2}Q_{\theta})]^{2}.
\]
\item $\hma_{N+1,\pm}$ and $\hF_{N+1}$ are introduced in 
(\ref{eq:hmaNdefhFNdef}). They are given by 
\[
\hma_{N+1,\pm}=t^{7/2}\ma_{N+1,\pm}+t^{1/2},\ \ \
\hF_{N+1}=\sup_{\theta\in\so}\hma_{N+1,+}+\sup_{\theta\in\so}\hma_{N+1,-}.
\]
\item $\hE_{N}$ is defined in (\ref{eq:hENdef}). It is given by 
\[
\hE_{N}=(\hPs_{N})^{2}+(\hZ^{1}_{N})^{2}.
\]
\item $E_{k}$ and $E$. These energies control suitably weighted
Sobolev norms of the distribution function. $E_{k}$ is defined in 
(\ref{eq:ekdef}). Moreover, $E=E_{0}$. 
\end{enumerate}

\end{document}